\documentclass[prx,twocolumn,english,superscriptaddress,floatfix,longbibliography]{revtex4-2}

\usepackage{amsmath}
\usepackage{amsfonts}
\usepackage{amssymb}
\usepackage{amsthm}
\usepackage{bm}
\usepackage{MnSymbol,wasysym}
\usepackage[pdftex]{graphicx}
\usepackage{epstopdf}
\usepackage{etoolbox}
\usepackage{times}
\usepackage[letterpaper,margin=1in]{geometry}
\usepackage[utf8]{inputenc}
\usepackage{color}
\usepackage[usenames,dvipsnames]{xcolor}
\usepackage{epstopdf}
\usepackage{units}
\usepackage{natbib}
\usepackage[normalem]{ulem}
\usepackage{enumerate}
\usepackage[caption=false]{subfig}
\usepackage{longtable}
\usepackage{float}
\usepackage{braket}
\usepackage[colorlinks = true,
            linkcolor = blue,
            urlcolor  = blue,
            citecolor = blue,
            anchorcolor = blue]{hyperref}
\usepackage{silence}
\usepackage[utf8]{inputenc}
\usepackage[caption=false]{subfig}

\usepackage{physics}
\usepackage{mathrsfs} 
\newtheorem{theorem}{Theorem}

\usepackage{comment}

\definecolor{Red}{rgb}{1,0,0}
\newcommand{\xfrac}[2]{{#1}/{#2}}

\usepackage[section]{placeins}
\usepackage{lineno}

\newcommand{\od}{\eta} 
\newcommand{\intall}{\ensuremath{\int_{-\infty}^\infty}}
\newcommand{\sn}[1]{\bar{#1}} %sub-normalized states

\newcommand{\elff}{\ensuremath{\alpha\ini(z,t)}}
\newcommand{\elfb}{\ensuremath{\alpha\fin(z,t)}}

\newcommand{\betf}{\ensuremath{\beta\ini(z,t)}}
\newcommand{\betb}{\ensuremath{\beta\fin(z,t)}}
\newcommand{\elfw}{\ensuremath{\tilde{\alpha}^{}_\rightleftarrows}(z,\omega)}
\newcommand{\betw}{\ensuremath{\tilde{\beta}^{}_\rightleftarrows}(z,\omega)}
\newcommand{\elfwf}{\ensuremath{\tilde{\alpha}\ini}(z,\omega)}
\newcommand{\elfwb}{\ensuremath{\tilde{\alpha}\fin}(z,\omega)}
\newcommand{\betwf}{\ensuremath{\tilde{\beta}\ini}(z,\omega)}
\newcommand{\betwb}{\ensuremath{\tilde{\beta}\fin}(z,\omega)}
\newcommand{\elfwbc}{\ensuremath{\tilde{\alpha}_{\rm in}}(\omega)}
\newcommand{\ini}{_{\rightarrow}}
\newcommand{\fin}{_{\leftarrow}}
\newcommand{\ett}{\ensuremath{\tau_T}} %excitation time - transmitted
\newcommand{\ets}{\ensuremath{\tau_S}} %excitation time - scattered
\newcommand{\etav}{\ensuremath{\tau_{\rm av}}} %excitation time - average
\newcommand{\bv}{L} %boundary variable; medium exists in range [0, \bv]
\newcommand{\gd}{\ensuremath{t_g}} %group delay
\newcommand{\wt}{\ensuremath{t_W}} %Wigner time
\newcommand{\sdt}{\ensuremath{t_S}} %delay time for scattered photons

\renewcommand{\erf}[1]{Eq.~(\ref{#1})}

\newcommand{\erfs}[2]{Eqs.~(\ref{#1})--(\ref{#2})}
\newcommand{\erfa}[2]{Eqs.~(\ref{#1}) and (\ref{#2})}
\newcommand{\arf}[1]{{Appendix}~\ref{#1}} 
\newcommand{\srf}[1]{Sec.~\ref{#1}}

\newcommand{\frf}[1]{Fig.~\ref{#1}}
\newcommand{\ea}{{\it et al.}}

\newcommand{\nn}{\nonumber}

\begin{document}
\title{How much time does a photon spend as an atomic excitation before being transmitted?}

\author{Kyle Thompson}\email[]{kthompson@physics.utoronto.ca}
\affiliation{Department of Physics, and Centre for Quantum Information and Quantum Control,  University of Toronto, 60 St. George Street, Toronto, Ontario, Canada M5S 1A7}
\author{Kehui Li}
\altaffiliation[Present address: ]{Department of Physics, Harvard University, and MIT-Harvard Center for Ultracold Atoms, Cambridge, MA 02138, USA}
\affiliation{Department of Physics, and Centre for Quantum Information and Quantum Control,  University of Toronto, 60 St. George Street, Toronto, Ontario, Canada M5S 1A7}
\author{Daniela Angulo}
\affiliation{Department of Physics, and Centre for Quantum Information and Quantum Control,  University of Toronto, 60 St. George Street, Toronto, Ontario, Canada M5S 1A7}
\author{Vida-Michelle Nixon}
\affiliation{Department of Physics, and Centre for Quantum Information and Quantum Control,  University of Toronto, 60 St. George Street, Toronto, Ontario, Canada M5S 1A7}
\author{Josiah Sinclair}
\affiliation{Department of Physics, MIT-Harvard Center for Ultracold Atoms and Research Laboratory of Electronics,
Massachusetts Institute of Technology, Cambridge, Massachusetts 02139, USA}
\author{Amal Vijayalekshmi Sivakumar}
\altaffiliation[Present address: ]{Department of Physics and Astronomy, Texas A\&M University,
578 University Dr, College Station, TX 77843, USA}
\affiliation{Department of Physical Sciences, Indian Institute of Science Education and Research Kolkata, Mohanpur 741246, India}
\author{Howard M. Wiseman}\email[]{prof.howard.wiseman@gmail.com}
\affiliation{Centre for Quantum Computation and Communication Technology (Australian Research Council), \\ Centre for Quantum Dynamics, Griffith University, Yuggera Country, Brisbane, Queensland 4111, Australia}
\author{Aephraim M.  Steinberg}
\email[]{steinberg@physics.utoronto.ca}
\affiliation{Department of Physics, and Centre for Quantum Information and Quantum Control,  University of Toronto, 60 St. George Street, Toronto, Ontario, Canada M5S 1A7}
\affiliation{Canadian Institute For Advanced Research, 180 Dundas St. W., Toronto, Ontario, Canada, M5G 1Z8}

\date{\today}

\begin{abstract}
When a single photon traverses a cloud of 2-level atoms, the average time it spends as an atomic excitation---as measured by weakly probing the atoms---can be shown to be the spontaneous lifetime of the atoms multiplied by the probability of the photon being scattered into a side mode. 
A tempting inference from this is that an average scattered photon spends one spontaneous lifetime as an atomic excitation, while photons that are transmitted spend zero time as atomic excitations.
However, recent experimental work by some of us [PRX Quantum {\bf 3}, 010314 (2022)] refutes this intuition.  
We examine this problem using the weak-value formalism and show that the time a transmitted photon spends as an atomic excitation is equal to the group delay, which can take on positive or negative values.
We also determine the corresponding time for scattered photons and find that it is equal to the time delay of the scattered photon pulse, which consists of a group delay and a time delay associated with elastic scattering, known as the Wigner time delay. This work provides new insight into the complex and surprising histories of photons travelling through absorptive media.
\end{abstract}

\maketitle

\section{Introduction}\label{sec:intro} 
The propagation of a beam of light  through a cloud of two-level atoms is among the most ubiquitous phenomena in AMO physics \cite{AllenEberlyMonograph}.   Understanding the coherent effects of light on atoms, and the slowing, refraction, absorption and scattering of light by atoms, underlies the entire field. It has been pivotal in the conceptual development of Quantum Optics \cite{loudon2000quantum,scully_zubairy_1997}, from Einstein's early analysis of wave-particle duality up through today's studies of cavity QED, nonclassical states of light, and their applications to quantum information \cite{einstein_1909, Walther_2006, Xanadu_BS2022}.  Although one might reasonably expect such a venerable problem to have no secrets left to reveal, in fact questions as seemingly straightforward as where energy is stored as it propagates through such a dielectric medium have proved controversial \cite{RYC1994,Diener1998}, and even the quantization of light in a dispersive medium was only treated rigorously in the 1990s \cite{Barnett1992,Jeffers1995}.  The problem grows thornier yet in the deep quantum regime where the incident beam is prepared in a single-photon state~\cite{KlausMolmer2020}. 

As a single photon propagates through a cloud of 2-level atoms, part of its energy is temporarily stored in the form of atomic excitations. 
One might therefore be tempted to say that the photon will `spend time inside of' the atoms. 
Given that such a system is quantum-mechanical in nature and that the state of the system will in general be described by an entangled state of photonic and atomic modes, is there any reasonable way that one could ask the naive, classically-inspired question: how much time does the photon spend inside of these atoms? 
In particular, does the answer depend on whether the photon is eventually transmitted through the cloud or scattered into a side mode?

In this work, we answer these questions by using an operational definition of this time, as follows. First, we imagine sending a single photon into the medium and using a weak probe to continuously monitor the probability of finding an atomic excitation anywhere in the medium at any given time.
Integrating this probability over all time provides a measure of the total amount of atomic excitation that occurred while the photon traversed the medium, and has units of time. 
Thus, the integral provides a reasonable---and experimentally measurable---definition of the total time a photon spent inside of the atoms.
This definition falls into the class of `dwell times' that have been discussed in contexts such as quantum tunneling and scattering~\cite{smith_lifetime_1960,hauge_tunneling_1989,ramos_measurement_2020}. 

On average, and in the limit of an arbitrarily weak measurement, the result of such a measurement is shown in \srf{subsec:average_excitation_time} to be $\etav = P_S /\Gamma$, the spontaneous lifetime of the atoms ($1/\Gamma$) multiplied by the probability of the photon being scattered into a side mode.
This expression could be taken to suggest that photons which are scattered spend one atomic lifetime as atomic excitations, while photons that are transmitted through the medium spend zero time as atomic excitations. 

This intuition inspired previous experimental work by some of us, in which Sinclair \ea\ measured the time that a transmitted photon spends as an atomic excitation by using a saturation-based Kerr nonlinearity~\cite{sinclair_measuring_2022}.
This scheme---depicted in \frf{fig:setup}---features a resonant, coherent state `pump' pulse which travels through a cloud of two-level atoms, while a weak, off-resonant CW beam probes the atoms. 
The pump causes a small degree of saturation in the medium, resulting in a cross-phase shift $\phi_{\rm probe}(t)$ being written onto the probe beam which is proportional to the amount of atomic excitation caused by the pump~\footnote{In fact this interaction is not instantaneous due to the finite propagation speed of the probe beam, so the acquired phase shift is proportional to the average excitation probability during the propagation time of the probe, which is short compared to the pulse durations considered in these experiments.}. 
By placing a single-photon detector after the atoms and post-selecting on detection, Sinclair \ea\ were able to measure the cross-phase shift caused by the average transmitted pump photon and compare this to the cross-phase shift caused by the average (not post-selected) incident pump photon. 
Note that although this experiment was performed with coherent state pulses, it has been shown that such coherent state experiments can be used to extract single-photon weak values~\cite{wiseman_obtaining_2023}.

The result of this experiment was that for a broad-band pulse and a high optical depth, the average transmitted photon spent nearly as much time as an atomic excitation as the average incident photon. 
The experiment therefore demonstrated that the simple intuition outlined above---that the average scattered photon spends one atomic lifetime as an atomic excitation, while the average transmitted photon spends zero time---is not generally true. 
In that work, it was speculated that, if transmitted photons were spending a non-zero amount of time as atomic excitations, there must be a mechanism by which photons are absorbed and then preferentially emitted in the forward direction. 
One such mechanism for `coherent forward emission' naturally arises when a broad-band pulse passes through an absorptive medium and the electric field envelope picks up a 180 degree phase flip, coherently inducing excited atoms to re-emit in the forward direction \cite{crisp1970,Costanzo2016}.
While such an explanation is plausible, the semiclassical toy model presented by Sinclair \ea\ based on this idea was incapable of properly modelling the dynamics of a post-selected quantum system. Thus, a quantum treatment is necessary to elucidate the history of transmitted photons.
 
 In this work, we present a quantum theoretical framework to calculate the time that a transmitted photon spends as an atomic excitation (\srf{sec:theory}), which we will refer to as the {\it atomic excitation time} experienced by a transmitted photon.
 The framework is based on the weak-value formalism \cite{aharonov_how_1988,Dressel_review_WV} and quantum trajectory theory \cite{dalibard_wave-function_1992,carmichael_open_1993,Wiseman2002, tsang_time-symmetric_2009, tsang_generalized_2022} and makes the striking prediction that the atomic excitation time experienced by a transmitted photon is equal to the group delay experienced by the photon (\srf{subsec:transmitted_photons}).
 This result may seem reasonable in the far-detuned limit where the group delay is positive  (i.e., the photon wavepacket is delayed compared to free-space propagation), but here we show that this equivalence also holds near resonance where the group delay is negative (`superluminal'). 

 Furthermore, in \srf{subsec:scattered_photons} we calculate the atomic excitation time experienced by scattered photons and find that it is equal to the time delay of the scattered photon pulse, which consists of a group delay plus a time delay associated with elastic scattering, known as the Wigner time delay~\cite{wigner_lower_1955,bourgain_direct_2013}. The three atomic excitation times (not post-selected, post-selected on transmission and post-selected on scattering) are compared and contrasted in various parameter regimes in~\srf{subsec:comparison_excitation_times}. Finally, in \srf{sec:discussion} we present a simple model which illustrates how such negative post-selected dwell times can be explained in terms of quantum-mechanical interference. 

%setup figure 
\begin{figure}[h!]
\includegraphics[width=1\linewidth]
{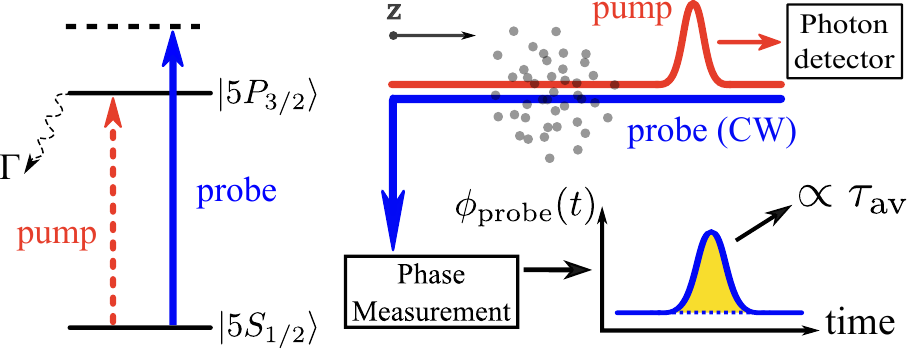}
\caption{Level scheme (left) and conceptual diagram of the phase measurement (right) used in the experiment of Sinclair \ea~\cite{sinclair_measuring_2022}. 
Resonant pump pulses faintly saturate a cloud of cold $^{85}$Rb atoms, causing the off-resonant probe beam to acquire a nonlinear cross-phase shift $\phi_{\rm probe}(t)$ which is proportional to the amount of atomic excitation in the medium at time $t$. 
Integrating this phase shift over time gives a quantity that is proportional to the atomic excitation time: $\etav$ without post-selection, and $\ett$ if one post-selects on detecting the transmitted pump photon.}
\label{fig:setup}
\end{figure}

\section{Theory}\label{sec:theory}
Consider a single photon incident on a medium of two-level atoms initially in the ground state. 
For simplicity, we will take the photon to be in a pure state, and treat the propagation one-dimensionally (along the $z$ axis). 
We also treat the atomic medium as a continuum which extends over a region $[0,\bv]$, such that at every $z\in[0,\bv]$ an atomic excitation is possible. 
Since the total number of excitations is conserved, the single excitation will appear either (i) as an axially propagating photon, (ii) as an atomic excitation or (iii) as a scattered photon. 
Since we wish to consider the case in which the photon is transmitted through the medium (i.e., by considering post-selection on that occurrence), we can remove the last of these three possibilities for the excitation in our description. 
Thus, the state of the system (medium and light) can be written as a sub-normalized (indicated by the overbar) pure state 
\begin{equation}\label{eqn:quantum_state}
\ket{\sn{\psi}\ini(t)} = \intall dz\qty[\elff \ket{p}_z + \betf \ket{e}_z ]  \,.
\end{equation}
\noindent Here $\ket{p}_z$ denotes a $\delta$-normalized state with an axial photon at position $z$ and no excitation elsewhere, and $\ket{e}_z$ denotes a $\delta$-normalized state with an atomic excitation at position $z$ and no excitation elsewhere. 
The complex coefficients $\elff$ and $\betf$ are the probability amplitudes for the excitation to be a photon or an atomic excitation at $(z,t)$, respectively, and the `$\rightarrow$' subscript indicates that the state will evolve forward in time from a normalized initial state $\ket{\psi\ini(-\infty)}$, which describes a photon yet to enter the medium. Note that $\betf$ is defined to be zero outside of $[0,\bv]$.

The interaction between the medium and the light is modelled by the dipole coupling Hamiltonian, \begin{equation}\label{eqn:Hermitian_Hamiltonian}
 \hat{H}_{\rm int} = -\hbar \int_{0}^\bv dz\, g(z)\qty[\ket{p}_z\!\bra{e}+\ket{e}_z\!\bra{p}] \, ,
\end{equation}
where the spatially-dependent coupling constant $g(z)$ takes into account the transverse beam profile and changes in the density of atoms. 
Scattering arises from there being a finite lifetime of the excited atomic state, given by $1/\Gamma$. 
The post-selection on the transmission of the photon can be effected continuously by post-selecting in each time-step on no spontaneous emission. 
According to the theory of quantum jumps~\cite{dalibard_wave-function_1992, carmichael_open_1993,wiseman_quantum_2009,minev_catch_2019}, this can be described by evolving the system in each time step $\delta t$ by a linear but non-unitary map $\hat{M}(\delta t)$ corresponding to the null result from a set of hypothetical perfect detectors looking for scattered photons. In the continuous-time limit, $\hat M(dt) = 1-i\hat{H}_{\rm NH}dt/\hbar$, where
\begin{equation}\label{eqn:non-Hermitian_Hamiltonian}
\hat{H}_{\rm NH}=\hat{H}_0+\hat{H}_{\rm int}-i\frac{\hbar\Gamma}{2}\int_{0}^\bv dz\, \ket{e}_z\!\bra{e} 
\end{equation}
is the non-Hermitian no-jump Hamiltonian, including $\hat{H}_0$, the bare Hamiltonian for the medium and the light. 
This Hamiltonian causes the excited state population amplitude to decay, but with no appearance of a zero-excitation (ground) state. That is, the forward, or pre-selected, state obeys 
\begin{equation}\label{eqn:ket_forward_evolution}
 \ket{\sn{\psi}\ini(t+dt)}=\hat{M}(dt)\ket{\sn{\psi}\ini(t)} \, ,
\end{equation}
\noindent and has a norm given by
\begin{align}\label{eqn:norm}
n\ini(t) & := \ip{\sn{\psi}\ini(t)}{\sn{\psi}\ini(t)} \notag\\ 
& = \intall dz\qty[\abs{\elff}^2+\abs{\betf}^2] \, ,
\end{align}
\noindent which decays monotonically with time. 

Under the non-Hermitian Hamiltonian in \erf{eqn:non-Hermitian_Hamiltonian}, together with the 1D propagation equation for the light and in a frame rotating at the atomic transition frequency, the coefficients in \erf{eqn:quantum_state} will evolve as
\begin{subequations}\label{eqn:howard_forward}
\begin{gather}
\qty[c\pdv{}{z}+\pdv{}{t}]\elff = ig(z)\betf \, , \label{eqn:howard_forward_alpha}\\ 
\pdv{}{t}\betf = ig(z)\elff-\frac{\Gamma}{2}\betf  \, . \label{eqn:howard_forward_beta}
\end{gather}
\end{subequations}  
The initial conditions, in the limit $t\to-\infty$, are
\begin{subequations}\label{eqn:inicond}
\begin{gather}
\elff = \alpha_{\rm in}(t-z/c)
\, , \label{eqn:inicond_alpha}\\
\betf = 0 \, .\label{eqn:inicond_beta}
\end{gather}    
\end{subequations}
Note that $\alpha_{\rm in}$ is defined so that 
at the entry point to the medium ($z=0$), we have  $\alpha\ini(0,t)=\alpha_{\rm in}(t)$.

\erfa{eqn:howard_forward_alpha}{eqn:howard_forward_beta} are equivalent to the Maxwell-Bloch equations for pulse propagation in the limit that each atom remains close to the ground state (i.e., linear propagation), and under the rotating-wave and paraxial approximations~\cite{AllenEberlyMonograph}.
In this correspondence, $ \elff$ and $\betf$ play the roles of the slowly-varying electric field envelope of the pulse and the atomic coherence, respectively. 
The coupling constant $g(z)$ is related to the atom density $N(z)$ and dipole matrix element $d_{eg}$ via 
\begin{align}
    g(z) &= d_{eg}\sqrt{\frac{ckN(z)}{2\epsilon_0\hbar}} \, . 
\end{align} 
It can also be related to the  resonant optical depth $\od_0$ via
\begin{align}
    \od_0 &= \sigma_0\int_{0}^\bv dz\, N(z) = \frac{4}{c\Gamma}\int_{0}^\bv  dz\, g^2(z) \, ,
\end{align}
\noindent where $\sigma_0 = \xfrac{2kd_{eg}^2}{\epsilon_0\hbar\Gamma}$ is the resonant atomic cross section. Note that since $g(z)$ depends on the density of the atoms, the atomic excitations in this work refer to collective excitations---i.e., Dicke-like states in which the excitation is shared amongst a large number of atoms \cite{Dicke_1954}.

The transmission probability $P_T$ is given by 
\begin{align}
    P_T &= n\ini(+\infty) = \ip{\sn{\psi}\ini(+\infty)}{\sn{\psi}\ini(+\infty)} \notag\\
    & = \intall dz |\alpha\ini(z,+\infty)|^2  \, ,
\end{align}
since in the limit of $t\rightarrow\infty$ no atomic excitations remain. Equivalently, if we define a normalized final state (as a bra for convenience), 
\begin{equation} \label{fincon}
    \bra{\psi\fin(+\infty)} = \frac{1}{\sqrt{P_T}} \intall dz\, \alpha^*\ini(z,+\infty)\,{}_z\!\bra{p} \, ,
 \end{equation}
then we have 
\begin{align}
    P_T &=   \qty[\ip{\psi\fin(+\infty)}{\sn{\psi}\ini(+\infty)}]^2 \notag\\
    &=  \qty[\bra{\psi\fin(+\infty)}\cdots \hat{M}(dt)\hat{M}(dt)\cdots\ket{\psi\ini(-\infty)}]^2 \, .
\end{align}
But this last expression shows that we can also write a symmetric-in-time expression, for arbitrary intermediate $t$,
\begin{equation} \label{eqn:post_selection_probability}
    P_T =   \qty[\ip{\sn{\psi}\fin(t)}{\sn{\psi}\ini(t)}]^2 \, ,
\end{equation}
where 
\begin{equation}\label{eqn:quantum_state_backwards}
    \bra{\sn{\psi}\fin(t)} = \intall dz\qty[ \elfb \, {}_z\!\bra{p} + \betb \,{}_z\!\bra{e} ]
\end{equation}
\noindent is a subnormalized state which evolves backwards from \erf{fincon} at $t=+\infty$ via
\begin{align}
    \bra{\sn{\psi}\fin(t-dt)}&= \bra{\sn{\psi}\fin(t)}\hat{M}(dt)\notag\\ 
    &= \bra{\sn{\psi}\fin(t)}\qty[1-i\hat{H}_{\rm NH}dt/\hbar]  \, .
\end{align}
Rewriting this as a forwards-in-time evolution for a ket, we have
\begin{align}
\ket{\sn{\psi}\fin(t+dt)}&=\hat{M}^\dagger(-dt)\ket{\sn{\psi}\fin(t)} \notag\\
    &= \qty[1-i\hat{H}_{\rm NH}^\dagger dt/\hbar]\ket{\sn{\psi}\fin(t)} \, .
\end{align}
Thus, the Hermitian part of the Hamiltonian remains unchanged but the anti-Hermitian part has its sign reversed. This yields the equations of motion for the backwards-evolving state,
\begin{subequations}\label{eqn:howard_backward}
\begin{gather}
\qty[c\pdv{}{z}+\pdv{}{t}]\elfb = ig(z)\betb  \, , \label{eqn:howard_backward_alpha}\\
\pdv{}{t}\betb = ig(z)\elfb +\frac{\Gamma}{2}\betb \, .\label{eqn:howard_backward_beta}
\end{gather}    
\end{subequations}
\noindent Note that the only difference between~\erfa{eqn:howard_forward}{eqn:howard_backward} is the sign of the decay term for $\betb$.

The norm of this state, $n\fin(t)$ (defined as in~\erf{eqn:norm} but with `$\leftarrow$' subscripts), increases monotonically forwards in time. It thus decreases monotonically backwards in time, which is the direction we use for solving this equation, with the final conditions in the limit $t\to+\infty$
\begin{subequations}\label{eqn:finalcond}
\begin{gather}
\elfb = \frac{1}{\sqrt{P_T}} \alpha\ini(z,t) \, , \label{eqn:finalcond_alpha}\\
\betb = 0 \, .\label{eqn:finalcond_beta}
\end{gather}    
\end{subequations}

Consider now the properties of the system at time $t$ as revealed by a weak probe, post-selected on the photon traversing the medium without scattering. 
As has been understood for some decades, such post-selected properties are not determined by a single quantum state, but rather by a pair of quantum states, one arising from the initial preparation and the other from the final measurement \cite{carmichael_open_1993,wiseman_quantum_2009,tsang_time-symmetric_2009,Dressel_review_WV}. 
From the above theory, these two states, at time $t$, are exactly the states $\ket{\sn{\psi}\ini(t)}$ and $\bra{\sn{\psi}\fin(t)}$. 
Specifically, if a probe weakly measures an observable $\hat A$ in a minimally disturbing way, then, from many repetitions of the experiment, and post-selecting on the photon being transmitted (hence the $T$ subscript), the average value of the probe readout, $A_T(t)$, will be found to be the real part of the weak value, 
\begin{equation}\label{eqn:wv}
A_T(t) = \text{Re}\frac{\bra{\sn{\psi}\fin(t)}\hat{A}\ket{\sn{\psi}\ini(t)}}{\ip{\sn{\psi}\fin(t)}{\sn{\psi}\ini(t)}} \, .
\end{equation}
Note that weak values in general can take on values outside of the spectrum of the operator being measured, and can even be complex \cite{aharonov_how_1988,ritchie_realization_1991, Josza_complexWV,ramos_measurement_2020}. The imaginary part of a weak value is generally associated with the backaction of the measurement on the system, while the real part describes the result of the probe measurement \cite{Dressel_imag2012}. Also, note that the subnormality of the states is irrelevant to this value. 

In an experiment like the one described earlier, the measurement is effected by an off-resonant probe beam which picks up a weak, excitation-dependent cross-phase shift due to dispersive coupling (for example, via a saturation-based Kerr nonlinearity), as depicted in \frf{fig:setup}. 
The nonlinear cross-phase shift acquired by the probe as a result of the single-photon pulse traversing the medium is proportional to the total amount of atomic excitation in the medium, which we describe with the operator
\begin{equation}\label{eqn:phase_operator}
    \hat{\phi}_{\rm probe} := \int_{0}^\bv  dz\, \epsilon(z) \ket{e}_z\!\bra{e} \, .
\end{equation}
Here $\epsilon(z)$ determines the strength of the dispersive coupling (in units of radians per unit length). 
Substituting $\hat{\phi}_{\rm probe}$ for $\hat A$ in \erf{eqn:wv} and using the definitions of the forwards and backwards states in~\erfa{eqn:quantum_state}{eqn:quantum_state_backwards}, the average cross-phase shift, post-selected on the photon being transmitted, $\phi_{{\rm probe}|T}(t)$, is given by 
\begin{align}\label{eqn:phiT}
    &\phi_{{\rm probe}|T}(t) = \text{Re}\frac{\bra{\sn{\psi}\fin(t)}\hat{\phi}_{\rm probe}\ket{\sn{\psi}\ini(t)}}{\ip{\sn{\psi}\fin(t)}{\sn{\psi}\ini(t)}} \nonumber\\
    &=\frac{1}{\sqrt{P_T}}\text{Re}\int_{0}^\bv dz \beta\fin^*(z,t)\beta\ini(z,t)\epsilon(z)\,,
\end{align}
where we have used~\erf{eqn:post_selection_probability} to evaluate the denominator. 
We then define the average \textit{atomic excitation time} experienced by a transmitted photon to be the time integral of this phase shift, normalized by the dispersive coupling strength $\epsilon$ (which we now assume to be constant over the whole atomic medium),
\begin{align}
\ett &:= \frac{1}{\epsilon}\intall  dt\, \phi_{{\rm probe}|T}(t) \nn \\
&= \text{Re}\frac{1}{\sqrt{P_T}}\intall   dt\int_{0}^\bv dz\beta^*\fin(z,t)\beta\ini(z,t) \, .
\label{eqn:tauT_weak_value}
\end{align}

If, by contrast, no post-selection is performed, the atomic excitation time experienced by the average photon, $\etav$, will be the time integral of the expectation value of the operator in~\erf{eqn:phase_operator}, normalized by the dispersive coupling strength $\epsilon$:
\begin{align}
\etav &= \frac{1}{\epsilon}\intall dt\expval{\hat{\phi}_{\rm probe}}{\sn{\psi}\ini(t)} \nn \\
& = \intall dt\int_{0}^\bv dz\, \abs{\beta\ini(z,t)}^2 \, .\label{eqn:tau0}
\end{align}
Note that here the subnormality of the state plays a crucial role---the integrals $\int_{0}^\bv dz\, \abs{\beta\ini(z,t)}^2$ and $\intall dz\, \abs{\alpha\ini(z,t)}^2$ are equal to the probability, at time $t$, of there being an atomic excitation, and a propagating photon, respectively, while the ``missing'' probability, $1-\ip{\sn{\psi}\ini(t)}{\sn{\psi}\ini(t)}$, is exactly the probability that the photon has already been scattered.

\section{Results}\label{sec:group_delay_and_excitation_time}
In this section we present the solutions to the above equations of motion (\srf{subsec:solutions}), as well as the predictions of the theory for an average incident photon (\srf{subsec:average_excitation_time}), a transmitted photon (\srf{subsec:transmitted_photons}) and a scattered photon (\srf{subsec:scattered_photons}). In~\srf{subsec:comparison_excitation_times} we compare and contrast these times in various limiting regimes.

\subsection{Solutions to the equations of motion}\label{subsec:solutions}
We begin by presenting the solutions to the equations of motion, ~\erfa{eqn:howard_forward}{eqn:howard_backward}, for a single-photon pulse incident on a 1D medium whose shape is determined by the spatially-dependent coupling constant $g(z)$, where $g(z):= 0$ outside of $[0,\bv]$. 

By transforming to the frequency domain, one can easily solve for $\tilde{\alpha}_{\ini(\fin)}(z,\omega)$ and $\tilde{\beta}_{\ini(\fin)}(z,\omega)$, where the tildes indicate a Fourier transform. 
Recall that~\erfa{eqn:howard_forward}{eqn:howard_backward} are defined in a frame rotating at the atomic transition frequency, so $\omega=0$ corresponds to the atomic resonance.
For the forward-propagating solutions, it is shown in~\arf{sec:app:solution_to_EOMs} that one obtains 
\begin{subequations}\label{eqn:solutions_forward_general}
\begin{gather} \label{eqn:alpha_solution_forward_general}
\elfwf =\elfwbc e^{-i\phi(z,\omega)-\xfrac{\od(z,\omega)}{2} - iz\xfrac{\omega}{c}} \, , \\
 \tilde{\beta}\ini(z,\omega) = \frac{ig(z)}{i\omega +\Gamma/2}\tilde{\alpha}\ini(z,\omega)\, ,\label{eqn:beta_solution_forward_general}
\end{gather}    
\end{subequations}

\noindent where $\elfwbc$ is the Fourier transform of $\alpha_{\rm in}(t):=\alpha_\rightarrow(0,t)$, $\od(z,\omega)$ is the optical depth experienced by the frequency component $\omega$ in the region $[0,z]$, $\phi(z,\omega)$ is the phase acquired by the frequency component $\omega$ of the pulse due to the atoms in the same region and the $\exp[-iz\omega/c]$ phase term simply represents free-space propagation from $0$ to $z$.
The phase and optical depth are related---as expected from the Kramers-Kronig relations \cite{Toll_KK_relation,Hu_KK_relation}---and are given by
\begin{equation}\label{eqn:phase_OD_definition}
    \phi(z,\omega) = -\frac{\omega}{\Gamma}\od(z,\omega) = -\frac{\omega}{\Gamma}\frac{4\mathcal{L}(\omega)}{c\Gamma}\int_0^z  dz^\prime g^2(z^\prime) \, ,
\end{equation}
\noindent where 
\begin{equation}\label{eqn:defLorentzian}
\mathcal{L}(\omega):= \qty[1+(2\omega/\Gamma)^2]^{-1}
\end{equation}
is a Lorentzian function with a full width at half maximum (FWHM) of $\Gamma$.

For the backward-propagating solutions post-selected on the photon being transmitted, it is shown in~\arf{sec:app:solution_to_EOMs} that one obtains
\begin{subequations}\label{eqn:backwards_solutions_general}
    \begin{gather}
    \elfwb = \frac{\elfwf}{\sqrt{P_T}} 
    e^{ \od(z,\omega) - \od(\bv,\omega)} \, ,
    \\
    \tilde{\beta}\fin(z,\omega) = \frac{ig(z)}{i\omega -\Gamma/2}\tilde{\alpha}\fin(z,\omega)\, .
\end{gather}
\end{subequations}

\subsection{Average atomic excitation time}\label{subsec:average_excitation_time}
In this section we show that the atomic excitation time $\etav$ experienced by the average incident photon is equal to the product of the spontaneous lifetime of the atoms and the probability of scattering, $\etav=P_S/\Gamma$. We begin by replacing the time integral in \erf{eqn:tau0} with an integral over frequency to express $\etav$ as
\begin{equation}\label{eqn:tau0_frequency}
    \etav = \intall d\omega\int_0^\bv  dz \abs{\tilde{\beta}\ini(z,\omega)}^2\, .
\end{equation}
Using the solutions from \erf{eqn:solutions_forward_general}, this can be written as
\begin{equation}
    \etav = \frac{c}{\Gamma}\intall d\omega \abs{\elfwbc}^2 \int_0^\bv  dz \pdv{\od(z,\omega)}{z}e^{-\od(z,\omega)} \,,
\end{equation}
\noindent where $\xfrac{\partial\od(z,\omega)}{\partial z}=\xfrac{4g^2(z)\mathcal{L}(\omega)}{c\Gamma}$. Evaluating the spatial integral leads to
\begin{equation}\label{eqn:tau0_final}
    \etav = \frac{1}{\Gamma}\intall d\omega \abs{\elfwbc}^2c\qty[1-e^{-\od_0\mathcal{L}(\omega)}]\, ,
\end{equation}
\noindent where $\od_0\mathcal{L}(\omega)=\od(L,\omega)$ is the optical depth of the entire medium for the frequency component $\omega$.
We now note that the transmission probability $P_T$ can be written as
\begin{align}\label{eqn:P_T_frequency_space}
    P_T& = \intall dz \abs{\alpha\ini(z,+\infty)}^2 = c\intall dt\abs{\alpha\ini(L,t)}^2\nn\\
    &= c\intall d\omega \abs{\elfwbc}^2e^{-\od_0\mathcal{L}(\omega)} = 1-P_S\, ,
\end{align}
\noindent where the second equality is a result of the fact that the pulse is strictly a function of $z-ct$ outside of the medium and we have used the fact that $\abs{\tilde{\alpha}\ini(L,\omega)}^2=\abs{\elfwbc}^2\exp[-\od_0\mathcal{L}(\omega)]$. Given that $c\intall d\omega\abs{\elfwbc}^2=1$ (since the input pulse is normalized), it is then clear that the integral in~\erf{eqn:tau0_final} is equal to the scattering probability $P_S$, and therefore that we have $\etav=P_S/\Gamma$.

\subsection{Transmitted photons}\label{subsec:transmitted_photons}
Here we calculate the atomic excitation time experienced by transmitted photons, $\ett$. We begin by noting that by Parceval's Theorem, 
\begin{align}
    &\intall dt\, \beta^*\fin(z,t)\beta\ini(z,t) \nn 
    \\ & = \intall d\omega\,  \tilde{\beta}^*\fin(z,\omega)\tilde{\beta}\ini(z,\omega) \, .
\end{align}
Using this to replace the time integral in the definition of $\ett$ given in \erf{eqn:tauT_weak_value} with a frequency integral, and substituting in the solutions for $\betwf$ and $\betwb$ given in~\erfa{eqn:solutions_forward_general}{eqn:backwards_solutions_general}, we can express $\ett$ as 
\begin{align}
&\ett =  \text{Re}\frac{1}{\sqrt{P_T}}\int_{0}^{\bv} dz\int_{-\infty}^\infty   d\omega\tilde{\beta}^*\fin(z,\omega)\tilde{\beta}\ini(z,\omega) \nn\\
&= \text{Re}\frac{c\Gamma\od_0}{4P_T}\intall d\omega \frac{\abs{\elfwbc}^2e^{-\od(L,\omega)}}{\omega^2-\Gamma^2/4-i\omega\Gamma} \nn\\
&= \frac{c}{P_T}\intall d\omega \abs{\elfwbc}^2e^{-\od_0\mathcal{L}(\omega)}\frac{\od_0}{\Gamma}\frac{\qty(\frac{2\omega}{\Gamma})^2-1}{\qty[1+\qty(\frac{2\omega}{\Gamma})^2]^2} \, .\label{eqn:tauT_solution}
\end{align}

This time is plotted in \frf{fig:tauT_vs_OD} for resonant Gaussian input pulses of various durations as a function of the resonant optical depth $\od_0$. 
It can be seen that $\ett$ begins at zero for all pulse lengths and immediately takes on negative values, much to our surprise upon performing these calculations. 
The time $\ett$ becomes positive for large values of $\od_0$, and the optical depth at which this occurs increases with increasing pulse duration.
For pulses much longer than the natural lifetime of the excited state (i.e., narrow-band pulses), $\ett$ tends towards $-\od_0$/$\Gamma$. 

\begin{figure}[h]
\includegraphics[width=1\linewidth]{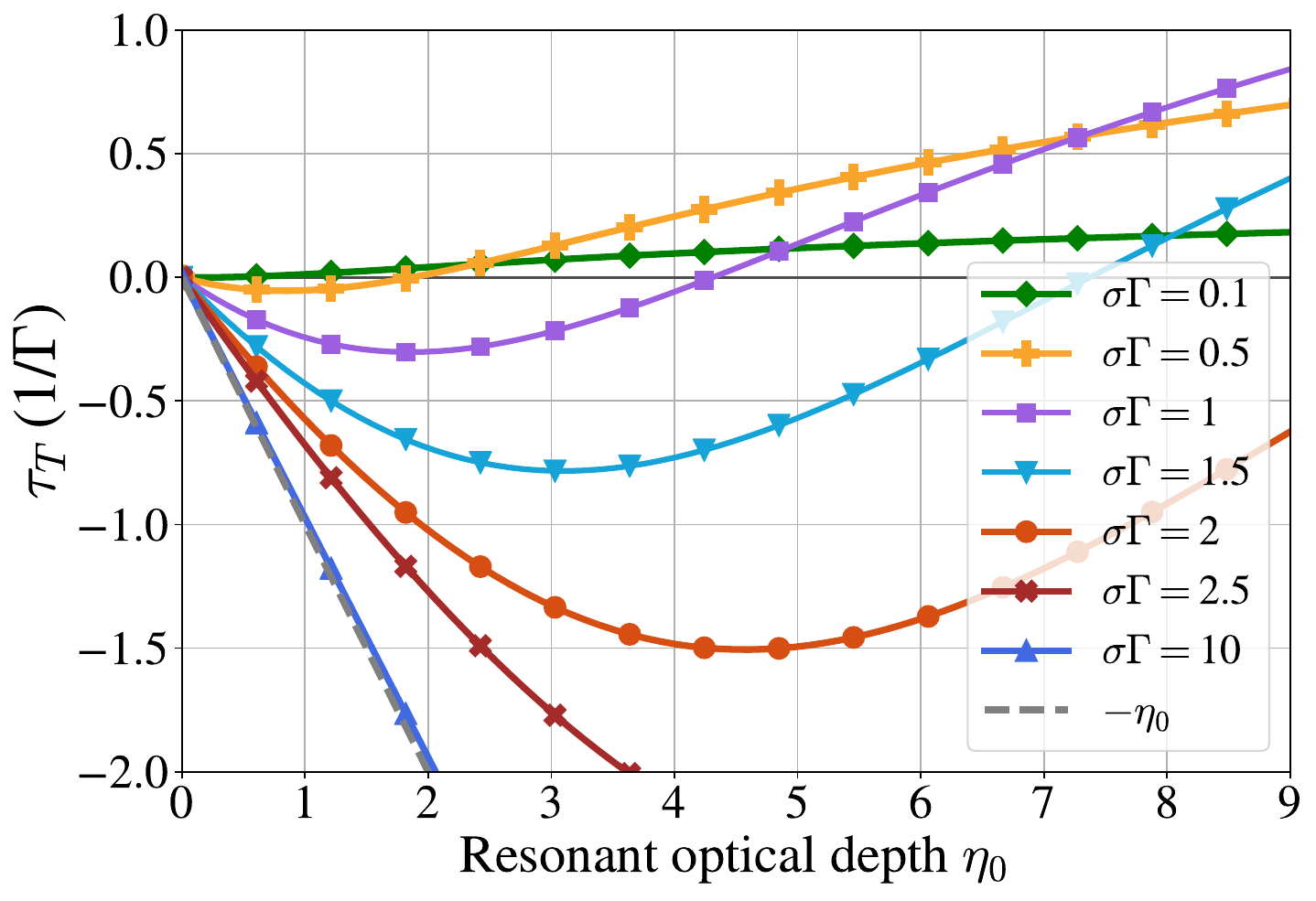}
\caption{Atomic excitation times experienced by transmitted resonant photons ($\ett$) for Gaussian incident photon pulses of different RMS intensity durations $\sigma$ (given in the legend in units of $1/\Gamma$) vs resonant optical depth $\od_0$. 
For pulses that are sufficiently narrow-band compared to the atomic lifetime, $\ett$ tends to -$\od_0/\Gamma$, whereas the time for broad-band pulses becomes positive for large $\od_0$.}
\label{fig:tauT_vs_OD}
\end{figure}

Upon closer inspection of the last line of~\erf{eqn:tauT_solution}, it can be seen that the final two terms in the integrand are equal to the frequency derivative of the spectral phase acquired by the light in the medium, $\xfrac{d\phi(L,\omega)}{d\omega}$. 
This quantity is in fact equal to the narrow-band group delay for the frequency component $\omega$. 
In general, for a narrow-band pulse detuned by $\Delta$ from the atomic resonance traversing a medium with resonant optical depth $\od_0$, the group delay of the pulse compared to free-space propagation is given by \cite{born_wolf}
\begin{equation}\label{eqn:group_delay}
\gd(\Delta,\od_0) := \eval{\dv{\phi(L,\omega)}{\omega}}_{\Delta}  = -\frac{\od_0}{\Gamma}\frac{1-\qty(\frac{2\Delta}{\Gamma})^2}{\qty[1+\qty(\frac{2\Delta}{\Gamma})^2]^2} \, ,
\end{equation}
\noindent where we emphasize the dependence on $\od_0$ for reasons that will be clear in~\srf{subsec:scattered_photons}.

The group delay $\gd(\Delta,\od_0)$ is plotted vs $\Delta$ in \frf{fig:tauT_tauS_vs_detuning}(a) for different values of $\od_0$. Note that close to resonance the group delay can be negative, which can be understood as the result of the preferential transmission of the leading half of the pulse due to the finite response time of the atomic polarizability~\cite{crisp_concept_1971,garrett_propagation_1970}. On resonance, $\gd(0,
\od_0)=-\od_0/\Gamma$.

\begin{figure}[h]
\includegraphics[width=1\linewidth]{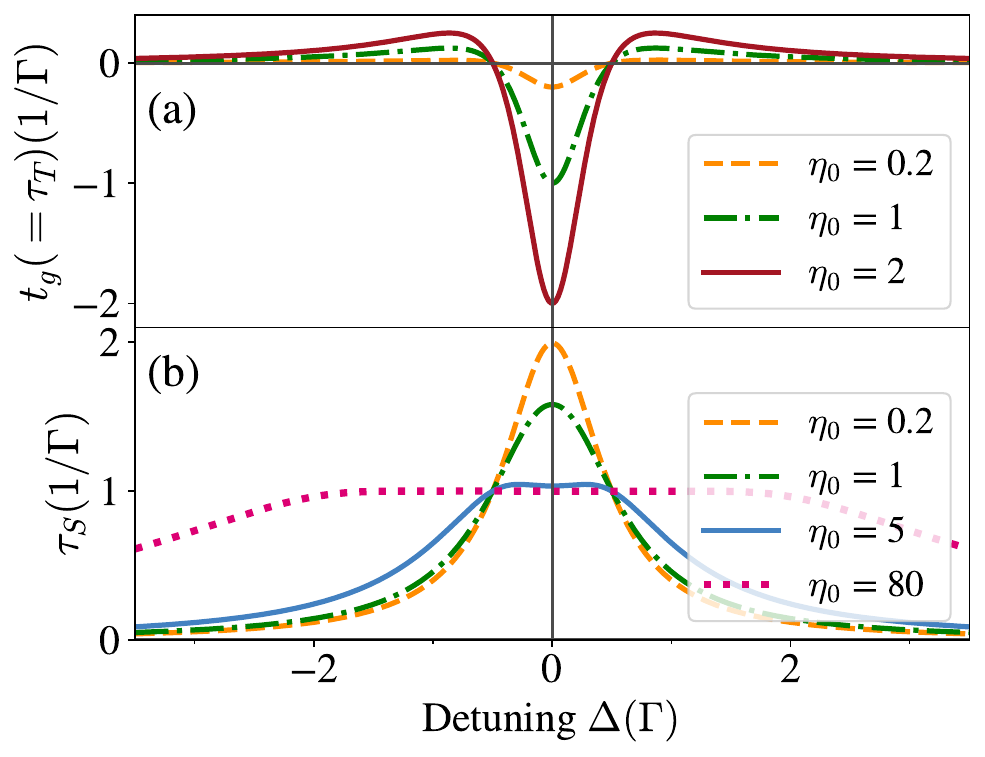}
\caption{(a) Group delay $\gd(\Delta,\od_0)$ $\qty(=\ett(\Delta,\od_0))$ and (b) atomic excitation time experienced by scattered photons $\ets(\Delta,\od_0)$ vs detuning in the narrow-band limit for various values of $\od_0$. 
On resonance, the group delay is given by $-\od_0/\Gamma$. 
For broad-band pulses, the group delay is averaged over the spectrum of the transmitted pulse, leading to positive delays (atomic excitation times) for sufficiently large $\od_0$ due to the preferential transmission of off-resonant frequency components (cf.  \frf{fig:tauT_vs_OD}). 
For $\od_0\rightarrow 0$, $\ets(\Delta,\od_0)$ approaches the Wigner time delay $\wt(\Delta)$. 
For $\od_0\gg1$, the probability of near-resonant photons being scattered approaches unity, causing $\ets(\Delta,\od_0)$ to saturate at $1/\Gamma$ (cf. \erf{eqn:loss_excitation_time}).}
\label{fig:tauT_tauS_vs_detuning}
\end{figure}

We now note that if one uses the third expression for $P_T$ from~\erf{eqn:P_T_frequency_space} in~\erf{eqn:tauT_solution}, it is clear that the expression for $\ett$ given in~\erf{eqn:tauT_solution} is a weighted average of the narrow-band group delay $\gd(\omega,\od_0)$, taken over the frequency spectrum of the transmitted pulse, $c\abs{\elfwbc}^2\exp[-\od_0\mathcal{L}(\omega)]$. 
For incident pulses that are symmetric in time, this weighted average---which we will refer to as the \textit{net group delay} of the transmitted pulse---is in fact equal to the time delay of the `center of mass' of the transmitted pulse due to having traversed the medium (see~\arf{sec:supplementary_broadband_group_delay} for a proof, as well as Refs. \cite{peatross_average_2000, ware_role_2001} for a discussion of how the center-of-mass delay generally consists of a net group delay and a `pulse-reshaping' delay which becomes important for pulses that do not have real Fourier transforms, such as chirped pulses). 
We have therefore shown that the atomic excitation time experienced by a transmitted photon is equal to the net group delay---and hence, for most commonly used pulses, the center-of-mass time delay---experienced by the photon: 
\begin{equation}\label{eqn:tauT_broadband_pulses}
\ett = \frac{c}{P_T}\intall d\omega \abs{\elfwbc}^2 e^{-\od_0\mathcal{L}(\omega)}\gd(\omega,\od_0)\, .
\end{equation}

\subsection{Scattered photons}\label{subsec:scattered_photons}
In addition to calculating the atomic excitation time experienced by transmitted photons, the formalism developed in \srf{sec:theory} can be used to calculate the atomic excitation time experienced by the average {\em scattered} photon. 
An expression for this time is derived in~\arf{sec:app_scattered_photons}, but turns out to be rather cumbersome to work with.
In order to find a simpler expression for $\ets$, we note that since each incident photon is either transmitted through the medium or scattered, we can express the average atomic excitation time $\etav$ as a weighted average of the atomic excitation times in these two cases:
\begin{equation}\label{eqn:sum_rule}
\etav = P_S/\Gamma = P_S\ets + P_T\ett \, .
\end{equation}
\noindent We can therefore express $\ets$ as $\ets=1/\Gamma-P_T\ett/P_S$. We have confirmed numerically that the formula for $\ett$ given in \erf{eqn:tauT_broadband_pulses} and the formula for $\ets$ given in~\erf{eqn:app:tauS_full_expression} of~\arf{sec:app_scattered_photons} obey this weighted average. 
Note that when $\ett>0$, as in the experiment of Sinclair \ea~\cite{sinclair_measuring_2022}, this sum rule implies $\ets < 1/\Gamma$, a fact that the semiclassical picture mentioned in~\srf{sec:intro} attributes to the reduced lifetime of the excited state due to coherent forward emission \cite{Sinclair2021_thesis}.
 
Beginning with the case of an infinitely narrow-band pulse with detuning $\Delta$ from the atomic transition, \erf{eqn:sum_rule} gives
\begin{equation}\label{eqn:loss_excitation_time}
\ets(\Delta,\od_0)= \frac{1}{\Gamma}-\frac{e^{-\od_0\mathcal{L}(\Delta)}}{1-e^{-\od_0\mathcal{L}(\Delta)}}\ett(\Delta,\od_0) \, ,
\end{equation}
\noindent where $\ett(\Delta,\od_0)=\gd(\Delta,\od_0)$ is given by \erf{eqn:group_delay}. 
This expression is plotted as a function of detuning for several values of $\od_0$ in \frf{fig:tauT_tauS_vs_detuning}(b). At low optical depths, $\ets(\Delta,\od_0)$ is simply a Lorenztian function with a FWHM of $\Gamma$ and a peak value of $2/\Gamma$. 
For $\od_0 \gg 1$, frequency components near resonance are very likely to be scattered ($P_S\rightarrow 1$), so \erf{eqn:sum_rule} implies that $\ets(\Delta,\od_0)$ will approach $1/\Gamma$ for near-resonant frequencies, as seen in \frf{fig:tauT_tauS_vs_detuning}(b) for the two highest optical depths.

For a broad-band pulse, $\ets$ can be calculated by finding $\ett$ for the pulse and using \erf{eqn:sum_rule}, together with the expressions for $P_T$ and $P_S$ given in~\erf{eqn:P_T_frequency_space}. 
Equivalently, $\ets$ can be found by averaging $\ets(\omega,\od_0)$ from~\erf{eqn:loss_excitation_time} over the frequency spectrum of the scattered light, as shown in~\arf{sec:scattering_delay_excitation}. 
The atomic excitation time $\ets$ experienced by broad-band photons will be discussed further in~\srf{subsec:comparison_excitation_times}. 

The atomic excitation time experienced by scattered photons turns out to admit a very intuitive interpretation. 
To see this, recall that in~\srf{subsec:transmitted_photons} it was shown that the atomic excitation time experienced by transmitted photons is equal to the net group delay of the transmitted pulse, which, for the case of a symmetric input pulse, is simply equal to the time delay of the center of mass of the pulse.
Might the same principle---that the atomic excitation time is equal to the delay time of the pulse---apply to scattered photons? 
In the case of a scattered photon, there are two time delays that need to be considered: the time delay acquired by propagating through the medium to the point at which the photon is scattered (i.e., the net group delay of the scattered pulse) and the time delay associated with the elastic scattering process itself. 
This scattering time delay, often referred to as the Wigner time delay~\cite{wigner_lower_1955,bourgain_direct_2013}, is equal to the derivative of the scattering phase shift with respect to the frequency of the light.
For a narrow-band pulse with detuning $\Delta$ scattering from a two-level atom, this time delay is given by
\begin{equation}\label{eqn:wigner_time}
\wt(\Delta)  = \frac{2}{\Gamma}\frac{1}{1+(2\Delta/\Gamma)^2} \, .
\end{equation}

For a given resonant optical depth $\od_0$, the time delay of the average scattered photon pulse $\sdt$ can then be calculated by adding the Wigner time delay to the average group delay acquired by scattered photons. 
This is found by averaging the group delay over the optical depths at which a photon can be scattered:
\begin{equation}\label{eqn:scattered_delay_time}
\sdt(\Delta,\od_0) = \wt(\Delta) + \frac{\int_0^{\od_0} d\od_0^\prime e^{-\od_0^\prime\mathcal{L}(\Delta)}\gd(\Delta, \od_0^\prime)}{\int_0^{\od_0} d\od_0^\prime e^{-\od_0^\prime\mathcal{L}(\Delta)}} \, ,
\end{equation}
\noindent where $\exp[-\od_0^\prime\mathcal{L}(\Delta)]$ is the probability of the photon propagating to an optical depth of $\od_0^\prime$ without being scattered. 
It is shown in~\arf{sec:scattered_delay_time_derivation} that for input pulses which are symmetric in time, the expression in~\erf{eqn:scattered_delay_time} (or the corresponding weighted average of this expression in the case of a broad-band pulse) is equal to the time delay of the center of mass of the scattered photon pulse.
Inserting the expression for $\gd(\Delta,\od_0^\prime)$ given in \erf{eqn:group_delay} into \erf{eqn:scattered_delay_time} and evaluating the integrals (see~\arf{sec:scattering_delay_excitation} for details), one finds that the average time delay of a narrow-band scattered photon pulse reduces to the expression for the atomic excitation time experienced by scattered photons given in \erf{eqn:loss_excitation_time}.
That is, $\sdt(\Delta,\od_0) = \ets(\Delta,\od_0)$. Since this equality holds for every frequency, one can take the weighted average of both quantities over the spectrum of the scattered light to show that this holds for pulses of arbitrary bandwidth, as shown in~\arf{sec:scattering_delay_excitation}. 
Therefore, for the case of a time-symmetric input pulse, we have shown that the atomic excitation time experienced by a scattered photon is equal to the time delay of the scattered photon pulse: $\sdt=\ets$.

\subsection{Comparison of atomic excitation times}\label{subsec:comparison_excitation_times}
We now compare and contrast the behaviours of the atomic excitation times $\etav$, $\ett$ and $\ets$ for resonant Gaussian pulses of three different RMS intensity durations $\sigma$, corresponding to narrow-band ($\sigma^{-1} \ll \Gamma$), medium-band ($\sigma^{-1} \approx \Gamma$) and broad-band ($\sigma^{-1} \gg \Gamma$) pulses. 
To facilitate this comparison, all times are plotted in~\frf{fig:sum_rule_effOD} as a function of the effective optical depth $\od_{\rm eff}:= -\ln(P_T)$, since on this axis the average atomic excitation time $\etav=P_S/\Gamma=(1-\exp[-\od_{\rm eff}])/\Gamma$ is independent of the pulse duration. 
Note that for each pulse duration, the three atomic excitation times satisfy the probability sum rule, $\etav=P_S\ets+P_T\ett$. We discuss the behaviour of $\ets$ and $\ett$ in turn below.

Beginning with the atomic excitation time experienced by scattered photons, for an infinitely narrow-band pulse ($\sigma\Gamma=\infty$), $\ets$ begins at the Wigner time delay of $2/\Gamma$.
This can be seen from \erf{eqn:scattered_delay_time}, recalling that $\gd(\Delta,\od_0)$ is proportional to $\od_0$, and therefore that the integral term goes to zero in the limit of $\od_0\rightarrow 0$. 
As $\od_{\rm eff}$ increases, $\ets$  gradually decays to $1/\Gamma$, which is evident from \erf{eqn:loss_excitation_time}. 
For the medium-band case ($\sigma\Gamma\approx 1$), $\ets$ begins at a value between $1/\Gamma$ and $2/\Gamma$ (about 1.5$/\Gamma$ for $\sigma\Gamma=1$) and decays towards $1/\Gamma$ as $\od_{\rm eff} $ increases. 
However, unlike the narrow-band case, here $\ets$ dips below $1/\Gamma$ as soon as $\ett >0$, and subsequently approaches $1/\Gamma$ from below. 
In the broad-band case ($\sigma\Gamma\ll1$), $\ets$ begins at $\sim1/\Gamma$, dips rapidly down to $\sim0.5/\Gamma$ until the near-resonant components of the pulse have been sufficiently attenuated, and subsequently climbs back to $1/\Gamma$ for large $\od_{\rm eff}$.
Approximate expressions for $\ets$ in the limits of $\od_0\ll1$ and $\od_0\gg1$ for broad-band pulses are derived in~\arf{sec:tauS_approximations_BB}, together with a plot of $\ets$ and said approximate expressions for the case of $\od_{\rm eff}\ll1$.

 Turning now to the atomic excitation time experienced by transmitted photons, in the narrow-band case, $\ett$ behaves as $\ett = -\od_{\rm eff}/\Gamma = -\od_0/\Gamma$, as noted in \srf{subsec:transmitted_photons}. 
 In the medium-band case,  $\ett$ starts off negative and becomes positive when $\od_{\rm eff}$ is of order 1 (actually about $2$ for $\sigma\Gamma=1$). 
 This is the point where the near-resonant components of the pulse---which experience a negative group delay---have been sufficiently attenuated. 
 For broad-band pulses ($\sigma\Gamma \ll 1$), $\ett$ grows quadratically with $\od_{\rm eff}$ for $\od_{\rm eff}\ll1$ (not visible in plot), and subsequently increases linearly with $\od_{\rm eff}$ as $\ett \approx 0.5\,\od_{\rm eff}/\Gamma$ after the near-resonant components of the pulse have been sufficiently attenuated. Derivations of these scalings, together with a plot of $\ett$ and these approximate forms in the regime of $\od_{\rm eff}\ll1$, can be found in~\arf{sec:tauT_approximations_BB}.
 
It is therefore clear that the way in which the average atomic excitation time is split between transmitted and scattered photons, as per \erf{eqn:sum_rule}, is far less simple than one would expect from the seemingly intuitive form of the equation $\etav=P_S/\Gamma$.

\begin{figure}[h!]
\includegraphics[width=1\linewidth]{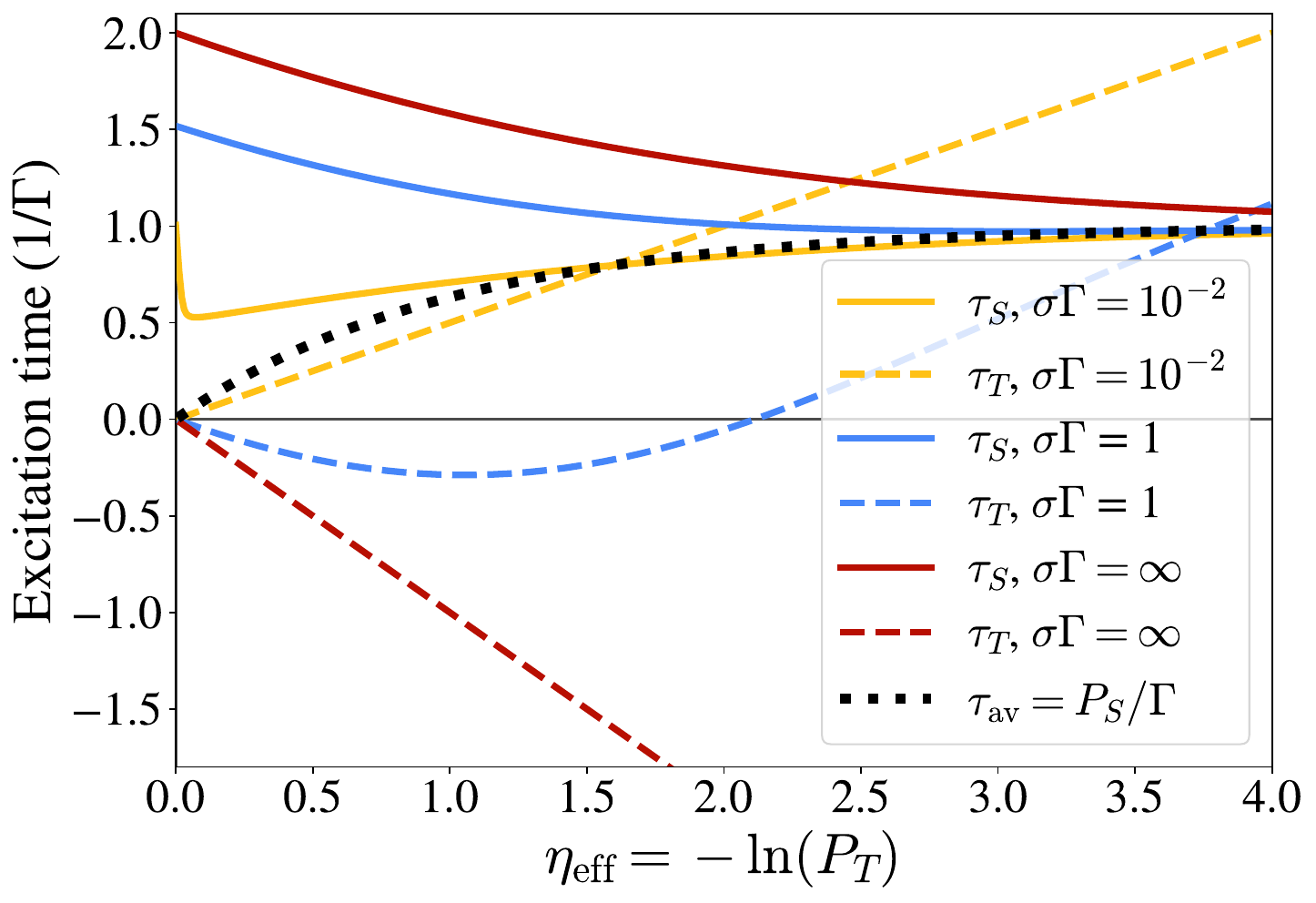}
\caption{Atomic excitation times experienced by transmitted ($\ett$, dashed curves) and scattered ($\ets$, solid curves) photons for resonant Gaussian pulses of three different RMS intensity durations $\sigma$ (indicated by color) vs effective optical depth $\od_{\rm eff}=-\ln (P_T)$. 
The curves for $\ett$ and $\ets$ in the narrow-band limit ($\sigma\Gamma=\infty$) are given by \erfa{eqn:group_delay}{eqn:loss_excitation_time}, respectively. 
For each pulse duration and effective optical depth, $\ett$ and $\ets$ satisfy $\etav = P_S/\Gamma = P_S\ets + P_T\ett$.}
\label{fig:sum_rule_effOD}
\end{figure}

\section{Discussion}\label{sec:discussion}
We have shown that the atomic excitation time experienced by a time-symmetric single-photon pulse propagating through a cloud of 2-level atoms, as measured by a phase shift on a weak probe, is equal to the time delay of the center of mass of the photon pulse, regardless of whether it is transmitted through the cloud or scattered into a side mode.
When the photon is transmitted, this time delay is equal to the net group delay experienced by the transmitted pulse. 
Given that the group delay for narrow-band, resonant light propagating through a sample of two-level atoms is negative, this result poses an interesting interpretational question: what does it mean for a photon to spend a negative amount of time as an atomic excitation? 

In the case of the group delay, negative times can be understood as the result of the incident pulse being reshaped by the atoms. 
In particular, the trailing half of the incident pulse experiences more absorption than the leading half due to the finite response time of the atomic polarizability~\cite{crisp_concept_1971,garrett_propagation_1970}. 
As a result, the center of mass of the transmitted pulse is shifted towards the leading half (i.e., to earlier times), giving the appearance of `superluminal' propagation, although the information velocity is bounded by $c$~\cite{brillouin_wave_1960,stenner_speed_2003}. 
No such interpretation is available for the atomic excitation time $\ett$ experienced by transmitted photons, even though it is equal to the net group delay. 
The atomic excitation time $\ett$ is derived from a post-selected measurement of the probe cross-phase-shift operator $\hat{\phi}_{\rm probe}$ defined in \erf{eqn:phase_operator}, which has an eigenvalue spectrum that is strictly non-negative. 
It is therefore clear that a negative atomic excitation time constitutes an \textit{anomalous weak value} \cite{pusey_anomalous_2014,Hallaji2017,Kwiat2008}. 
Anomalous weak values are generally the result of interference effects~\cite{resch_experimental_2004}, and this is indeed the case here. 

In order to gain insight into the origin of negative atomic excitation times, we will consider a simpler system---which is loosely analogous to a cloud of atoms in the limit of low optical depth---in which quantum interference leads to negative dwell times: a Fabry-Perot cavity. (See \arf{sec:cavity_model}.)
In this analogy, the mode inside the cavity (labelled C) corresponds to the `atomic excitation', as illustrated in \frf{fig:cavity_model}(a).
Transmission through the cavity, labelled D, is chosen to represent scattering/loss in the atomic system, since energy must first propagate through the cavity before being transmitted (just as a photon must excite an atom before being scattered into a side mode).
Reflection from the cavity, labelled B, will therefore correspond to transmission in the atomic system. Note that reflection from the cavity includes fields that reflect from the front mirror without ever entering the cavity as well as fields that make one or more round trips inside the cavity.

The field reflection and transmission coefficients of the cavity are denoted by $r_1$, $t_1$, $r_2$ and $t_2$, where $r_j = \abs{r_j}\approx$ 1 and $t_j=i\abs{t_j}$ for $j=\{1,2\}$. The transmission coefficients $t_j$ are taken to be imaginary to satisfy unitarity, and also to match the atomic system, in which the atomic probability amplitude $\betf$ is 90 degrees out of phase with the incident field on resonance (as per~\erf{eqn:beta_solution_forward_general}). 
Additionally, we let $r_2<r_1$ so that the probability of transmission through the cavity is small, corresponding to the low-optical-depth limit of the atomic system. 
(See~\arf{sec:cavity_corr} for a more detailed explanation of this choice.)

We then consider a resonant, arbitrarily narrow-band single photon incident on the cavity, and a pointer system described by a wavefunction $\Psi(x)$ that weakly monitors the photon number inside of the cavity through a Von Neumann-type coupling \cite{VonNeumann_book}.
Note that $x$ here is a pointer variable that measures photon number, rather than a spatial variable.
Similar to \srf{sec:theory}, one can define a \textit{cavity dwell time} using the form of \erfa{eqn:tauT_weak_value}{eqn:tau0}, but with the operator $\hat{\phi}_{\rm probe}/\epsilon=\int_0^\bv dz \ket{e}_z\!\bra{e}$ replaced by a projector onto the cavity mode, $\ket{{\rm C}}\!\bra{{\rm C}}$.
We then wish to calculate the cavity dwell time post-selected on detecting the photon in B, denoted $\tau_{\rm B}$.

\begin{figure}[h!]
\includegraphics[width=1\linewidth]{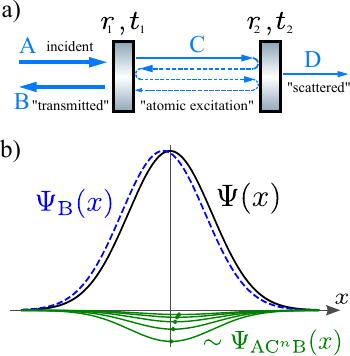}
\caption{a) The Fabry-Perot cavity used to model a cloud of atoms in the limit of low optical depth.
Field reflection and transmission coefficients are $r_j$ and $t_j$.
b) Qualitative illustration of the wavefunctions of the pointer used to measure the photon number in the cavity for several `Feynman paths': reflection off of mirror 1 (black, $\Psi(x)$), and completing $n$ round trips in the cavity followed by transmission into B (green, $\Psi_{{\rm AC}^n{\rm B}}(x)\approx \Psi(x-nx_{\rm rt})$, shown above multiplied by the probability amplitudes $-\abs{t_1}^2r_2(r_1r_2)^{n-1}$ from~\erf{eqn:cavity_psiB} for a handful of $n$ values). The peaks of the functions $\Psi_{{\rm AC}^n{\rm B}}(x)$ are indicated with dots.
The blue dashed curve is the normalized wavefunction of the pointer after post-selecting on detection of the photon in B, $\Psi_{\rm{B}}(x)$.}
\label{fig:cavity_model}
\end{figure}

 Though this weak value can be found by a direct calculation---which we present in \arf{sec:cavity_model}---it is illuminating to evaluate it by considering the evolution of the pointer system. (See~\arf{sec:cavity_corr} for details of the connection between these two approaches.)
 To begin, we note that there are infinitely many `Feynman paths' a photon can take to end up in B: it can be reflected from the front mirror with probability amplitude $r_1$ (path AB), or it can undergo any number $n$ of round trips through the cavity before being transmitted through mirror 1 with probability amplitude $-\abs{t_1}^2r_2(r_1r_2)^{n-1}$ (path AC$^n$B). 
 For path AB, the cavity will remain `empty', and the wavefunction describing the pointer in this case, $\Psi_{\rm{AB}}(x) := \Psi(x)$, will not experience a shift in its mean value. 
 For path AC$^n$B, the pointer will experience a shift $n\, x_{\rm{rt}}$ that will correspond to $n$ round trips through the cavity, and therefore a cavity dwell time of $n$ round trip times, $n\tau_{\rm{rt}}$: $\Psi_{{\rm AC}^n{\rm B}}(x)\approx \Psi(x-nx_{\rm rt})$.
The wavefunction for the pointer post-selected on detecting the photon in B, $\Psi_{\rm B}(x)$, is then a superposition of these wavefunctions, weighted by the probability amplitudes for each path leading to B: 
\begin{align}\label{eqn:cavity_psiB}
\Psi_{\rm B}(x) \propto &r_1\Psi(x) \nn\\
&-\abs{t_1}^2r_2\sum_{n=1}^\infty(r_1r_2)^{n-1}\Psi(x - n{x_{\rm{rt}}} ) \, .
\end{align}

We now consider the limit in which the shift in the pointer variable is small compared to the width of the original pointer distribution so that the measurement constitutes a weak measurement \cite{aharonov_how_1988,Dressel_review_WV}. 
In this limit, one can Taylor-expand the shifted wavefunctions to first order in $nx_{\rm{rt}}$ as $\Psi(x-nx_{\rm rt}) \approx \Psi(x)-nx_{\rm rt}\Psi^\prime(x)$~\footnote{Of course, for any fixed pointer distribution width $\Delta\Psi(x)$ there will be an $n$ beyond which this approximation breaks down; however, the magnitude of the probability amplitudes $-\abs{t_1}^2r_2(r_1r_2)^{n-1}$ decrease exponentially with $n$, meaning one can always choose a width $\Delta\Psi(x)$ that is sufficiently large that the probability amplitudes for $n\gtrapprox\Delta\Psi(x)/x_{\rm rt}$ are negligible}. 
Regrouping terms and summing the infinite series in~\erf{eqn:cavity_psiB} then leads to
\begin{equation}
    \Psi_{\rm B}(x) \propto \frac{r_1-r_2}{1-r_1r_2}\Psi(x) + \frac{r_2\abs{t_1}^2}{\qty(1-r_1r_2)^2}x_{\rm rt}\Psi^\prime(x)\,.
\end{equation}
Dividing through by $(r_1-r_2)/(1-r_1r_2)$ and viewing the result as a first-order Taylor expansion, we can then approximate $\Psi_{\rm B}(x)$ as
\begin{equation}\label{eqn:phiB_shift}
    \Psi_{\rm B}(x) \approx \Psi\qty(x -\qty[-\frac{r_2\abs{t_1}^2}{(1-r_1r_2)(r_1-r_2)}x_{\rm rt}] )\, .
\end{equation}
Since $x_{\rm rt}$ is the shift in the pointer due to a single round trip in the cavity, the quantity multiplying $x_{\rm rt}$ in~\erf{eqn:phiB_shift} is the effective number of round trips in the cavity experienced by a photon that ends up in mode B.
This means that the post-selected cavity dwell time $\tau_{\rm B}$ is simply the effective number of round trips multiplied by the round-trip time $\tau_{\rm rt}$, 
\begin{equation}\label{pointer tau_B}
    \tau_{\rm B} \approx -\tau_{\rm rt}\frac{r_2\abs{t_1}^2}{(1-r_1r_2)(r_1-r_2)} \, .
\end{equation}

We therefore find that although each of the paths AC$^n$B corresponds to a positive pointer shift $n x_{\rm rt}$, interference between all of these paths and the path AB in which the photon spends no time in the cavity leads to a net shift of the pointer in the negative direction, as illustrated qualitatively in~\frf{fig:cavity_model}(b). 
Furthermore, it is shown in~\arf{sec:cavity_corr} that this time can also be written as $\tau_{\rm B}\approx -\od_0^c/\gamma_2$, where $\od^c_0$ is the `optical depth' of the cavity (the negative logarithm of the probability for the photon to end up in D) and $\gamma_2$ is the decay rate of mirror 2. 
This is exactly minus the overall (not post-selected) cavity dwell time, which in the limit of low optical depth is just $\od_0^c/\gamma_2$ (see~\arf{sec:cavity_model}).
The result is perfectly analogous with the atomic case in the narrow-band limit, for which we have $\ett\approx -\od_0/\Gamma$ and $\etav\approx +\od_0/\Gamma$.  Such negative times are a generic feature of post-selection on an outcome which exhibits destructive interference, as in the ``three-box problem'' \cite{aharonov_complete_1991, vaidman_weak-measurement_1996}, a three-path interferometer in which
 the weak value of the projector onto a path given a 180 degree phase shift is -1 \cite{resch_experimental_2004}.

The cavity model thus suggests that the negative atomic excitation time experienced by transmitted photons for the atomic system found in \srf{sec:group_delay_and_excitation_time}---specifically in the limit of low optical depth, where the cavity model is applicable---can be explained in terms of interfering pathways that lead to the photon being transmitted through the medium. 
In this case the interfering pathways would be: 1) the photon passes through the medium without interacting with any atoms (analogous to path AB in the cavity model), and 2) the photon is converted into an atomic excitation for some time and is then coherently transferred back into the photonic mode via Rabi flopping (analogous to paths AC$^n$B).

\section{Conclusion}
Motivated by a recent experiment~\cite{sinclair_measuring_2022}, we have calculated the time that a photon spends as an atomic excitation while propagating through a cloud of 2-level atoms, conditioned on the photon either being transmitted through the cloud in its original spatial mode, or being scattered into a side mode by the atoms.
Our calculation combines weak-value theory and quantum trajectory theory from open quantum systems, and we also provide an intuitive explanation of our results in one regime in terms of interfering paths. 

For the case of a transmitted photon, we find that, in all regimes,  the average amount of time it spends as an atomic excitation is exactly equal to the net group delay experienced by the transmitted photon pulse, which, for time-symmetric pulses, is equal to the time delay of the center of mass of the pulse. 
This may not sound surprising, but it is when one realizes that the group delay experienced by narrow-band, resonant pulses in a two-level medium is negative.  
Thus, we have a fascinating example of an anomalous weak value (negative atomic excitation time) that is found to correspond to a well-known physical quantity (the group delay). 
This suggests the possibility of broader interpretation of the group delay, even when it is negative---for instance, in transverse Fizeau experiments, recent observations of ``anomalous drag''~\cite{banerjee_anomalous_2022, safari_light-drag_2016} could be thought of as excitations drifting at the velocity of a moving medium, but for a negative amount of time.

In the event that a photon is scattered from the medium, we find that the atomic excitation time is equal to the sum of the net group delay experienced by the average scattered photon and the time delay associated with elastic scattering, known as the Wigner time delay. 
For time-symmetric incident pulses, this quantity is equal to the time delay of the center of mass of the scattered photon pulse. 
Thus, regardless of whether it is transmitted or scattered, the time which a time-symmetric single-photon pulse spends as an atomic excitation is equal to the time delay of the center of mass of the pulse.
To the best of our knowledge, this is the first time this equivalence has been reported in the context of absorptive media.

The physically measurable manifestation of a negative atomic excitation time would be a change in the sign of the pointer variable upon post-selecting on transmission of the photon. 
In particular, using the experimental setup described in \srf{sec:intro}, the post-selected phase shift of the probe beam would have a sign opposite to the average (not post-selected) probe phase shift. 
Note that this is different from the common intuition (see, e.g., Ref.~\cite{hau_light_1999}) that nonlinear effects are enhanced by large group delays (slow light), since in those cases it is the phase shift per {\em incident} photon which is being discussed, while our result is for the phase shift per {\em transmitted} photon, which may not even have the same sign.
An experiment is currently underway in order to test the theory presented here, in the regime where it predicts negative atomic excitation times \cite{InPreparation}. 

\begin{acknowledgments}
We would like to thank John Sipe for many helpful discussions, and for his early insight that $\ett$ might just be equal to the group delay (at least when the latter was positive).
This work was supported by NSERC Discovery Grant RGPIN-2020-05767, the Fetzer Franklin Fund of the John E. Fetzer Memorial Trust and the Australian Research Council
Centre of Excellence Program CE170100012: the Centre for Quantum Computation and
Communication Technology (CQC2T). AMS is a fellow of CIFAR.
\end{acknowledgments}

\appendix

\section{Solutions to equations of motion}\label{sec:app:solution_to_EOMs}
Here we solve \erfa{eqn:howard_forward}{eqn:howard_backward} for a single-photon pulse incident on a 1D medium of arbitrary shape defined by the spatially-dependent coupling constant $g(z)$. 

Taking the Fourier transforms of \erfa{eqn:howard_forward}{eqn:howard_backward}, we have, for the forwards- and backwards-evolving coefficients $\elff$, $\elfb$, $\betf$ and $\betb$ (labelled as $\alpha^{}_\rightleftarrows (z,t)$ and $\beta^{}_\rightleftarrows (z,t)$ here for the sake of compact notation):
\begin{gather}
    i\omega\elfw+c\pdv{\elfw}{z}=ig(z)\betw\, ,\\
    i\omega\betw =ig(z)\elfw\mp\frac{\Gamma}{2}\betw \,.
\end{gather}
Rearranging terms, we can write this as
\begin{align}
    \pdv{\elfw}{z} &= -\frac{1}{c}\qty[\frac{g^2(z)}{i\omega \pm \Gamma/2}+i\omega]\elfw \,, \\ \label{eqn:elfw_diff_eqn}
    \betw &= \frac{ig(z)}{i\omega\pm\Gamma/2}\elfw \, .
\end{align}
We let $g(z)$ be any continuous function that is zero outside of the domain $z\in[0,\bv]$. The general solution to \erf{eqn:elfw_diff_eqn} for $z\in[0,\bv]$ is then
\begin{align}
    \elfw =& C_\rightleftarrows(\omega) \nn\\
    & \times \exp\qty[-\int_{0}^z dz^\prime\qty(\frac{g^2(z^\prime)}{ic\omega\pm c\Gamma/2}+i\frac{\omega}{c})] \,,
\end{align}
\noindent where $C_\rightleftarrows(\omega)$ is a function that will be determined using the boundary conditions. For the forward solution, we get
\begin{align}\label{eqn:alpha_solution_forward_unsimplified}
    \elfwf =& \elfwbc \nn\\
    & \times\exp\qty[-\int_{0}^z dz^\prime\qty(\frac{g^2(z^\prime)}{ic\omega+ c\Gamma/2}+i\frac{\omega}{c})] \,,
\end{align}
\noindent where $\elfwbc := \tilde{\alpha}\ini(0,\omega)$. We then note that
\begin{equation}
    \frac{1}{ic\omega\pm c\Gamma/2} = \frac{4}{c\Gamma}\frac{-i\omega/\Gamma \pm 1/2}{1+\qty(2\omega/\Gamma)^2} \,,
\end{equation}
\noindent so that the first term in the exponential of \erf{eqn:alpha_solution_forward_unsimplified} can be written as
\begin{align}
    &-\int_{0}^z dz^\prime\frac{g^2(z^\prime)}{ic\omega+ c\Gamma/2} \nn\\
    &\quad= \frac{4}{c\Gamma}\int_{0}^z dz^\prime g^2(z^\prime)\qty[\frac{i\omega/\Gamma - 1/2}{1+\qty(2\omega/\Gamma)^2}] \nn\\ 
    &\quad= -i\phi(z,\omega)-\od(z,\omega)/2\,,
\end{align}
\noindent where $\phi(z,\omega)$ and $\od(z,\omega)$ are defined in \erf{eqn:phase_OD_definition}.
The forward solution is then
\begin{align}\label{eqn:app:solutions_forward_general}
    \elfwf=& \elfwbc \nn\\
    &\times\exp\qty[-i\phi(z,\omega)-\frac{\od(z,\omega)}{2} - i\frac{\omega}{c}z] \, , \\
    & \tilde{\beta}\ini(z,\omega) = \frac{ig(z)}{i\omega +\Gamma/2}\tilde{\alpha}\ini(z,\omega) \,.
\end{align}
The boundary condition for post-selecting on transmission is given by
\begin{equation}
\tilde{\alpha}\fin(\bv,\omega) = \frac{\tilde{\alpha}\ini(\bv,\omega)}{\sqrt{P_T}} \,,
\end{equation}
\noindent which leads to a backwards solution of
\begin{align}
    \elfwb &=\frac{\elfwbc}{\sqrt{P_T}}\exp\qty[-i\phi(z,\omega) +\frac{\od(z,\omega)}{2}] \nn \\
    &\quad \times \exp \qty[- \od(L,\omega) - \frac{i\omega z}{c}] \\
    &= \frac{\elfwf}{\sqrt{P_T}} 
    e^{ \od(z,\omega) - \od(\bv,\omega)} \,, \\ 
    \tilde{\beta}\fin(z,\omega) &= \frac{ig(z)}{i\omega-\Gamma/2}\tilde{\alpha}\fin(z,\omega)\,. \label{eqn:app:backwards_solutions_general}
\end{align}

\section{Equivalence of the net group delay and center-of-mass time delay}\label{sec:supplementary_broadband_group_delay}
Here we show that for incident pulses that are symmetric in time, the net group delay (i.e., \erf{eqn:tauT_broadband_pulses}) is equal to the time delay of the center of mass of the transmitted photon pulse compared to free-space propagation. To begin, we prove the following theorem:

\begin{theorem}\label{theorem:COM_group_delay}
    Let $g_i(t)$ be a function with $\intall dt\abs{g_i(t)}^2t = 0$, and let $g_f(t)$ be defined by its Fourier transform: $\tilde{g}_f(\omega)=\tilde{g}_i(\omega)\exp\qty(-ih(\omega))$, where $h(\omega)$ is a real function and we use the following Fourier transform convention:
    \begin{equation*}
        \mathscr{F}\qty[g(t)]:= \tilde{g}(\omega) = \frac{1}{\sqrt{2\pi}}\intall dt\, g(t)e^{-i\omega t}\, .
    \end{equation*}
    Then the following is true:
    \begin{equation*}
        \intall dt \abs{g_f(t)}^2t=\intall d\omega \abs{\tilde{g}_i(\omega)}^2\dv{h(\omega)}{\omega} \, .
    \end{equation*}
\end{theorem}

\begin{proof}
    First, we make use of some basic properties of the Fourier transform. For any functions $f(t)$ and $s(t)$, we have
    \begin{align}
        \intall dt f^*(t)s(t) &= \intall d\omega \tilde{f}^*(\omega)\tilde{s}(\omega) \, ,\\
        \mathscr{F}\qty[f(t)t] &= i\dv{\tilde{f}(\omega)}{\omega} \, ,
    \end{align}
    \noindent where the star superscript indicates complex conjugation. Combining these properties, we have
\begin{align}\label{eqn:supplementary_theorem1_hammer}
        \intall dt\abs{f(t)}^2t &= \intall dt f^*(t)f(t)t \nn\\
        &= \intall d\omega \tilde{f}^*(\omega)\mathscr{F}\qty[f(t)t]\,, \nn \\
        &= i\intall d\omega \tilde{f}^*(\omega)\dv{\tilde{f}(\omega)}{\omega} \:.
    \end{align}
    We now apply~\erf{eqn:supplementary_theorem1_hammer} to $g_f(t)$ twice as follows:
    \begin{align}
        &\intall dt\abs{g_f(t)}^2t \nn\\
        &=\intall d\omega \tilde{g}_f^*(\omega)\qty[i\dv{\tilde{g}_i(\omega)}{\omega}e^{-ih(\omega)}+\dv{h(\omega)}{\omega}\tilde{g}_f(\omega)] \nn \\
        &= i \intall \hspace{-1mm}d\omega \tilde{g}_i^*(\omega)\dv{\tilde{g}_i(\omega)}{\omega}+\hspace{-0.5mm}\intall \hspace{-1mm}d\omega\abs{\tilde{g}_f(\omega)}^2\dv{h(\omega)}{\omega} \nn \\
        &= \intall dt\abs{g_i(t)}^2t + \intall d\omega\abs{\tilde{g}_i(\omega)}^2\dv{h(\omega)}{\omega}\, ,
    \end{align}
    where we have used the fact that $\abs{g_f(t)}^2=\abs{g_i(t)}^2$. Since $\intall dt\abs{g_i(t)}^2t=0$, we then have
    \begin{equation}
        \intall dt\abs{g_f(t)}^2t =  \intall d\omega\abs{\tilde{g}_i(\omega)}^2\dv{h(\omega)}{\omega} \:, 
    \end{equation}
    as needed.
\end{proof}

We now make use of Theorem~\ref{theorem:COM_group_delay} with $\tilde{g}_i(\omega):=\elfwbc\exp[-\od_0\mathcal{L}(\omega)/2]$ and $h(\omega):=\phi(L,\omega)+\omega L/c$, so that $\tilde{g}_f(\omega)=\tilde{\alpha}_\rightarrow(L,\omega)$. 
Recall that $\alpha_{\rm in}(t)$ is symmetric in time by hypothesis, so that $\elfwbc$ is real. Also, ${\cal L}(\omega)$ of Eq.~(\ref{eqn:defLorentzian}) is real, so $\tilde{g}_i(\omega)$ is also real and therefore $g_i(t)$ is symmetric in time and satisfies $\intall dt\abs{g_i(t)}^2t=0$.
The theorem then gives
\begin{align}
    &\intall dt\abs{\alpha_\rightarrow(L,t)}^2t \qquad \quad \quad \, \,\nn\\
    &= \intall \hspace{-1mm}d\omega\abs{\elfwbc}^2 e^{-\od_0\mathcal{L}(\omega)}\qty[\dv{\phi(L,\omega)}{\omega}+\frac{L}{c}]\, .
\end{align}
Using the expressions for $P_T$ given in~\erf{eqn:P_T_frequency_space}, the time delay of the center of mass of the pulse compared to free-space propagation is then
\begin{align}
    &\frac{\intall dt\abs{\alpha_\rightarrow(L,t)}^2t}{\intall dt\abs{\alpha_\rightarrow(L,t)}^2}-\frac{L}{c} \nn\\
    &= \frac{c}{P_T}\intall dt\abs{\alpha_\rightarrow(L,t)}^2t-\frac{L}{c} \nn\\
    &= \frac{c}{P_T}\intall d\omega\abs{\elfwbc}^2 e^{-\od_0\mathcal{L}(\omega)}\gd(\omega)\, ,
\end{align}
\noindent where $\gd(\omega)=\dv*{\phi(L,\omega)}{\omega}$ is the narrow-band group delay for the frequency component $\omega$.

\section{Atomic excitation time experienced by scattered photons}\label{sec:app_scattered_photons}
In this section we will derive an expression for the atomic excitation time experienced by scattered photons, $\ets$. 
To start, we assume that we have a 4$\pi$ array of imaging detectors for the scattered photon. 
Let us assume that a scattered photon is detected in a time $[T,T+dt)$ from position $[Z,Z+dz)$. 
After this detection event, the cross-phase shift on the probe beam is zero, and the post-selected state of the system at time $T$ is $\propto \ket{e}_Z$. 
That is, considering a retrodicted state, we have the final conditions $\alpha_{\leftarrow|Z,T}(z,T):= 0$, $\beta_{\leftarrow|Z,T}(z,T)\propto \delta(z-Z)$, where the $Z,T$ subscripts specify the position of the scattering event and the time of the detection event.
For specificity, we take $\beta_{\leftarrow|Z,T}(z,T)=\delta(z-Z)$. Note that this choice leaves the final state unnormalized and causes $\beta_{\leftarrow|Z,T}(z,T)$ to have units of length$^{-1}$ instead of length$^{-1/2}$; however, this choice is valid since the final state appears in both the numerator and denominator of the weak value formula.

These final conditions are evolved backwards in the same way as before, using \erf{eqn:howard_backward}, to give $\alpha_{\leftarrow|Z,T}(z,t)$ and $\beta_{\leftarrow|Z,T}(z,t)$ for $t\leq T$. 
Using the weak value formula from \erf{eqn:wv} in the main text with the post-selected state described above, the ensemble average of the atomic excitation time experienced by this scattered photon, denoted $\tau_{S|Z,T}$, is given by
\begin{widetext}
\begin{equation}
    \tau_{S|Z,T} = \text{Re} \int_{-\infty}^T dt \frac{\int_0^\bv dz \betf\beta^*_{\leftarrow|Z,T}(z,t)}{\int_0^\bv dz^\prime\qty[\beta\ini(z^\prime,t)\beta^*_{\leftarrow|Z,T}(z^\prime,t)+\alpha\ini(z^\prime,t)\alpha^*_{\leftarrow|Z,T}(z^\prime,t)]} \,.
\end{equation}
\end{widetext}
As in the derivation of $\ett$ (\erf{eqn:tauT_weak_value}) in the main text, the denominator in this equation is independent of $t$. 
For simplicity, we evaluate it at $t=T$, which gives
\begin{equation}
    \tau_{S|Z,T} = \text{Re} \int_{-\infty}^T dt \frac{\int_0^\bv dz \betf\beta^*_{\leftarrow|Z,T}(z,t)}{\beta\ini(Z,T)} \,.
\end{equation}

Now, the probability of no detection up to time $T$ and then a detection in an interval of duration $dT$ about time $T$ and a position interval of size $dZ$ about $Z$, given that that a photon is scattered, is $\Gamma\abs{\beta\ini(Z,T)}^2dTdZ/P_S$.
Averaging over the detection time and scattering location of this photon, the ensemble average of the atomic excitation time experienced by a single scattered photon is therefore given by
\begin{align}\label{eqn:app:tauS_full_expression}
    \ets &= \frac{1}{P_S}\intall dT \;\Gamma \int_0^\bv dZ \abs{\beta\ini(Z,T)}^2 \tau_{S|Z,T}\nn\\
    &=\frac{1}{P_S}\text{Re}\intall dT \;\Gamma \int_0^\bv dZ \beta^*\ini(Z,T)\nn\\
    &\quad \times\int_{-\infty}^T dt\int_0^\bv dz \betf\beta^*_{\leftarrow|Z,T}(z,t) \,.
\end{align}
\noindent We have confirmed numerically that the above expression and the expressions for $\ett$ and $\etav$ given in ~\erfa{eqn:tauT_solution}{eqn:tau0} of the main text satisfy the probability sum rule, $\etav = P_S\ets+P_T\ett$.

\section{Derivation of~\erf{eqn:scattered_delay_time} (scattering time delay)}\label{sec:scattered_delay_time_derivation}
In this section we will prove that  for incident pulses which are symmetric in time, the time delay of the center of mass of the average scattered photon pulse is equal to a weighted average of the expression for $\sdt(\omega,\od_0)$ given in~\erf{eqn:scattered_delay_time}, taken over the spectrum of the scattered light.  

Given that scattering occurs from the excited state at a constant rate $\Gamma$, the time profile of the average scattered photon pulse---as measured, for example, by summing the signals from an array of 4$\pi$ imaging detectors surrounding the atoms---will simply follow the excited state population, $\int_0^\bv dz\abs{\betf}^2$. 
With $t=0$ defined as the time at which the center of mass of the incident photon pulse arrives at $z=0$, the time delay $\sdt$ of the center of mass of the scattered pulse will be given by
\begin{align}\label{eqn:supp_COM_scattering_time}
    \sdt &= \frac{\intall dt\int_0^\bv dz \abs{\betf}^2t}{\intall dt\int_0^\bv dz \abs{\betf}^2} \nn\\
    &= \frac{\Gamma}{P_S}\intall dt\int_0^\bv dz \abs{\betf}^2t \, ,
\end{align}
\noindent where we have used the fact that the integral in the denominator of the first line is equal to $\etav=P_S/\Gamma$, as shown in \srf{subsec:average_excitation_time}.
Our goal is to show that $\sdt$ is equal to the weighted average of the expression for $\sdt(\omega,\od_0)$ given in~\erf{eqn:scattered_delay_time}, weighted by the spectrum of the scattered light, $c\abs{\elfwbc}^2\qty[1-\exp(-\od_0\mathcal{L}(\omega))]$.

To begin, consider the solution for $\betwf$ from~\erf{eqn:app:solutions_forward_general}, 
\begin{align}
    \betwf &= \frac{ig(z)}{i\omega+\Gamma/2}\elfwbc \nn \\
    &\quad \times \exp\qty[-i\phi(z,\omega)-i\frac{\omega}{c}z-\frac{\od(z,\omega)}{2}]. \nn
\end{align}
We write this in a form for which 
Theorem~\ref{theorem:COM_group_delay} may be applied, 
by using the fact that $\elfwbc$ is real 
(since we are considering time-symmetric input pulses), 
\begin{align}
    \betwf &= \frac{g(z)\elfwbc e^{-\od(z,\omega)/2}}{\sqrt{\omega^2+(\Gamma/2)^2}} \nn \\
    & \quad \times \exp\qty[-i\qty(\phi(z,\omega)+\frac{\omega}{c}z)] \nn \\
    & \quad \times \exp\qty[-i \qty(\tan^{-1}\qty(\frac{2\omega}{\Gamma}) -\frac{\pi}{2})] \, .
\end{align}
\noindent 
Since we are interested in the delay of the scattered pulse compared to free-space propagation, we will ignore the $\exp(-i\omega z/c)$ term, as this term is simply the time delay associated with propagating to the average $z$ position at which a photon is scattered. 
We will therefore make use of Theorem~\ref{theorem:COM_group_delay} with
\begin{align}
    g_f(t)&:= \betf \, , \nn\\
    \tilde{g}_i(\omega) &:=  \frac{g(z)\elfwbc e^{-\od(z,\omega)/2}}{\sqrt{\omega^2+(\Gamma/2)^2}} \, , \nn\\
    h(\omega)&:= \phi(z,\omega)+\tan^{-1}\qty(\frac{2\omega}{\Gamma})-\frac{\pi}{2} \, .
\end{align}
Applying Theorem~\ref{theorem:COM_group_delay} then gives
\begin{align}\label{eqn:supp_scattered_theorem}
    &\frac{\Gamma}{P_S}\intall dt\int_0^\bv dz \abs{\betf}^2t \nn\\ 
    &=\frac{\Gamma}{P_S}\int_0^\bv dz\intall d\omega\abs{\betwf}^2\nn \\
    &\quad\times\dv{}{\omega}\qty[\phi(z,\omega)+\tan^{-1}\qty(\frac{2\omega}{\Gamma})] \, .
\end{align}
\noindent The second term on the right hand side of ~\erf{eqn:supp_scattered_theorem} is in fact equal to the weighted average of the Wigner time delay,
\begin{align}\label{eqn:supp_scattering_delay_second_term}
    &\frac{\Gamma}{P_S}\hspace{-1mm}\int_0^\bv \hspace{-1mm}dz\hspace{-1mm}\intall \hspace{-1mm}d\omega\abs{\betwf}^2\dv{}{\omega}\qty[\tan^{-1}\qty(\frac{2\omega}{\Gamma})] \nn\\
    &= \frac{\Gamma}{P_S}\hspace{-1mm}\intall \hspace{-1mm}d\omega\, \wt(\omega)\hspace{-1mm}\int_0^\bv \hspace{-1mm}dz\abs{\betwf}^2 \nn\\
    &= \frac{c}{P_S}\hspace{-1mm}\intall \hspace{-1mm}d\omega\, \abs{\elfwbc}^2\qty[1-e^{-\od_0\mathcal{L}(\omega)}]\wt(\omega) \, ,
\end{align}
\noindent where we have used the expression for $\int_0^\bv dz\abs{\betwf}^2$ derived in~\srf{subsec:average_excitation_time} (cf.~\erfa{eqn:tau0_frequency}{eqn:tau0_final}). 
The first term on the right hand side of ~\erf{eqn:supp_scattered_theorem} is
\begin{widetext}
\begin{align}\label{eqn:supp_scattering_delay_first_term}
    &\frac{\Gamma}{P_S}\int_0^\bv dz\intall d\omega\abs{\betwf}^2\dv{}{\omega}\qty[\phi(z,\omega)] \nn\\
    &= \frac{1}{P_S}\intall d\omega \frac{\abs{\elfwbc}^2}{\omega^2+(\Gamma/2)^2}\frac{(2\omega/\Gamma)^2-1}{\qty[1+(2\omega/\Gamma)^2]^2}\int_0^\bv dz g^2(z)e^{-\od(z,\omega)}\od(z,\omega) \nn \\
    &=\frac{c}{\Gamma P_S}\intall d\omega\abs{\elfwbc}^2\frac{(2\omega/\Gamma)^2-1}{1+(2\omega/\Gamma)^2}\int_0^\bv dz \pdv{\od(z,\omega)}{z}e^{-\od(z,\omega)}\od(z,\omega) \nn \\
    & = \frac{c}{\Gamma P_S}\intall d\omega\abs{\elfwbc}^2\frac{(2\omega/\Gamma)^2-1}{1+(2\omega/\Gamma)^2}\qty[1-(1+\od_0\mathcal{L}(\omega))e^{-\od_0\mathcal{L}(\omega)}]\nn\\
    &= \frac{c}{P_S}\intall d\omega\abs{\elfwbc}^2\qty[1-e^{-\od_0\mathcal{L}(\omega)}]\frac{1}{\Gamma}\frac{1-(2\omega/\Gamma)^2}{1+(2\omega/\Gamma)^2}\qty[\frac{\od_0\mathcal{L}(\omega)e^{-\od_0\mathcal{L}(\omega)}}{1-e^{-\od_0\mathcal{L}(\omega)}}-1] \, ,
\end{align}
\end{widetext}
\noindent where we have used the fact that $\partial_z\od(z,\omega) = 4g^2(z)\mathcal{L}(\omega)/c\Gamma$. Combining~\erfa{eqn:supp_scattering_delay_second_term}{eqn:supp_scattering_delay_first_term}, and comparing with the expression for $\sdt(\omega,\od_0)$ in the first line of~\erf{eqn:supp_tS_tauS_narrowband}, we have
\begin{equation}
    \sdt = \frac{c}{P_S}\intall d\omega\abs{\elfwbc}^2\qty[1-e^{-\od_0\mathcal{L}(\omega)}] \sdt(\omega,\od_0), 
\end{equation}
\noindent which is a weighted average of $\sdt(\omega,\od_0)$, taken over the spectrum of the scattered light. Note that for an arbitrarily narrow-band pulse with detuning $\Delta$, we have $c\abs{\elfwbc}^2=\delta(\omega-\Delta)$, in which case $\sdt=\sdt(\Delta,\od_0)$, where $\sdt(\Delta,\od_0)$ is as defined in~\erf{eqn:scattered_delay_time}.

\section{Equivalence of the time delay of the scattered pulse and $\ets$}\label{sec:scattering_delay_excitation}
Here we prove that for incident pulses which are symmetric in time, the time delay of the center of mass of the average scattered photon pulse is equal to the atomic excitation time experienced by the average scattered photon, $\ets$, as given by the sum rule $\ets=1/\Gamma-P_T\ett/P_S$. We begin with the case of an infinitely narrow-band photon with detuning $\Delta$.
From \erf{eqn:scattered_delay_time}, we have
\begin{align}\label{eqn:scattering_delay_time}
    \sdt(\Delta,\od_0) = \wt(\Delta) + \frac{\int_0^{\od_0} d\od_0^\prime e^{-\od_0^\prime\mathcal{L}(\Delta)}\gd(\Delta, \od_0^\prime)}{\int_0^{\od_0} d\od_0^\prime e^{-\od_0^\prime\mathcal{L}(\Delta)}} \,,
\end{align}
\noindent where $\mathcal{L}(\bullet)$ 
is the Lorentzian function defined in \erf{eqn:defLorentzian}. 
We will focus on the second term, which is simply a weighted average of the group delay where the weighting function is the probability of the photon propagating to each optical depth $\od_0^\prime$ without being scattered. The group delay in the above expression can be written as
\begin{align}\label{eqn:group_delay_Lorentzian}
    \gd(\Delta, \od_0^\prime) = -\frac{\od_0^\prime}{\Gamma}\mathcal{L}^2(\Delta)\qty[1-\qty(2\Delta/\Gamma)^2] \,.
\end{align}
\noindent The numerator in the second term on the right hand side of~\erf{eqn:scattering_delay_time} is therefore
\begin{align}
    &\int_0^{\od_0} d\od_0^\prime e^{-\od_0^\prime\mathcal{L}(\Delta)}\gd(\Delta, \od_0^\prime) \nn\\
    &= -\frac{\mathcal{L}^2(\Delta)}{\Gamma}\qty[1-\qty(\frac{2\Delta}{\Gamma})^2] \int_0^{\od_0} d\od_0^\prime e^{-\od_0^\prime\mathcal{L}(\Delta)}\od_0^\prime \nn\\
    & = \frac{\mathcal{L}(\Delta)}{\Gamma}\qty[1-\frac{4\Delta^2}{\Gamma^2}]\qty[\od_0e^{-\od_0\mathcal{L}(\Delta)} - \frac{1-e^{-\od_0\mathcal{L}(\Delta)}}{\mathcal{L}(\Delta)}] \nn\\
    & = \frac{1}{\Gamma}\qty[1-\frac{4\Delta^2}{\Gamma^2}] \qty[\od_0\mathcal{L}(\Delta)e^{-\od_0\mathcal{L}(\Delta)} + e^{-\od_0\mathcal{L}(\Delta)} -1] .
\end{align}
\noindent Meanwhile, the normalization term is simply
\begin{align}
    \int_0^{\od_0} d\od_0^\prime e^{-\od_0^\prime\mathcal{L}(\Delta)} = \frac{1}{\mathcal{L}(\Delta)}\qty[1-e^{-\od_0\mathcal{L}(\Delta)}] \,.
\end{align}
\noindent Combining all of the above, the scattering time delay for infinitely narrow-band light at detuning $\Delta$ is 
\begin{align}\label{eqn:supp_tS_tauS_narrowband}
    \sdt(\Delta,\od_0) &= \wt(\Delta) + \frac{1}{\Gamma}\mathcal{L}(\Delta)\qty[1-\qty(2\Delta/\Gamma)^2] \nn\\
    &\quad \times\qty[\frac{\od_0\mathcal{L}(\Delta)e^{-\od_0\mathcal{L}(\Delta)}}{1-e^{-\od_0\mathcal{L}(\Delta)}}-1] \nn\\
    &= \frac{2}{\Gamma}\mathcal{L}(\Delta)-\frac{1}{\Gamma}\mathcal{L}(\Delta)\qty[2-\frac{1}{\mathcal{L}(\Delta)}] \nn\\
    &\quad- \gd(\Delta,\od_0)\frac{e^{-\od_0\mathcal{L}(\Delta)}}{1-e^{-\od_0\mathcal{L}(\Delta)}} \nn\\
    &= \frac{1}{\Gamma} - \frac{e^{-\od_0\mathcal{L}(\Delta)}}{1-e^{-\od_0\mathcal{L}(\Delta)}}\ett(\Delta,\od_0) \nn\\
    &= \ets(\Delta,\od_0) \,,
\end{align}
\noindent where we have used the fact that $\wt(\Delta) = (2/\Gamma)\mathcal{L}(\Delta)$, the fact that $\gd(\Delta,\od_0) = \ett(\Delta,\od_0)$ and \erf{eqn:group_delay_Lorentzian}. 
For a broad-band pulse, the narrow-band expression for $\sdt(\Delta,\od_0)$ can be averaged over the spectrum of the scattered light to obtain the time delay of the center of mass of the scattered pulse, as shown in~\arf{sec:scattered_delay_time_derivation}. 
Using the second last equality in~\erf{eqn:supp_tS_tauS_narrowband}, we have
\begin{widetext}
\begin{align}
    \sdt &= \frac{c\intall d\omega\abs{\elfwbc}^2\qty[1-e^{-\od_0\mathcal{L}(\omega)}]\sdt(\omega,\od_0)}{c\intall d\omega\abs{\elfwbc}^2\qty[1-e^{-\od_0\mathcal{L}(\omega)}]} \nn \\
    &= \frac{c\intall d\omega\abs{\elfwbc}^2\qty[\qty(1-e^{-\od_0\mathcal{L}(\omega)})/\Gamma-e^{-\od_0\mathcal{L}(\omega)}\ett(\omega,\od_0)]}{c\intall d\omega\abs{\elfwbc}^2\qty[1-e^{-\od_0\mathcal{L}(\omega)}]} = \frac{1}{\Gamma}-\frac{P_T}{P_S}\ett = \ets \, ,
\end{align}
\end{widetext}
\noindent where we have used the definition of the net group delay from \erf{eqn:tauT_broadband_pulses} and the expressions for $P_T$ and $P_S$ from~\erf{eqn:P_T_frequency_space}. Hence, for any input pulse that is symmetric in time, the time delay of the center of mass of the average scattered photon pulse, $\sdt$, is equal to the atomic excitation time experienced by the average scattered photon, $\ets$.

\section{Limiting behaviour of $\ets$ for broad-band pulses}\label{sec:tauS_approximations_BB}
Here we derive approximate expressions for $\ets$ for broad-band pulses in the limits of $\od_0\ll 1$ and $\od_0\gg1$. First, we simplify $\ets(\omega,\od_0)$ from~\erf{eqn:supp_tS_tauS_narrowband} in the limit of $\od_0\ll 1$ as follows:
\begin{align}
    \ets(\omega,\od_0) &= \frac{1}{\Gamma} - \frac{e^{-\od_0\mathcal{L}(\omega)}}{1-e^{-\od_0\mathcal{L}(\omega)}}\ett(\omega,\od_0) \nn\\
    &\approx \frac{1}{\Gamma} -\qty[\frac{1}{\od_0\mathcal{L}(\omega)}-\frac{1}{2}]\nn\\
    &\quad \times \qty[-\frac{\od_0\mathcal{L}(\omega)}{\Gamma}\frac{1-(2\omega/\Gamma)^2}{1+(2\omega/\Gamma)^2}]\nn \\
    &=\frac{1}{\Gamma}\qty[1+\frac{1-(2\omega/\Gamma)^2}{1+(2\omega/\Gamma)^2}]+\frac{1}{2}\gd(\omega,\od_0) \nn \\
    &= \wt(\omega)+\frac{1}{2}\gd(\omega,\od_0) \,,
\end{align}
\noindent where the approximation was a Taylor expansion to first order in $\od_0\mathcal{L}(\omega):=\od_0\qty[1+(2\omega/\Gamma)^2]^{-1}$. 
We then average this time over the spectrum of the scattered light, 
\begin{equation}\label{eqn:appendix_tauS_approx_1}
    \ets \approx \frac{\intall d\omega c\abs{\elfwbc}^2\qty[1-e^{-\od_0\mathcal{L}(\omega)}]\ets(\omega,\od_0)}{\intall d\omega c\abs{\elfwbc}^2\qty[1-e^{-\od_0\mathcal{L}(\omega)}]} \,,
\end{equation}
\noindent where $c\abs{\elfwbc}^2$ is the normalized spectrum of the input pulse. 
Since we wish to investigate the behaviour of arbitrarily broad-band pulses, we replace $c\abs{\elfwbc}^2$ in~\erf{eqn:appendix_tauS_approx_1} with 1, since the spectrum of the pulse will be approximately flat over the region in which the integrands are appreciably non-zero:
\begin{equation}
    \ets \approx \frac{\intall d\omega \qty[1-e^{-\od_0\mathcal{L}(\omega)}]\ets(\omega,\od_0)}{\intall d\omega \qty[1-e^{-\od_0\mathcal{L}(\omega)}]} \,.
\end{equation}
It is straightforward to show that Taylor expanding the numerator and denominator up to second order in $\od_0$ and performing the frequency integrals gives
\begin{align}\label{eqn:appendix_tauS_approx_BB_lowOD}
    \ets &\approx \frac{\frac{\pi}{2}\od_0-\frac{\pi}{4}\od_0^2+O(\od_0^3)}{\frac{\pi\Gamma}{2}\od_0-\frac{\pi\Gamma}{8}\od_0^2+O(\od_0^3)} \nn\\
    &=\frac{1}{\Gamma}\frac{1-\od_0/2+O(\od_0^2)}{1-\qty(\od_0/4+O(\od_0^2))} \nn\\
    &\approx \frac{1}{\Gamma}\qty[1-\od_0/2+O(\od_0^2)]\qty[1+\od_0/4+O(\od_0^2)] \nn\\
    &\approx\frac{1}{\Gamma} - \frac{\od_0}{4\Gamma} +O(\od_0^2) \, . 
\end{align}
Furthermore, for broad-band Gaussian incident pulses of RMS duration $\sigma \ll 1/\Gamma$, 
 and in the limit of $\od_0\ll 1$, it is straightforward to expand $P_T$ as 
\begin{equation}\label{eqn:P_T_approx_BB_lowOD}
    P_T\approx 1-\sqrt{\frac{\pi}{2}}\sigma\Gamma\od_0\,,    
\end{equation}
meaning that the effective optical depth $\od_{\rm eff}$ is given by
\begin{equation}
    \od_{\rm eff}=-\ln(P_T)\approx -\ln\qty(1-\sqrt{\frac{\pi}{2}}\sigma\Gamma\od_0)\approx \sqrt{\frac{\pi}{2}}\sigma\Gamma\od_0 \,.
\end{equation}
Plugging this into~\erf{eqn:appendix_tauS_approx_BB_lowOD} then gives $\ets$ as a function of $\od_{\rm eff}$,
\begin{equation}\label{eqn:tauS_approx_BB_lowOD}
    \ets \approx \frac{1}{\Gamma}-\sqrt{\frac{2}{\pi}}\frac{\od_{\rm eff}}{4\sigma\Gamma^2}\,.
\end{equation}

To determine the behavior of $\ets$ for $\od_0\gg1$, we note that the approximate form of $\ett$ in this regime (which is derived in \arf{sec:tauT_approximations_BB}) is simply $\ett\approx 0.5\od_{\rm eff}/\Gamma$.
Using this expression together with the probability sum rule, $\etav=P_S/\Gamma=P_S\ets+P_T\ett$, we can then approximate $\ets$ as
\begin{equation}\label{eqn:tauS_approx_BB_highOD}
    \ets\approx \frac{1}{\Gamma}-\frac{e^{-\od_{\rm eff}}}{1-e^{-\od_{\rm eff}}}\frac{\od_{\rm eff}}{2\Gamma}\,,
\end{equation}
where $P_T=1-P_S=\exp(-\od_{\rm eff})$.

The atomic excitation time $\ets$ experienced by the average scattered photon is plotted in~\frf{fig:tauS_approximations_BB} together with the approximations from \erfa{eqn:tauS_approx_BB_lowOD}{eqn:tauS_approx_BB_highOD}.

\begin{figure}[h!]
\includegraphics[width=1\linewidth]{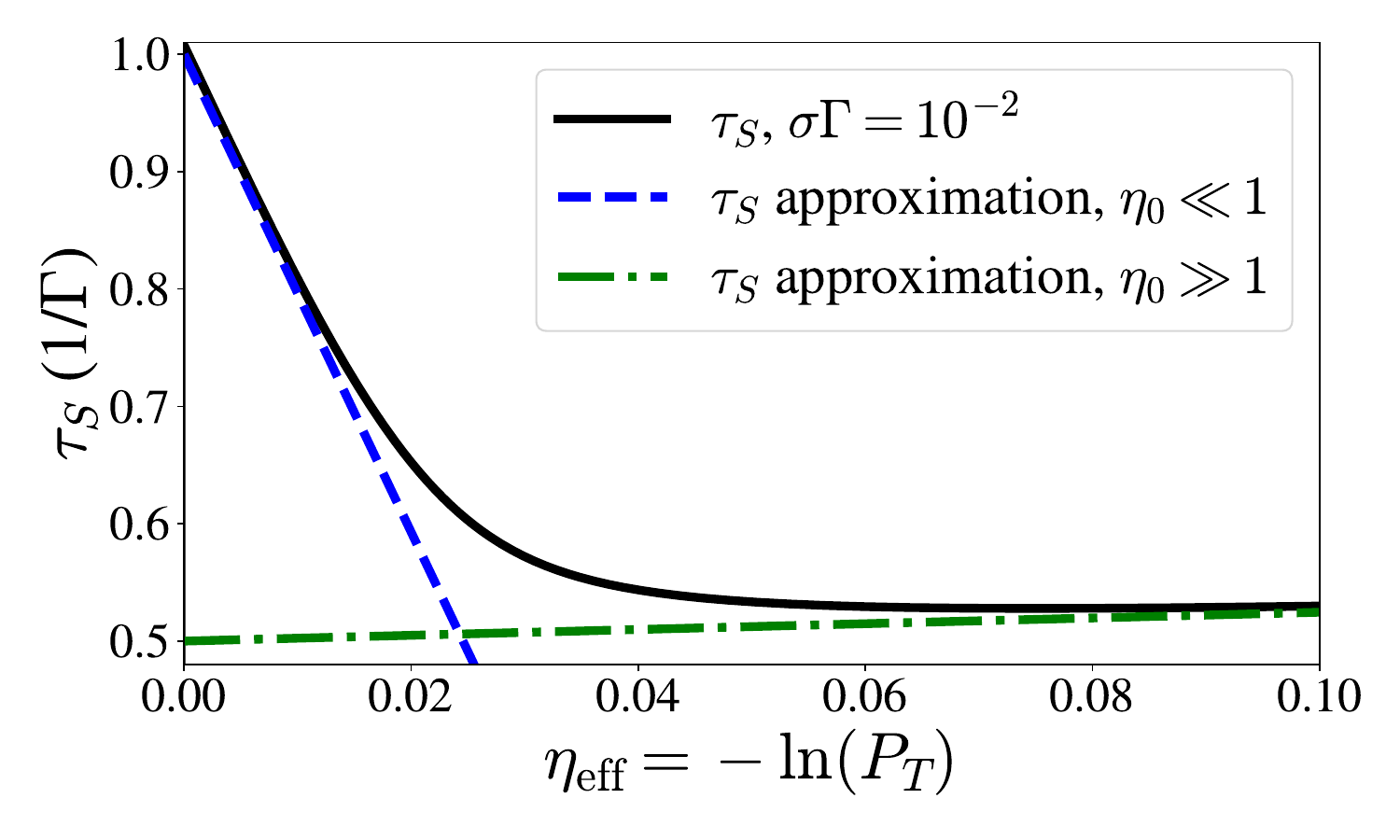}
\caption{Atomic excitation time $\ets$ experienced by scattered photons for a broad-band pulse vs $\od_{\rm eff}$, together with approximate forms in the limits of $\od_0\ll1$ (blue dashed, given by~\erf{eqn:tauS_approx_BB_lowOD}) and $\od_0\gg1$ (green dot-dashed, given by~\erf{eqn:tauS_approx_BB_highOD}).}
\label{fig:tauS_approximations_BB}
\end{figure}

\section{Limiting behavior of $\ett$ for broad-band pulses}\label{sec:tauT_approximations_BB}
Here we will derive an approximate form of $\ett$ for a broad-band Gaussian pulse of RMS duration $\sigma$ in the limits of $\od_0\ll1$ and $\od_0\gg1$. Beginning in the limit of $\od_0\ll1$, we have, making use of~\erf{eqn:appendix_tauS_approx_BB_lowOD},
\begin{equation}
    \ett=\frac{P_S}{P_T}\qty[\frac{1}{\Gamma}-\ets]\approx \frac{P_S}{P_T}\frac{\od_0}{4\Gamma}\,.
\end{equation}
We then note that for $\od_0\ll1$,
\begin{equation}
    \frac{P_S}{P_T}\approx P_S\approx \sqrt{\frac{\pi}{2}}\sigma\Gamma\od_0 \,,
\end{equation}
where we have made use of~\erf{eqn:P_T_approx_BB_lowOD}. Hence, we have
\begin{equation}\label{eqn:tauT_approx_BB_lowOD}
    \ett\approx \sqrt{\frac{\pi}{2}}\frac{\sigma}{4}\od_0^2 \,.
\end{equation}

We now turn to the limit of $\od_0\gg1$. The full expression for $\ett$ is
\begin{equation}\label{eqn:tauT_gaussian_full}
    \ett = \frac{1}{P_T}\intall d\omega\,c\abs{\elfwbc}^2e^{-\od_0\mathcal{L}(\omega)}\gd(\omega,\od_0)\,,
\end{equation}
where $c\abs{\elfwbc}^2$ is the normalized spectrum of the input pulse, $\mathcal{L}(\omega)= [1+(2\omega/\Gamma)^2]^{-1}$, $\gd(\omega,\od_0)$ is the narrow-band group delay (as defined in \erf{eqn:group_delay}) and $P_T$ is given by~\erf{eqn:P_T_frequency_space}.
We will assume a Gaussian input pulse of intensity duration $\sigma\ll1/\Gamma$, so that the input spectrum is given by
\begin{equation}
    c\abs{\elfwbc}^2 = \sqrt{\frac{2}{\pi}}\sigma e^{-2\sigma^2\omega^2}\,.
\end{equation}
Now, for $\od_0\gg1$, we can make the following approximation:
\begin{equation}
    \exp[-\frac{\od_0}{1+\qty(\frac{2\omega}{\Gamma})^2}] \approx \exp[-\od_0\qty(\frac{\Gamma}{2\omega})^2]\,.
\end{equation}
This approximation is valid here because the left hand side is approximately zero near resonance and only becomes non-negligible when $\od_0\mathcal{L}(\omega)$ is of order unity, at which point $\qty(2\omega/\Gamma)^2\gg1$, so $1+\qty(2\omega/\Gamma)^2\approx \qty(2\omega/\Gamma)^2$. With this approximation, $P_T$ becomes
\begin{equation}\label{eqn:P_T_approx_highOD}
    P_T\approx \sqrt{\frac{2}{\pi}}\sigma\intall d\omega e^{-2\sigma^2\omega^2}e^{-\od_0\frac{\Gamma^2}{4\omega^2}}=e^{-\sigma\Gamma\sqrt{2\od_0}}\,,
\end{equation}
from which it is clear that $\od_{\rm eff}\approx \sigma\Gamma\sqrt{2\od_0}$.
Similarly, we can replace $\gd(\omega,\od_0)$ in~\erf{eqn:tauT_gaussian_full} with an approximate form that is valid for $2\omega/\Gamma \gg1$,
\begin{equation}
    \gd(\omega,\od_0)=-\frac{\od_0}{\Gamma}\frac{1-\qty(\frac{2\omega}{\Gamma})^2}{\qty[1+\qty(\frac{2\omega}{\Gamma})^2]^2} \approx \frac{\od_0}{\Gamma}\qty(\frac{\Gamma}{2\omega})^2\,.
\end{equation}
We then have
\begin{align}
    \ett &\approx \frac{1}{P_T}\sqrt{\frac{2}{\pi}}\sigma \intall d\omega e^{-2\sigma^2\omega^2}e^{-\od_0\frac{\Gamma^2}{4\omega^2}}\,\frac{\Gamma\od_0}{4\omega^2}\nn \\
    & = \frac{1}{P_T}\sigma\sqrt{\frac{\od_0}{2}}e^{-\sigma\Gamma\sqrt{2\od_0}} \approx \sigma\sqrt{\frac{\od_0}{2}} \,,
\end{align}
where the last approximation makes use of \erf{eqn:P_T_approx_highOD}. We can then express this in terms of $\od_{\rm eff}\approx \sigma\Gamma\sqrt{2\od_0}$ as
\begin{equation}\label{eqn:tauT_approx_BB_highOD}
    \ett \approx \frac{\od_{\rm eff}}{2\Gamma}\,.
\end{equation}
The atomic excitation time $\ett$ experienced by transmitted photons is plotted in~\frf{fig:tauT_approximations_BB} together with the approximations from \erfa{eqn:tauT_approx_BB_lowOD}{eqn:tauT_approx_BB_highOD}.

\begin{figure}[h!]
\includegraphics[width=1\linewidth]{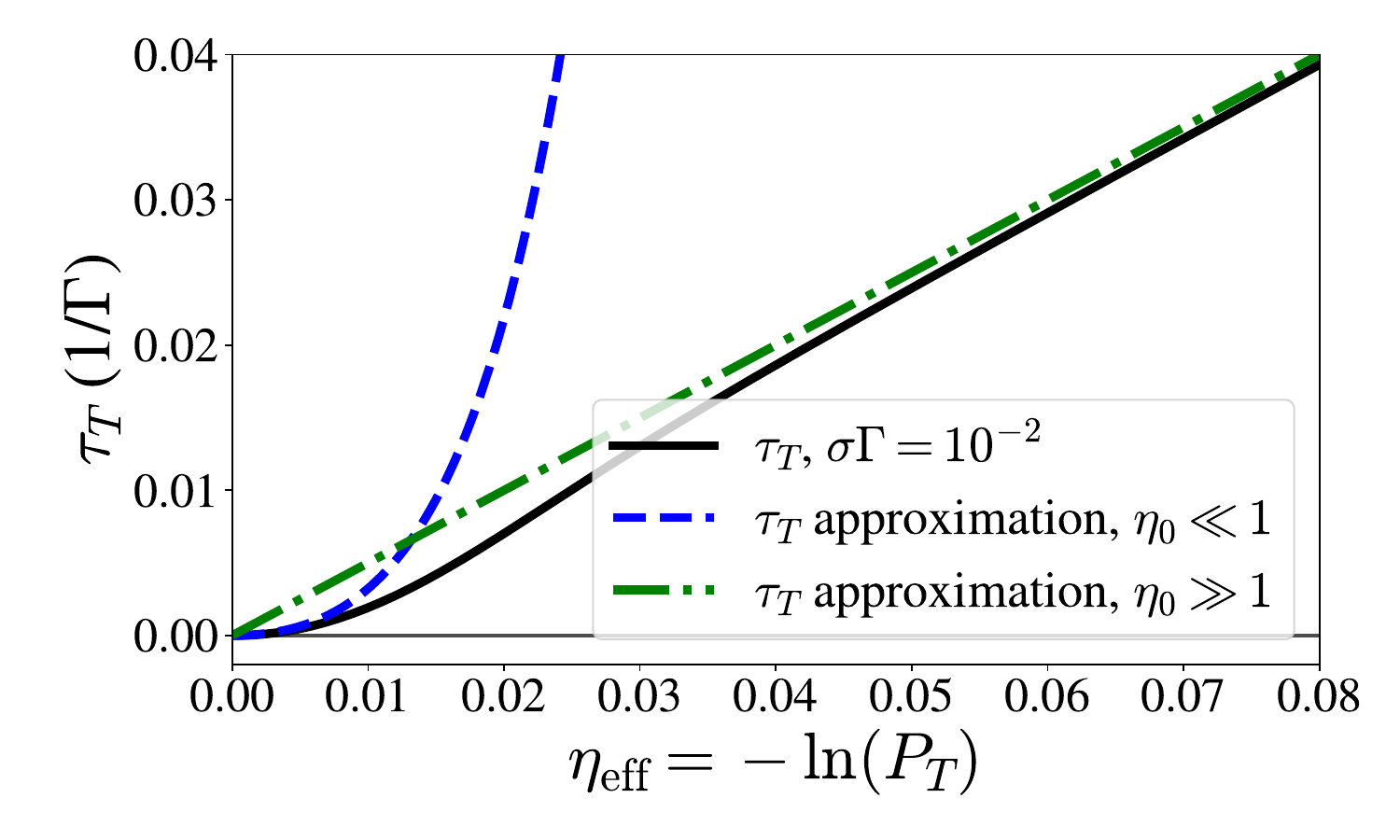}
\caption{Atomic excitation time $\ett$ experienced by transmitted photons for a broad-band pulse vs $\od_{\rm eff}$, together with approximate forms in the limits of $\od_0\ll1$ (blue dashed, given by~\erf{eqn:tauT_approx_BB_lowOD}) and $\od_0\gg1$ (green dot-dashed, given by~\erf{eqn:tauT_approx_BB_highOD}).}
\label{fig:tauT_approximations_BB}
\end{figure}

\section{Cavity model details}\label{sec:cavity_model}
Here we present a complete weak-value calculation of the cavity dwell time for the model presented in~\srf{sec:discussion}.

The model consists of a single-photon input field incident on a Fabry-Perot cavity. For mathematical convenience, we treat this as having zero size, and being located at $z$ = 0. 
The cavity has two mirrors with intensity decay rates $\gamma_1$ and $\gamma_2$.
The probability amplitude ${{\rm{\alpha_{in}}}(t)}$ describes the photon incident on mirror 1 and causes the probability amplitude $\beta(t)$ for the photon to be inside the cavity to grow from an initial value of zero. 
The amplitudes for the photon to be reflected from the cavity and transmitted through the cavity at time $t$ are denoted by ${{\rm{\alpha_{ref}}}(t)}$ and ${{\rm{\alpha_{tr}}}(t)}$, respectively. 
The time evolution of these amplitudes is governed by the following equations,
\begin{gather} \label{Field evolution}
    \Dot{\beta}(t) = -\dfrac{\gamma_{1} + \gamma_{2}}{2}\beta(t) - \sqrt{\gamma_{1}} \, \alpha_{\rm{in}}(t)\,,\\
    \alpha_{\rm{ref}}(t) = \alpha_{\rm{in}}(t) + \sqrt{\gamma_{1}}\,\beta(t)\,,\label{Ref_field}\\
    \alpha_{\rm{tr}}(t)= \sqrt{\gamma_{2}}\,\beta(t)\,.\label{Field evolution 3}
\end{gather} 
 These are isomorphic to the classical field equations for a driven cavity, and we may use the term ``field'' in place of ``probability amplitude'' below. 

The boundary condition used to solve the field differential equation is that the cavity field is zero for $t\rightarrow\pm\infty$. 
It is also necessary that the input field amplitudes vanish at the boundaries for normalization of the input field intensity:
\begin{gather}
    \beta(t \to -\infty) = \beta(t \to +\infty) = 0\,,\\
    \alpha_{\rm{in}}(t \to -\infty) = \alpha_{\rm{in}}(t \to +\infty) = 0\,.
\end{gather}
The mean number of photons in the cavity at time $t$ is given by $\abs{\beta(t)}^2$, and the mean number of photons incident on the cavity per unit time at time $t$ is $\abs{\alpha_{\rm in}(t)}^2$.
Since the input field is a single-photon field, $\int^{\infty}_{-\infty}dt\abs{\alpha_{\rm in}(t)}^2 = 1$.
Also, for this single-photon input field, the reflection and transmission probability densities are 
\begin{gather}\label{eqn:P_ref_density}
    {\rm{\rho}_{\rm{ref}}}(t) = \abs{\alpha_{\rm{ref}}(t)}^2\,,\hspace{3mm}{\rm{\rho}_{\rm{tr}}}(t) = \abs{\alpha_{\rm{tr}}(t)}^2\,,
\end{gather}
and the total reflection and transmission probabilities are found by integrating the respective probability densities over all time,
\begin{gather}\label{eqn:P_ref_0}
    P_{\rm ref} = \hspace{-1mm}\int^{\infty}_{-\infty}\hspace{-1mm}dt\abs{\alpha_{\rm ref}(t)}^2\,,\quad P_{\rm tr} = \hspace{-1mm}\int^{\infty}_{-\infty}\hspace{-1mm}dt\abs{\alpha_{\rm tr}(t)}^2\,.
\end{gather}

The fields propagating towards the left and right in this system can be represented by field amplitudes $\alpha_l(z,t)$ and $\alpha_r(z,t)$, respectively. 
For any position and time, these amplitudes can be expressed in terms of the amplitudes of the incident, reflected or transmitted fields at the time the field enters/exits the cavity as follows:
\begin{align}\label{alpha_definitions}
    &\alpha_l(z,t) = {{\rm{\alpha_{ref}}}(t+z/c)}\Theta(-z)\,,\\
    &\alpha_r(z,t)= {{\rm{\alpha_{in}}}(t-z/c)}\Theta(-z) + {{\rm{\alpha_{tr}}}(t-z/c)}\Theta(z)\,,
\end{align}
where $\Theta$ is the Heaviside step function with $\Theta(0) = 0$ to avoid the cavity itself.

Now, at any arbitrary time, the single photon can be present in any of the left or right moving fields or in the cavity. 
As shown in \frf{fig:cavity_model}, these possible locations for the photon are labelled A, B, C and D, which represent the photon being in the incident field, reflected field, inside the cavity and in the transmitted field, respectively.
The $\delta$-normalized state of a single photon at a position $z$ outside of the cavity which propagates either to the left ($l$) or to the right ($r$) can be defined as
\begin{align}
    &\ket{1}_{l,z} = \ket{\rm B}_{z}\Theta(-z)\,,\\&\ket{1}_{r,z} = \ket{\rm A}_z\Theta(-z) + \ket{\rm D}_z\Theta(z)\,,
\end{align}
and we let $\ket{\rm C}\bra{\rm C}$ denote a unit-normalized state with a photon inside the cavity (equivalent to an atomic excitation in this analogy) at position $z=0$. Hence, the normalized general state of the photon can be represented as a superposition of the above states as
\begin{align}
    \ket{\psi(t)} = &\beta(t) \ket{\rm C} + \int^{\infty}_{-\infty}\frac{dz}{\sqrt{c}}\,\alpha_l(z,t)\ket{1}_{l,z} \nn\\
    &+ \int^{\infty}_{-\infty}\frac{dz}{\sqrt{c}}\,\alpha_r(z,t)\ket{1}_{r,z}\,.
\end{align}

We first note that the cavity dwell time of the average (not post-selected) photon---which is simply the time integral of the average photon number in the cavity---is given by
\begin{equation}
    \intall dt \abs{\braket{\rm C}{\psi(t)}}^2 = \intall dt \abs{\beta(t)}^2= P_{\rm tr}/\gamma_2\,, 
\end{equation}
where we have used \erfa{Field evolution 3}{eqn:P_ref_0}. Recalling that transmission in the cavity model corresponds to scattering in the atomic system, this result is perfectly analogous to the atomic system, for which we had $\etav=P_S/\Gamma$.

We now consider the case when a photon is reflected and then detected at a time $t_{\rm{d}}$ and position $z_{\rm{d}}\,(<0)$. Let $t_{\rm ref}$ be the time at which the photon leaves the cavity after reflection, which can be written as $t_{\rm ref} = t_{\rm{d}} - \abs{z_{\rm{d}}}/c$. The state of the photon when it is detected can be labelled as
\begin{equation}\label{post-state_ref}
    \ket{\chi_{t_{\rm ref}}(t_{\rm d})} := \ket{1}_{l,z_{\rm d}} = \ket{\rm B}_{z_{\rm d}}\,.
\end{equation}
The weak value for the cavity photon number measured at a time $t_{\rm{m}}$, and post-selected on a reflected photon leaving the cavity at a time $t_{\rm ref}$, is given by
\begin{align} \label{wv_ref}
    \left\langle n_{\rm{w}}\right\rangle_{\rm{ref}}(t_{\rm{m}},t_{\rm ref}) = \text{Re}\dfrac{\Braket{\chi_{t_{\rm ref}}(t_{\rm{m}})|\hat{n}_{\rm{c}}|\psi(t_{\rm{m}})}}{\Braket{\chi_{t_{\rm ref}}(t_{\rm{m}})|\psi(t_{\rm{m}})}}\,,
\end{align}
where $\hat{n}_{\rm{c}} := \ket{{\rm C}}\!\bra{{\rm C}}$ is the cavity photon number operator. 
Here $\ket{\psi(t_{\rm{m}})} = \hat{U}(t_{\rm{m}},-\infty)\ket{\psi(-\infty)}$ is the pre-selected state of the photon evolved from time $-\infty$ to $t_{\rm m}$ and 
$\ket{\chi_{t_{\rm ref}}(t_{\rm m})} = \hat{U}^{\dagger}(t_{\rm{d}},t_{\rm{m}})\ket{\chi_{t_{\rm ref}}(t_{\rm d})}$ is the post-selected state at time $t_{\rm m}$, which represents the state of a photon reflected at time $t_{\rm ref}$, evolved backwards to time $t_{\rm m}$. 

The cavity dwell time $\tau_{{\rm B},\,t_{\rm ref}}$ of a reflected photon which exits the cavity at time $t_{\rm ref}$ is found by integrating the weak value of the photon number over the weak measurement time,
\begin{equation}\label{tau_B_cond}
    \tau_{{\rm B},\,t_{\rm ref}} = \int_{-\infty}^{\infty}dt_{\rm{m}} \left\langle n_{\rm{w}} \right\rangle_{\rm{ref}}(t_{\rm{m}},t_{\rm ref})\,.
\end{equation}
The cavity dwell time of a reflected photon is then calculated by averaging the dwell time for a reflected photon which leaves the cavity at time $t_{\rm ref}$, weighted by the conditional probability density of the photon leaving the cavity at $t_{\rm ref}$ given that it is reflected, ${\rm \rho}(t_{\rm ref}|{\rm ref})$. 
This conditional probability density is given by
\begin{equation}
    {\rm \rho}(t_{\rm ref}|{\rm ref}) = \dfrac{{\rm \rho_{ref}}(t_{\rm ref})}{P_{\rm ref}}\,,
\end{equation}
where ${\rm \rho_{ref}}(t_{\rm ref})$ and $P_{\rm ref}$ are defined in \erfa{eqn:P_ref_density}{eqn:P_ref_0}, respectively. 
Hence, the average cavity dwell time for a reflected photon, $\tau_{\rm B}$, is given by
\begin{equation}\label{tau_B_definition}
     \tau_{\rm{B}} = \int_{-\infty}^{\infty}\,dt_{\rm ref}\,{\rm \rho}(t_{\rm ref}|{\rm ref})\,\tau_{{\rm B},\,t_{\rm ref}}\,.
\end{equation}

We now wish to calculate $\tau_{\rm B}$. We begin by finding the weak value of the cavity photon number given in~\erf{wv_ref}.
The denominator of this expression can be expanded as 
\begin{align}\label{eqn:cavity_wk_value_denom}
    \Braket{\chi_{t_{\rm ref}}(t_{\rm{m}})|\psi(t_{\rm{m}})} &= \bra{\chi_{t_{\rm ref}}(t_{\rm d})}\hat{U}(t_{\rm{d}},t_{\rm{m}})\ket{\psi(t_{\rm{m}})} \nn\\
    &= \bra{\chi_{t_{\rm ref}}(t_{\rm d})}\ket{\psi(t_{\rm{d}})} \nn\\
    &= \alpha_{\rm{ref}}(t_{\rm ref})/\sqrt{c}\,,
\end{align}
 since the overlap of the states remains unchanged in time. Furthermore, the numerator can be expanded as
 \begin{align}
     &\Braket{\chi_{t_{\rm ref}}(t_{\rm{m}})|\hat{n}_{\rm{c}}|\psi(t_{\rm{m}})}\nn\\
     &=\bra{\chi_{t_{\rm ref}}(t_{\rm d})}\hat{U}(t_{\rm{d}},t_{\rm{m}})\ket{{\rm C}}\!\bra{{\rm C}}\ket{\psi(t_{\rm{m}})}\nn\\
     &=\beta(t_{\rm m})\bra{\chi_{t_{\rm ref}}(t_{\rm d})}\hat{U}(t_{\rm{d}},t_{\rm{m}})\ket{{\rm C}}\,.
 \end{align}
To calculate $\hat{U}(t_{\rm{d}},t_{\rm{m}})\ket{{\rm C}}$, the field evolution equations (\erfs{Field evolution}{Field evolution 3}) have to be solved again with different initial conditions: an initial state at time $t_{\rm m}$ with the photon being inside the cavity and no input field (i.e, $\alpha_{\rm{in}}(t)=0$ for $t>t_{\rm m}$). 
Since the field inside the cavity will simply decay at a rate of $(\gamma_1+\gamma_2)/2:=\overline{\gamma}$, this state will be
\begin{align}
    &\hat{U}(t_{\rm{d}},t_{\rm{m}})\ket{{\rm C}} \nn\\
    =& e^{-\overline{\gamma}(t_{\rm{d}}-t_{\rm{m}})}\,\Theta(t_{\rm d} - t_{\rm{m}})\ket{\rm C} \nn\\
    & +\sqrt{\gamma_1}\int_{-\infty}^0 \frac{dz}{\sqrt{c}} e^{-\overline{\gamma}\qty(T_{{\rm d},z}-t_{\rm{m}})}\,\Theta\qty(T_{{\rm d},z}-t_{\rm{m}})\ket{1}_{l,z} \nn\\
    &+ \sqrt{\gamma_2}\int_0^\infty \frac{dz}{\sqrt{c}} e^{-\overline{\gamma}\qty(T_{{\rm d},z}-t_{\rm{m}})}\,\Theta\qty(T_{{\rm d},z}-t_{\rm{m}})\ket{1}_{r,z}\,,
\end{align}
where $T_{{\rm d},z}:= t_{\rm d}-\abs{z}/c$, so $T_{\rm d,0}=t_{\rm d}$ and $T_{{\rm d},z_{\rm d}}=t_{\rm ref}$.
The numerator of~\erf{wv_ref} is then given by
\begin{align}\label{eqn:cavity_wk_value_numerator}
    &\beta(t_{\rm m})\bra{\chi_{t_{\rm ref}}(t_{\rm d})}\hat{U}(t_{\rm{d}},t_{\rm{m}})\ket{{\rm C}}\nn\\
    &=\sqrt{\gamma_1}\beta(t_{\rm m}) e^{-(\gamma_1 + \gamma_2)(t_{\rm{ref}}-t_{\rm{m}})/2}\,\Theta(t_{\rm ref} - t_{\rm{m}})/\sqrt{c}\,.
\end{align}
Here the Heaviside function causes the weak value to be zero when the weak measurement is performed after the photon has exited the cavity via reflection. 

Before proceeding, it is useful to express some quantities in frequency space. Firstly, transforming the evolution equations~\erfa{Field evolution}{Ref_field} into the frequency domain gives
\begin{align} \label{Fields_omega}
    \tilde{\beta}(\omega) &= -\dfrac{2\sqrt{\gamma_{1}}\,\tilde{\alpha}_{\rm{in}}(\omega)}{\gamma_{1}+\gamma_{2} + 2i\omega}\,,\\ \label{field soln}
    \tilde{\alpha}_{\rm{ref}}(\omega) &= \bigg(\dfrac{\gamma_{2}-\gamma_{1} + 2i\omega}{\gamma_{1}+\gamma_{2} + 2i\omega}\bigg)\tilde{\alpha}_{\rm{in}}(\omega)\,,\\ 
    \tilde{\alpha}_{\rm{tr}}(\omega) &= -\dfrac{2\sqrt{\gamma_{1}\gamma_{2}}\,\tilde{\alpha}_{\rm{in}}(\omega)}{\gamma_{1}+\gamma_{2} + 2i\omega}\,.
\end{align} 
Additionally, the reflection probability can be expressed as
\begin{equation} \label{P_ref_w}
    P_{\rm ref} = \hspace{-1mm}\int_{-\infty}^{\infty}d\omega\,{\rm {\tilde{\rho}_{ref}}}(\omega)\,,
\end{equation}
\noindent where
\begin{equation}\label{rho_ref_w}
    {\rm{\tilde{\rho}_{ref}}}(\omega) = \abs{\tilde{\alpha}_{\rm{ref}}(\omega)}^2 = \dfrac{(\gamma_1 - \gamma_2)^2 + 4\omega^2}{(\gamma_1 + \gamma_2)^2 + 4\omega^2}\abs{\tilde{\alpha}_{\rm{in}}(\omega)}^2\,.
\end{equation}
Here the tildes represent a Fourier transform and~\erf{P_ref_w} follows from Plancherel's theorem. 
Also, $\omega=0$ corresponds to the resonance frequency of the cavity. 

We can now express $\beta(t_{\rm m})$ in terms of its Fourier transform $\tilde{\beta}(\omega)$ as
\begin{align}\label{eqn:cavity:beta_transform}
    \beta(t_{\rm m}) &= \frac{1}{\sqrt{2\pi}}\intall d\omega \tilde{\beta}(\omega) e^{i\omega t_{\rm m}} \nn\\
    &= -2\frac{\sqrt{\gamma_1}}{\sqrt{2\pi}}\intall d\omega \frac{\tilde{\alpha}_{\rm in}(\omega)e^{i\omega t_{\rm m}}}{\gamma_1+\gamma_2+2i\omega}\,,
\end{align}
\noindent where we have used~\erf{Fields_omega}.

Using~\erf{wv_ref}, \erf{eqn:cavity_wk_value_denom}, \erf{eqn:cavity_wk_value_numerator} and \erf{eqn:cavity:beta_transform}, the weak value for the cavity photon number turns out to be
\begin{widetext}
\begin{align}
    \left\langle n_{\rm{w}}\right\rangle_{\rm{ref}}(t_{\rm{m}},t_{\rm ref}) = -\dfrac{2\gamma_1}{\sqrt{2\pi}\alpha_{\rm{ref}}(t_{\rm ref})}\displaystyle\int_{-\infty}^{\infty}d\omega\,\dfrac{\tilde{\alpha}_{\rm{in}}(\omega)}{\gamma_1 + \gamma_2 + 2i\omega}e^{i\omega t_{\rm{m}}-(\gamma_1 + \gamma_2)(t_{\rm{ref}}-t_{\rm{m}})/2}\,\Theta(t_{\rm ref} - t_{\rm{m}})\,.
\end{align}
\end{widetext}
Integrating this expression over $t_{\rm m}$ then gives 
\begin{align}\label{eqn:tauB_tref_frequency}
    &\tau_{{\rm B},t_{\rm ref}}=\int_{-\infty}^{\infty}dt_{\rm{m}} \left\langle n_{\rm{w}} \right\rangle_{\rm{ref}}(t_{\rm{m}},t_{\rm ref}) \nn\\
    &=-\frac{4\gamma_1}{\sqrt{2\pi}\alpha_{\rm ref}(t_{\rm ref})}\intall \frac{d\omega\, \tilde{\alpha}_{\rm in}(\omega)e^{i\omega t_{\rm ref}}}{\qty[\gamma_1+\gamma_2+2i\omega]^2}\,.
\end{align}
Using \erf{eqn:P_ref_density} for $\rho_{\rm ref}(t_{\rm ref})$, \erfa{P_ref_w}{rho_ref_w} for $P_{\rm ref}$ and the above equation for $\tau_{{\rm B},t_{\rm ref}}$ in the definition of $\tau_{\rm B}$ given in \erf{tau_B_definition}, $\tau_{\rm B}$ is then given by
\begin{align}
    &\tau_{\rm B}=\frac{1}{P_{\rm ref}}\intall dt_{\rm ref}\abs{\alpha_{\rm ref}(t_{\rm ref})}^2\tau_{{\rm B},t_{\rm ref}} \nn\\
    &=-\frac{4\gamma_1}{P_{\rm ref}}\intall d\omega\frac{\tilde{\alpha}_{\rm in}(\omega)\tilde{\alpha}^*_{\rm ref}(\omega)}{\qty[\gamma_1+\gamma_2+2i\omega]^2}\nn\\
    &=\dfrac{\displaystyle\int_{-\infty}^{\infty}d\omega\left(\dfrac{\gamma_1 - \gamma_2 + 2i\omega}{\gamma_1 + \gamma_2 + 2i\omega}\right)\dfrac{\abs{\tilde{\alpha}_{\rm{in}}(\omega)}^2}{[(\gamma_1+\gamma_2)^2 + 4\omega^2]}}{\dfrac{1}{4\gamma_1}\displaystyle\int_{-\infty}^{\infty}d\omega\,\left[\dfrac{(\gamma_1-\gamma_2)^2 + 4\omega^2}{(\gamma_1+\gamma_2)^2 + 4\omega^2}\right]\abs{\tilde{\alpha}_{\rm{in}}(\omega)}^2}\,, \label{eqn:H30}
\end{align}
where we have used \erf{field soln} for $\tilde{\alpha}^*_{\rm ref}(\omega)$. 

\section{Cavity and atom models correspondence}\label{sec:cavity_corr}

In this section we compare the cavity model to the atomic system, demonstrate that the Feynman-path calculation for $\tau_{\rm B}$ presented in~\srf{sec:discussion} agrees with the direct calculation of $\tau_{\rm B}$ given in~\arf{sec:cavity_model}, and derive the approximate expression $\tau_{\rm B}\approx-\od_0^c/\gamma_2$ mentioned in~\srf{sec:discussion}, which is valid in the limit of a narrow-band input pulse and low $\od_0^c$.

To make a comparison with the case of an atomic medium, 
we specialize to a resonant input pulse in the narrow-band limit, $\abs{\tilde{\alpha}_{\rm{in}}(\omega)}^2 = \delta(\omega)$. In this regime, \erf{eqn:H30} 
reduces to
\begin{equation}\label{eqn:tauB_decay_rates}
    \tau_{\rm B} = \rm{\dfrac{4\gamma_{1}}{\gamma_1^2 - \gamma_2^2}}\,.
\end{equation}
Now, using~\erfa{P_ref_w}{rho_ref_w}, the decay rates of the cavity can be related to the resonant `optical depth' of the cavity $\eta_0^c$ by
\begin{equation}\label{P_ref}
    P_{\rm B} := P_{\rm ref}= \bigg(\dfrac{\gamma_{1}-\gamma_{2}}{\gamma_{1}+\gamma_{2}}\bigg)^2 := e^{-\rm{\eta^c_0}}\,.
\end{equation}
Note that the probability of reflection from the cavity is defined as $e^{-\rm{\eta^c_0}}$ since reflection from the cavity is analogous to transmission through the atomic cloud.
Here, since \erf{P_ref} is symmetric with respect to $\gamma_1$ and $\gamma_2$, one of them needs to be chosen lesser than the other to express them in terms of $\od^c_0$.
Letting $\gamma_>$ ($\gamma_<$) denote the larger (smaller) of the two, we have
\begin{equation}\label{gamma_ratio}
    \dfrac{\gamma_<}{\gamma_>} = \dfrac{1 - e^{-\od^c_0/2}}{1 + e^{-\od^c_0/2}}\approx \frac{\od^c_0}{4}\,,
\end{equation}
where the approximation is valid for $\od^c_0\ll1$. 

In order to determine which choice ($\gamma_1<\gamma_2$ or $\gamma_1>\gamma_2$) better models the atomic problem, we note that the optical depth of a uniform medium of length $L$ and coupling constant $g_0$ is given by $\od_0=4g_0^2L/c\Gamma$. 
Since $\Gamma$ represents the loss rate in the atomic system and $\gamma_2$ represents the ``loss rate'' (transmission rate) in the cavity model, it is natural to draw an analogy between these rates. 
Similarly, $\gamma_1$ describes the coupling into the cavity, while $g_0^2L/c$ is a measure of the effective coupling rate into the atoms. 
We therefore choose $\gamma_1<\gamma_2$ so that $\od^c_0 = 4\gamma_1/\gamma_2$, in analogy with the atomic problem.
Using~\erf{gamma_ratio} in \erf{eqn:tauB_decay_rates}, the cavity dwell time for a reflected photon then becomes
\begin{align}\label{tau_B}
    \tau_{\rm B} &= -\left(e^{\eta_0^c/2}-e^{-\eta_0^c/2}\right)/\gamma_2 \nn \\
    &= -2\sinh\left(\eta_0^c/2\right)/\gamma_2\,.
\end{align}
For $\eta_0^c<<1$, up to first order in $\eta_0^c$ we have $\tau_{\rm B} \approx -\eta_0^c/\gamma_2$, in perfect analogy with the atomic case, for which we had $\ett \approx -\eta_0/\Gamma$.
The cavity model is thus a good analogy for the atomic cloud in the limit of low optical depth and for resonant, narrow-band pulses.

The above result can also be expressed in terms of the reflection ($r_1,r_2$) and transmission ($t_1,t_2$) coefficients of the cavity mirrors, where $r_j = \abs{r_j}$ and $t_j = i\abs{t_j}$. 
Firstly, the probabilities of the photon ultimately being reflected or transmitted can be found by summing the probability amplitudes for all Feynman paths leading to reflection and transmission as follows:
\begin{align}\label{Feyn-ref}
    P_{\rm B} &= \abs{r_1 + t_1r_2t_1 + t_1r_2r_1r_2t_1 + \cdots}^2 \nn \\
    &= \abs{r_1 -\abs{t_1}^2r_2\sum_{n=0}^{\infty}(r_1r_2)^n }^2\, \nn \\ 
    &= \left(r_1 - \dfrac{\abs{t_1}^2r_2}{1-r_1r_2}\right)^2 \nn \\ 
    &= \left(\dfrac{r_1 - r_2}{1-r_1r_2}\right)^2\,,
\end{align}
\begin{align}
    P_{\rm D} := P_{\rm tr} &= \abs{t_1t_2 + t_1r_2r_1t_2 + t_1r_2r_1r_2r_1t_2 + \cdots}^2 \nn \\
    &= \abs{-\abs{t_1}\abs{t_2}\sum_{n=0}^{\infty}(r_1r_2)^n }^2\, \nn \\ &= \left(\dfrac{\abs{t_1}\abs{t_2}}{1-r_1r_2}\right)^2\,.
\end{align}
In the above equations, $\abs{r_i}^2 + \abs{t_i}^2 = 1$ and $r_1r_2 < 1$ have been used.

Next, we wish to express the cavity decay rates $\gamma_1$ and $\gamma_2$ in terms of the reflection and transmission coefficients of the mirrors. The decay rate of mirror $j$ can be found by calculating the ring-down time of the field amplitude in the cavity---which is equal to $2/\gamma_j$, since $\gamma_j$ is the decay rate of the intensity---assuming the other mirror is perfectly reflecting. Hence, we have
\begin{equation}
    \dfrac{2}{\gamma_j} = \sum_{k = 0}^{\infty} (2k+1)\dfrac{\tau_{\rm rt}}{2} \dfrac{r_j^k}{\sum_{n = 0}^{\infty}r_j^n}= \dfrac{\tau_{\rm rt}}{2}\left(\dfrac{1+r_j}{1-r_j}\right)\,,
\end{equation}
where $\tau_{\rm rt}$ is the round-trip time of the cavity. 
Rearranging, we have
\begin{equation}\label{gamma-r relation}
    \gamma_i = \dfrac{4}{\tau_{\rm rt}}\left(\dfrac{1-r_i}{1+r_i}\right)\,.
\end{equation}
Substituting the decay rates into \erf{P_ref} gives
\begin{equation}\label{r-OD relation}
    P_{\rm B}= \bigg(\dfrac{\gamma_{1}-\gamma_{2}}{\gamma_{1}+\gamma_{2}}\bigg)^2 = \left(\dfrac{r_1 - r_2}{1-r_1r_2}\right)^2 = e^{-\od^c_0}\,,
\end{equation}
as expected from \erf{Feyn-ref}.  

From~\erf{r-OD relation} it is clear that $\od^c_0=0$ is achieved by taking either $r_1=1$ or $r_2=1$. Since we have chosen $\gamma_1 <\gamma_2$, \erf{gamma-r relation} implies that $r_2 < r_1$, and therefore the limit of $\od^c_0\ll1$ will correspond to $r_1\approx 1$.
Using \erfa{gamma-r relation}{gamma_ratio}, $r_1$ and $r_2$ can be expressed as
\begin{equation}
    r_2 = \dfrac{1 - \gamma_2\tau_{\rm rt}/4}{1 + \gamma_2\tau_{\rm rt}/4}\,,\hspace{5mm}r_1 = \dfrac{r_2 + e^{-\od^c_0/2}}{1 + r_2e^{-\od^c_0/2}}\,.
\end{equation}
Since the cavity is assumed to be infinitesimally small, $\tau_{\rm rt}$ is infinitesimally small. 
If $\gamma_2$ is assumed to be a finite constant, we then have $\gamma_2\tau_{\rm rt}\ll1$, which implies that $r_2\approx1$. 
So, we have $r_1\approx1$ and $r_2\approx1$ with $r_2<r_1$.
Using~\erf{gamma-r relation} in~\erf{eqn:tauB_decay_rates}, the cavity dwell time for a reflected photon can be found to be
\begin{equation}
    \tau_{\rm B} = -\dfrac{\tau_{\rm rt}}{4}\abs{t_1}^2\dfrac{(1+r_2)^2}{(r_1-r_2)(1-r_1r_2)}\,.
\end{equation}
For $r_2\approx 1$, this can be approximated as
\begin{equation}
    \tau_{\rm B} \approx -\tau_{\rm rt}\dfrac{r_2\abs{t_1}^2}{(1-r_1r_2)(r_1-r_2)}\,,
\end{equation}
in agreement with the pointer-shift calculation of $\tau_{\rm B}$ given in \erf{pointer tau_B} of the main text.
Furthermore, for $\od^c_0\ll1$ and $r_1\approx 1$, up to first order in $\od^c_0$ and $1 - r_1$ we have
\begin{equation}
    P_{\rm D} \approx \od^c_0 = 2(1-r_1)\left(\dfrac{1+r_2}{1-r_2}\right)
\end{equation}
and
\begin{equation}
    \tau_{\rm B} \approx -\dfrac{\tau_{\rm rt}}{2}(1-r_1)\left(\dfrac{1+r_2}{1-r_2}\right)^2\,.
\end{equation}
Hence, using the above equations and \erf{gamma-r relation}, it is clear that we have
\begin{equation}
    \tau_{\rm B} \approx -P_{\rm D}/\gamma_2 \approx -\od_0^c/\gamma_2\,.
\end{equation}

\bibliography{CvNC}

\begin{thebibliography}{56}%
\makeatletter
\providecommand \@ifxundefined [1]{%
 \@ifx{#1\undefined}
}%
\providecommand \@ifnum [1]{%
 \ifnum #1\expandafter \@firstoftwo
 \else \expandafter \@secondoftwo
 \fi
}%
\providecommand \@ifx [1]{%
 \ifx #1\expandafter \@firstoftwo
 \else \expandafter \@secondoftwo
 \fi
}%
\providecommand \natexlab [1]{#1}%
\providecommand \enquote  [1]{``#1''}%
\providecommand \bibnamefont  [1]{#1}%
\providecommand \bibfnamefont [1]{#1}%
\providecommand \citenamefont [1]{#1}%
\providecommand \href@noop [0]{\@secondoftwo}%
\providecommand \href [0]{\begingroup \@sanitize@url \@href}%
\providecommand \@href[1]{\@@startlink{#1}\@@href}%
\providecommand \@@href[1]{\endgroup#1\@@endlink}%
\providecommand \@sanitize@url [0]{\catcode `\\12\catcode `\$12\catcode `\&12\catcode `\#12\catcode `\^12\catcode `\_12\catcode `\%12\relax}%
\providecommand \@@startlink[1]{}%
\providecommand \@@endlink[0]{}%
\providecommand \url  [0]{\begingroup\@sanitize@url \@url }%
\providecommand \@url [1]{\endgroup\@href {#1}{\urlprefix }}%
\providecommand \urlprefix  [0]{URL }%
\providecommand \Eprint [0]{\href }%
\providecommand \doibase [0]{https://doi.org/}%
\providecommand \selectlanguage [0]{\@gobble}%
\providecommand \bibinfo  [0]{\@secondoftwo}%
\providecommand \bibfield  [0]{\@secondoftwo}%
\providecommand \translation [1]{[#1]}%
\providecommand \BibitemOpen [0]{}%
\providecommand \bibitemStop [0]{}%
\providecommand \bibitemNoStop [0]{.\EOS\space}%
\providecommand \EOS [0]{\spacefactor3000\relax}%
\providecommand \BibitemShut  [1]{\csname bibitem#1\endcsname}%
\let\auto@bib@innerbib\@empty
%</preamble>
\bibitem [{\citenamefont {Allen}\ and\ \citenamefont {Eberly}(1975)}]{AllenEberlyMonograph}%
  \BibitemOpen
  \bibfield  {author} {\bibinfo {author} {\bibfnamefont {L.}~\bibnamefont {Allen}}\ and\ \bibinfo {author} {\bibfnamefont {J.}~\bibnamefont {Eberly}},\ }\href@noop {} {\emph {\bibinfo {title} {Optical Resonance and Two-Level Atoms}}}\ (\bibinfo  {publisher} {John Wiley and Sons, Inc},\ \bibinfo {year} {1975})\BibitemShut {NoStop}%
\bibitem [{\citenamefont {Loudon}(1983)}]{loudon2000quantum}%
  \BibitemOpen
  \bibfield  {author} {\bibinfo {author} {\bibfnamefont {R.}~\bibnamefont {Loudon}},\ }\href@noop {} {\emph {\bibinfo {title} {{The Quantum Theory of Light}}}},\ \bibinfo {edition} {2nd}\ ed.\ (\bibinfo {year} {1983})\BibitemShut {NoStop}%
\bibitem [{\citenamefont {Scully}\ and\ \citenamefont {Zubairy}(1997)}]{scully_zubairy_1997}%
  \BibitemOpen
  \bibfield  {author} {\bibinfo {author} {\bibfnamefont {M.~O.}\ \bibnamefont {Scully}}\ and\ \bibinfo {author} {\bibfnamefont {M.~S.}\ \bibnamefont {Zubairy}},\ }\href {https://doi.org/10.1017/CBO9780511813993} {\emph {\bibinfo {title} {Quantum Optics}}}\ (\bibinfo  {publisher} {Cambridge University Press},\ \bibinfo {year} {1997})\BibitemShut {NoStop}%
\bibitem [{\citenamefont {Einstein}(1909)}]{einstein_1909}%
  \BibitemOpen
  \bibfield  {author} {\bibinfo {author} {\bibfnamefont {A.}~\bibnamefont {Einstein}},\ }\bibfield  {title} {\bibinfo {title} {Entwicklung unserer anschauungen \"{u}ber das wesen und die konstitution der strahlung},\ }\href@noop {} {\bibfield  {journal} {\bibinfo  {journal} {Physikalische Zeitschrift}\ }\textbf {\bibinfo {volume} {10}},\ \bibinfo {pages} {817} (\bibinfo {year} {1909})}\BibitemShut {NoStop}%
\bibitem [{\citenamefont {Walther}\ \emph {et~al.}(2006)\citenamefont {Walther}, \citenamefont {Varcoe}, \citenamefont {Englert},\ and\ \citenamefont {Becker}}]{Walther_2006}%
  \BibitemOpen
  \bibfield  {author} {\bibinfo {author} {\bibfnamefont {H.}~\bibnamefont {Walther}}, \bibinfo {author} {\bibfnamefont {B.~T.~H.}\ \bibnamefont {Varcoe}}, \bibinfo {author} {\bibfnamefont {B.-G.}\ \bibnamefont {Englert}},\ and\ \bibinfo {author} {\bibfnamefont {T.}~\bibnamefont {Becker}},\ }\bibfield  {title} {\bibinfo {title} {Cavity quantum electrodynamics},\ }\href {https://doi.org/10.1088/0034-4885/69/5/R02} {\bibfield  {journal} {\bibinfo  {journal} {Rep. Prog. Phys.}\ }\textbf {\bibinfo {volume} {69}},\ \bibinfo {pages} {1325} (\bibinfo {year} {2006})}\BibitemShut {NoStop}%
\bibitem [{\citenamefont {Madsen}\ \emph {et~al.}(2022)\citenamefont {Madsen}, \citenamefont {Laudenbach}, \citenamefont {Askarani}, \citenamefont {Rortais}, \citenamefont {Vincent}, \citenamefont {Bulmer}, \citenamefont {Miatto}, \citenamefont {Neuhaus}, \citenamefont {Helt}, \citenamefont {Collins}, \citenamefont {Lita}, \citenamefont {Gerrits}, \citenamefont {Nam}, \citenamefont {Vaidya}, \citenamefont {Menotti}, \citenamefont {Dhand}, \citenamefont {Vernon}, \citenamefont {Quesada},\ and\ \citenamefont {Lavoie}}]{Xanadu_BS2022}%
  \BibitemOpen
  \bibfield  {author} {\bibinfo {author} {\bibfnamefont {L.~S.}\ \bibnamefont {Madsen}}, \bibinfo {author} {\bibfnamefont {F.}~\bibnamefont {Laudenbach}}, \bibinfo {author} {\bibfnamefont {M.~F.}\ \bibnamefont {Askarani}}, \bibinfo {author} {\bibfnamefont {F.}~\bibnamefont {Rortais}}, \bibinfo {author} {\bibfnamefont {T.}~\bibnamefont {Vincent}}, \bibinfo {author} {\bibfnamefont {J.~F.~F.}\ \bibnamefont {Bulmer}}, \bibinfo {author} {\bibfnamefont {F.~M.}\ \bibnamefont {Miatto}}, \bibinfo {author} {\bibfnamefont {L.}~\bibnamefont {Neuhaus}}, \bibinfo {author} {\bibfnamefont {L.~G.}\ \bibnamefont {Helt}}, \bibinfo {author} {\bibfnamefont {M.~J.}\ \bibnamefont {Collins}}, \bibinfo {author} {\bibfnamefont {A.~E.}\ \bibnamefont {Lita}}, \bibinfo {author} {\bibfnamefont {T.}~\bibnamefont {Gerrits}}, \bibinfo {author} {\bibfnamefont {S.~W.}\ \bibnamefont {Nam}}, \bibinfo {author} {\bibfnamefont {V.~D.}\ \bibnamefont {Vaidya}}, \bibinfo {author} {\bibfnamefont {M.}~\bibnamefont {Menotti}}, \bibinfo {author}
  {\bibfnamefont {I.}~\bibnamefont {Dhand}}, \bibinfo {author} {\bibfnamefont {Z.}~\bibnamefont {Vernon}}, \bibinfo {author} {\bibfnamefont {N.}~\bibnamefont {Quesada}},\ and\ \bibinfo {author} {\bibfnamefont {J.}~\bibnamefont {Lavoie}},\ }\bibfield  {title} {\bibinfo {title} {{Quantum computational advantage with a programmable photonic processor}},\ }\href {https://doi.org/10.1038/s41586-022-04725-x} {\bibfield  {journal} {\bibinfo  {journal} {Nature}\ }\textbf {\bibinfo {volume} {606}},\ \bibinfo {pages} {75} (\bibinfo {year} {2022})}\BibitemShut {NoStop}%
\bibitem [{\citenamefont {Chiao}\ and\ \citenamefont {Boyce}(1994)}]{RYC1994}%
  \BibitemOpen
  \bibfield  {author} {\bibinfo {author} {\bibfnamefont {R.~Y.}\ \bibnamefont {Chiao}}\ and\ \bibinfo {author} {\bibfnamefont {J.}~\bibnamefont {Boyce}},\ }\bibfield  {title} {\bibinfo {title} {Superluminality, parelectricity, and {E}arnshaw's theorem in media with inverted populations},\ }\href {https://doi.org/10.1103/PhysRevLett.73.3383} {\bibfield  {journal} {\bibinfo  {journal} {Phys. Rev. Lett.}\ }\textbf {\bibinfo {volume} {73}},\ \bibinfo {pages} {3383} (\bibinfo {year} {1994})}\BibitemShut {NoStop}%
\bibitem [{\citenamefont {Diener}(1998)}]{Diener1998}%
  \BibitemOpen
  \bibfield  {author} {\bibinfo {author} {\bibfnamefont {G.}~\bibnamefont {Diener}},\ }\bibfield  {title} {\bibinfo {title} {Energy balance and energy transport velocity in dispersive media},\ }\href {https://doi.org/https://doi.org/10.1002/andp.199851007-806} {\bibfield  {journal} {\bibinfo  {journal} {Ann. Phys.}\ }\textbf {\bibinfo {volume} {510}},\ \bibinfo {pages} {639} (\bibinfo {year} {1998})}\BibitemShut {NoStop}%
\bibitem [{\citenamefont {Huttner}\ and\ \citenamefont {Barnett}(1992)}]{Barnett1992}%
  \BibitemOpen
  \bibfield  {author} {\bibinfo {author} {\bibfnamefont {B.}~\bibnamefont {Huttner}}\ and\ \bibinfo {author} {\bibfnamefont {S.~M.}\ \bibnamefont {Barnett}},\ }\bibfield  {title} {\bibinfo {title} {Quantization of the electromagnetic field in dielectrics},\ }\href {https://doi.org/10.1103/PhysRevA.46.4306} {\bibfield  {journal} {\bibinfo  {journal} {Phys. Rev. A}\ }\textbf {\bibinfo {volume} {46}},\ \bibinfo {pages} {4306} (\bibinfo {year} {1992})}\BibitemShut {NoStop}%
\bibitem [{\citenamefont {Matloob}\ \emph {et~al.}(1995)\citenamefont {Matloob}, \citenamefont {Loudon}, \citenamefont {Barnett},\ and\ \citenamefont {Jeffers}}]{Jeffers1995}%
  \BibitemOpen
  \bibfield  {author} {\bibinfo {author} {\bibfnamefont {R.}~\bibnamefont {Matloob}}, \bibinfo {author} {\bibfnamefont {R.}~\bibnamefont {Loudon}}, \bibinfo {author} {\bibfnamefont {S.~M.}\ \bibnamefont {Barnett}},\ and\ \bibinfo {author} {\bibfnamefont {J.}~\bibnamefont {Jeffers}},\ }\bibfield  {title} {\bibinfo {title} {Electromagnetic field quantization in absorbing dielectrics},\ }\href {https://doi.org/10.1103/PhysRevA.52.4823} {\bibfield  {journal} {\bibinfo  {journal} {Phys. Rev. A}\ }\textbf {\bibinfo {volume} {52}},\ \bibinfo {pages} {4823} (\bibinfo {year} {1995})}\BibitemShut {NoStop}%
\bibitem [{\citenamefont {Kiilerich}\ and\ \citenamefont {M\o{}lmer}(2020)}]{KlausMolmer2020}%
  \BibitemOpen
  \bibfield  {author} {\bibinfo {author} {\bibfnamefont {A.~H.}\ \bibnamefont {Kiilerich}}\ and\ \bibinfo {author} {\bibfnamefont {K.}~\bibnamefont {M\o{}lmer}},\ }\bibfield  {title} {\bibinfo {title} {Quantum interactions with pulses of radiation},\ }\href {https://doi.org/10.1103/PhysRevA.102.023717} {\bibfield  {journal} {\bibinfo  {journal} {Phys. Rev. A}\ }\textbf {\bibinfo {volume} {102}},\ \bibinfo {pages} {023717} (\bibinfo {year} {2020})}\BibitemShut {NoStop}%
\bibitem [{\citenamefont {Smith}(1960)}]{smith_lifetime_1960}%
  \BibitemOpen
  \bibfield  {author} {\bibinfo {author} {\bibfnamefont {F.~T.}\ \bibnamefont {Smith}},\ }\bibfield  {title} {\bibinfo {title} {Lifetime {Matrix} in {Collision} {Theory}},\ }\href {https://doi.org/10.1103/PhysRev.118.349} {\bibfield  {journal} {\bibinfo  {journal} {Phys. Rev.}\ }\textbf {\bibinfo {volume} {118}},\ \bibinfo {pages} {349} (\bibinfo {year} {1960})}\BibitemShut {NoStop}%
\bibitem [{\citenamefont {Hauge}\ and\ \citenamefont {Støvneng}(1989)}]{hauge_tunneling_1989}%
  \BibitemOpen
  \bibfield  {author} {\bibinfo {author} {\bibfnamefont {E.~H.}\ \bibnamefont {Hauge}}\ and\ \bibinfo {author} {\bibfnamefont {J.~A.}\ \bibnamefont {Støvneng}},\ }\bibfield  {title} {\bibinfo {title} {Tunneling times: a critical review},\ }\href {https://doi.org/10.1103/RevModPhys.61.917} {\bibfield  {journal} {\bibinfo  {journal} {Rev. Mod. Phys.}\ }\textbf {\bibinfo {volume} {61}},\ \bibinfo {pages} {917} (\bibinfo {year} {1989})}\BibitemShut {NoStop}%
\bibitem [{\citenamefont {Ramos}\ \emph {et~al.}(2020)\citenamefont {Ramos}, \citenamefont {Spierings}, \citenamefont {Racicot},\ and\ \citenamefont {Steinberg}}]{ramos_measurement_2020}%
  \BibitemOpen
  \bibfield  {author} {\bibinfo {author} {\bibfnamefont {R.}~\bibnamefont {Ramos}}, \bibinfo {author} {\bibfnamefont {D.}~\bibnamefont {Spierings}}, \bibinfo {author} {\bibfnamefont {I.}~\bibnamefont {Racicot}},\ and\ \bibinfo {author} {\bibfnamefont {A.~M.}\ \bibnamefont {Steinberg}},\ }\bibfield  {title} {\bibinfo {title} {Measurement of the time spent by a tunnelling atom within the barrier region},\ }\href {https://doi.org/10.1038/s41586-020-2490-7} {\bibfield  {journal} {\bibinfo  {journal} {Nature}\ }\textbf {\bibinfo {volume} {583}},\ \bibinfo {pages} {529} (\bibinfo {year} {2020})}\BibitemShut {NoStop}%
\bibitem [{\citenamefont {Sinclair}\ \emph {et~al.}(2022)\citenamefont {Sinclair}, \citenamefont {Angulo}, \citenamefont {Thompson}, \citenamefont {Bonsma-Fisher}, \citenamefont {Brodutch},\ and\ \citenamefont {Steinberg}}]{sinclair_measuring_2022}%
  \BibitemOpen
  \bibfield  {author} {\bibinfo {author} {\bibfnamefont {J.}~\bibnamefont {Sinclair}}, \bibinfo {author} {\bibfnamefont {D.}~\bibnamefont {Angulo}}, \bibinfo {author} {\bibfnamefont {K.}~\bibnamefont {Thompson}}, \bibinfo {author} {\bibfnamefont {K.}~\bibnamefont {Bonsma-Fisher}}, \bibinfo {author} {\bibfnamefont {A.}~\bibnamefont {Brodutch}},\ and\ \bibinfo {author} {\bibfnamefont {A.~M.}\ \bibnamefont {Steinberg}},\ }\bibfield  {title} {\bibinfo {title} {Measuring the {Time} {Atoms} {Spend} in the {Excited} {State} {Due} to a {Photon} {They} {Do} {Not} {Absorb}},\ }\href {https://doi.org/10.1103/PRXQuantum.3.010314} {\bibfield  {journal} {\bibinfo  {journal} {PRX Quantum}\ }\textbf {\bibinfo {volume} {3}},\ \bibinfo {pages} {010314} (\bibinfo {year} {2022})}\BibitemShut {NoStop}%
\bibitem [{Note1()}]{Note1}%
  \BibitemOpen
  \bibinfo {note} {In fact this interaction is not instantaneous due to the finite propagation speed of the probe beam, so the acquired phase shift is proportional to the average excitation probability during the propagation time of the probe, which is short compared to the pulse durations considered in these experiments.}\BibitemShut {Stop}%
\bibitem [{\citenamefont {Wiseman}\ \emph {et~al.}(2023)\citenamefont {Wiseman}, \citenamefont {Steinberg},\ and\ \citenamefont {Hallaji}}]{wiseman_obtaining_2023}%
  \BibitemOpen
  \bibfield  {author} {\bibinfo {author} {\bibfnamefont {H.~M.}\ \bibnamefont {Wiseman}}, \bibinfo {author} {\bibfnamefont {A.~M.}\ \bibnamefont {Steinberg}},\ and\ \bibinfo {author} {\bibfnamefont {M.}~\bibnamefont {Hallaji}},\ }\bibfield  {title} {\bibinfo {title} {Obtaining a single-photon weak value from experiments using a strong (many-photon) coherent state},\ }\href {https://doi.org/10.1116/5.0137579} {\bibfield  {journal} {\bibinfo  {journal} {AVS Quantum Sci.}\ }\textbf {\bibinfo {volume} {5}},\ \bibinfo {pages} {024401} (\bibinfo {year} {2023})}\BibitemShut {NoStop}%
\bibitem [{\citenamefont {Crisp}(1970)}]{crisp1970}%
  \BibitemOpen
  \bibfield  {author} {\bibinfo {author} {\bibfnamefont {M.~D.}\ \bibnamefont {Crisp}},\ }\bibfield  {title} {\bibinfo {title} {{Propagation of Small-Area Pulses of Coherent Light through a Resonant Medium}},\ }\href {https://doi.org/10.1103/physreva.2.2172.2} {\bibfield  {journal} {\bibinfo  {journal} {Phys. Rev. A}\ }\textbf {\bibinfo {volume} {2}},\ \bibinfo {pages} {2172} (\bibinfo {year} {1970})}\BibitemShut {NoStop}%
\bibitem [{\citenamefont {Costanzo}\ \emph {et~al.}(2016)\citenamefont {Costanzo}, \citenamefont {Coelho}, \citenamefont {Pellegrino}, \citenamefont {Mendes}, \citenamefont {Acioli}, \citenamefont {Cassemiro}, \citenamefont {Felinto}, \citenamefont {Zavatta},\ and\ \citenamefont {Bellini}}]{Costanzo2016}%
  \BibitemOpen
  \bibfield  {author} {\bibinfo {author} {\bibfnamefont {L.~S.}\ \bibnamefont {Costanzo}}, \bibinfo {author} {\bibfnamefont {A.~S.}\ \bibnamefont {Coelho}}, \bibinfo {author} {\bibfnamefont {D.}~\bibnamefont {Pellegrino}}, \bibinfo {author} {\bibfnamefont {M.~S.}\ \bibnamefont {Mendes}}, \bibinfo {author} {\bibfnamefont {L.}~\bibnamefont {Acioli}}, \bibinfo {author} {\bibfnamefont {K.~N.}\ \bibnamefont {Cassemiro}}, \bibinfo {author} {\bibfnamefont {D.}~\bibnamefont {Felinto}}, \bibinfo {author} {\bibfnamefont {A.}~\bibnamefont {Zavatta}},\ and\ \bibinfo {author} {\bibfnamefont {M.}~\bibnamefont {Bellini}},\ }\bibfield  {title} {\bibinfo {title} {Zero-area single-photon pulses},\ }\href {https://doi.org/10.1103/PhysRevLett.116.023602} {\bibfield  {journal} {\bibinfo  {journal} {Phys. Rev. Lett.}\ }\textbf {\bibinfo {volume} {116}},\ \bibinfo {pages} {023602} (\bibinfo {year} {2016})}\BibitemShut {NoStop}%
\bibitem [{\citenamefont {Aharonov}\ \emph {et~al.}(1988)\citenamefont {Aharonov}, \citenamefont {Albert},\ and\ \citenamefont {Vaidman}}]{aharonov_how_1988}%
  \BibitemOpen
  \bibfield  {author} {\bibinfo {author} {\bibfnamefont {Y.}~\bibnamefont {Aharonov}}, \bibinfo {author} {\bibfnamefont {D.~Z.}\ \bibnamefont {Albert}},\ and\ \bibinfo {author} {\bibfnamefont {L.}~\bibnamefont {Vaidman}},\ }\bibfield  {title} {\bibinfo {title} {How the result of a measurement of a component of the spin of a spin-1/2 particle can turn out to be 100},\ }\href {https://doi.org/10.1103/PhysRevLett.60.1351} {\bibfield  {journal} {\bibinfo  {journal} {Phys. Rev. Lett.}\ }\textbf {\bibinfo {volume} {60}},\ \bibinfo {pages} {1351} (\bibinfo {year} {1988})}\BibitemShut {NoStop}%
\bibitem [{\citenamefont {Dressel}\ \emph {et~al.}(2014)\citenamefont {Dressel}, \citenamefont {Malik}, \citenamefont {Miatto}, \citenamefont {Jordan},\ and\ \citenamefont {Boyd}}]{Dressel_review_WV}%
  \BibitemOpen
  \bibfield  {author} {\bibinfo {author} {\bibfnamefont {J.}~\bibnamefont {Dressel}}, \bibinfo {author} {\bibfnamefont {M.}~\bibnamefont {Malik}}, \bibinfo {author} {\bibfnamefont {F.~M.}\ \bibnamefont {Miatto}}, \bibinfo {author} {\bibfnamefont {A.~N.}\ \bibnamefont {Jordan}},\ and\ \bibinfo {author} {\bibfnamefont {R.~W.}\ \bibnamefont {Boyd}},\ }\bibfield  {title} {\bibinfo {title} {Colloquium: Understanding quantum weak values: Basics and applications},\ }\href {https://doi.org/10.1103/RevModPhys.86.307} {\bibfield  {journal} {\bibinfo  {journal} {Rev. Mod. Phys.}\ }\textbf {\bibinfo {volume} {86}},\ \bibinfo {pages} {307} (\bibinfo {year} {2014})}\BibitemShut {NoStop}%
\bibitem [{\citenamefont {Dalibard}\ \emph {et~al.}(1992)\citenamefont {Dalibard}, \citenamefont {Castin},\ and\ \citenamefont {Mølmer}}]{dalibard_wave-function_1992}%
  \BibitemOpen
  \bibfield  {author} {\bibinfo {author} {\bibfnamefont {J.}~\bibnamefont {Dalibard}}, \bibinfo {author} {\bibfnamefont {Y.}~\bibnamefont {Castin}},\ and\ \bibinfo {author} {\bibfnamefont {K.}~\bibnamefont {Mølmer}},\ }\bibfield  {title} {\bibinfo {title} {Wave-function approach to dissipative processes in quantum optics},\ }\href {https://doi.org/10.1103/PhysRevLett.68.580} {\bibfield  {journal} {\bibinfo  {journal} {Phys. Rev. Lett.}\ }\textbf {\bibinfo {volume} {68}},\ \bibinfo {pages} {580} (\bibinfo {year} {1992})}\BibitemShut {NoStop}%
\bibitem [{\citenamefont {Carmichael}(1993)}]{carmichael_open_1993}%
  \BibitemOpen
  \bibfield  {author} {\bibinfo {author} {\bibfnamefont {H.~J.}\ \bibnamefont {Carmichael}},\ }\href@noop {} {\emph {\bibinfo {title} {An {Open} {Systems} {Approach} to {Quantum} {Optics}}}}\ (\bibinfo  {publisher} {Springer-Verlag, Berlin},\ \bibinfo {year} {1993})\BibitemShut {NoStop}%
\bibitem [{\citenamefont {Wiseman}(2002)}]{Wiseman2002}%
  \BibitemOpen
  \bibfield  {author} {\bibinfo {author} {\bibfnamefont {H.~M.}\ \bibnamefont {Wiseman}},\ }\bibfield  {title} {\bibinfo {title} {{Weak values, quantum trajectories, and the cavity-QED experiment on wave-particle correlation}},\ }\href {https://doi.org/10.1103/PhysRevA.65.032111} {\bibfield  {journal} {\bibinfo  {journal} {Phys. Rev. A}\ }\textbf {\bibinfo {volume} {65}},\ \bibinfo {pages} {032111} (\bibinfo {year} {2002})}\BibitemShut {NoStop}%
\bibitem [{\citenamefont {Tsang}(2009)}]{tsang_time-symmetric_2009}%
  \BibitemOpen
  \bibfield  {author} {\bibinfo {author} {\bibfnamefont {M.}~\bibnamefont {Tsang}},\ }\bibfield  {title} {\bibinfo {title} {Time-{Symmetric} {Quantum} {Theory} of {Smoothing}},\ }\href {https://doi.org/10.1103/PhysRevLett.102.250403} {\bibfield  {journal} {\bibinfo  {journal} {Phys. Rev. Lett.}\ }\textbf {\bibinfo {volume} {102}},\ \bibinfo {pages} {250403} (\bibinfo {year} {2009})}\BibitemShut {NoStop}%
\bibitem [{\citenamefont {Tsang}(2022)}]{tsang_generalized_2022}%
  \BibitemOpen
  \bibfield  {author} {\bibinfo {author} {\bibfnamefont {M.}~\bibnamefont {Tsang}},\ }\bibfield  {title} {\bibinfo {title} {Generalized conditional expectations for quantum retrodiction and smoothing},\ }\href {https://doi.org/10.1103/PhysRevA.105.042213} {\bibfield  {journal} {\bibinfo  {journal} {Phys. Rev. A}\ }\textbf {\bibinfo {volume} {105}},\ \bibinfo {pages} {042213} (\bibinfo {year} {2022})}\BibitemShut {NoStop}%
\bibitem [{\citenamefont {Wigner}(1955)}]{wigner_lower_1955}%
  \BibitemOpen
  \bibfield  {author} {\bibinfo {author} {\bibfnamefont {E.~P.}\ \bibnamefont {Wigner}},\ }\bibfield  {title} {\bibinfo {title} {Lower {Limit} for the {Energy} {Derivative} of the {Scattering} {Phase} {Shift}},\ }\href {https://doi.org/10.1103/PhysRev.98.145} {\bibfield  {journal} {\bibinfo  {journal} {Phys. Rev.}\ }\textbf {\bibinfo {volume} {98}},\ \bibinfo {pages} {145} (\bibinfo {year} {1955})}\BibitemShut {NoStop}%
\bibitem [{\citenamefont {Bourgain}\ \emph {et~al.}(2013)\citenamefont {Bourgain}, \citenamefont {Pellegrino}, \citenamefont {Jennewein}, \citenamefont {Sortais},\ and\ \citenamefont {Browaeys}}]{bourgain_direct_2013}%
  \BibitemOpen
  \bibfield  {author} {\bibinfo {author} {\bibfnamefont {R.}~\bibnamefont {Bourgain}}, \bibinfo {author} {\bibfnamefont {J.}~\bibnamefont {Pellegrino}}, \bibinfo {author} {\bibfnamefont {S.}~\bibnamefont {Jennewein}}, \bibinfo {author} {\bibfnamefont {Y.~R.~P.}\ \bibnamefont {Sortais}},\ and\ \bibinfo {author} {\bibfnamefont {A.}~\bibnamefont {Browaeys}},\ }\bibfield  {title} {\bibinfo {title} {Direct measurement of the {Wigner} time delay for the scattering of light by a single atom},\ }\href {https://doi.org/10.1364/OL.38.001963} {\bibfield  {journal} {\bibinfo  {journal} {Opt. Lett.}\ }\textbf {\bibinfo {volume} {38}},\ \bibinfo {pages} {1963} (\bibinfo {year} {2013})}\BibitemShut {NoStop}%
\bibitem [{\citenamefont {Wiseman}\ and\ \citenamefont {Milburn}(2009)}]{wiseman_quantum_2009}%
  \BibitemOpen
  \bibfield  {author} {\bibinfo {author} {\bibfnamefont {H.~M.}\ \bibnamefont {Wiseman}}\ and\ \bibinfo {author} {\bibfnamefont {G.~J.}\ \bibnamefont {Milburn}},\ }\href {https://doi.org/10.1017/CBO9780511813948} {\emph {\bibinfo {title} {Quantum {Measurement} and {Control}}}}\ (\bibinfo  {publisher} {Cambridge University Press},\ \bibinfo {address} {Cambridge},\ \bibinfo {year} {2009})\BibitemShut {NoStop}%
\bibitem [{\citenamefont {Minev}\ \emph {et~al.}(2019)\citenamefont {Minev}, \citenamefont {Mundhada}, \citenamefont {Shankar}, \citenamefont {Reinhold}, \citenamefont {Gutiérrez-Jáuregui}, \citenamefont {Schoelkopf}, \citenamefont {Mirrahimi}, \citenamefont {Carmichael},\ and\ \citenamefont {Devoret}}]{minev_catch_2019}%
  \BibitemOpen
  \bibfield  {author} {\bibinfo {author} {\bibfnamefont {Z.~K.}\ \bibnamefont {Minev}}, \bibinfo {author} {\bibfnamefont {S.~O.}\ \bibnamefont {Mundhada}}, \bibinfo {author} {\bibfnamefont {S.}~\bibnamefont {Shankar}}, \bibinfo {author} {\bibfnamefont {P.}~\bibnamefont {Reinhold}}, \bibinfo {author} {\bibfnamefont {R.}~\bibnamefont {Gutiérrez-Jáuregui}}, \bibinfo {author} {\bibfnamefont {R.~J.}\ \bibnamefont {Schoelkopf}}, \bibinfo {author} {\bibfnamefont {M.}~\bibnamefont {Mirrahimi}}, \bibinfo {author} {\bibfnamefont {H.~J.}\ \bibnamefont {Carmichael}},\ and\ \bibinfo {author} {\bibfnamefont {M.~H.}\ \bibnamefont {Devoret}},\ }\bibfield  {title} {\bibinfo {title} {To catch and reverse a quantum jump mid-flight},\ }\href {https://doi.org/10.1038/s41586-019-1287-z} {\bibfield  {journal} {\bibinfo  {journal} {Nature}\ }\textbf {\bibinfo {volume} {570}},\ \bibinfo {pages} {200} (\bibinfo {year} {2019})}\BibitemShut {NoStop}%
\bibitem [{\citenamefont {Dicke}(1954)}]{Dicke_1954}%
  \BibitemOpen
  \bibfield  {author} {\bibinfo {author} {\bibfnamefont {R.~H.}\ \bibnamefont {Dicke}},\ }\bibfield  {title} {\bibinfo {title} {Coherence in spontaneous radiation processes},\ }\href {https://doi.org/10.1103/PhysRev.93.99} {\bibfield  {journal} {\bibinfo  {journal} {Phys. Rev.}\ }\textbf {\bibinfo {volume} {93}},\ \bibinfo {pages} {99} (\bibinfo {year} {1954})}\BibitemShut {NoStop}%
\bibitem [{\citenamefont {Ritchie}\ \emph {et~al.}(1991)\citenamefont {Ritchie}, \citenamefont {Story},\ and\ \citenamefont {Hulet}}]{ritchie_realization_1991}%
  \BibitemOpen
  \bibfield  {author} {\bibinfo {author} {\bibfnamefont {N.~W.~M.}\ \bibnamefont {Ritchie}}, \bibinfo {author} {\bibfnamefont {J.~G.}\ \bibnamefont {Story}},\ and\ \bibinfo {author} {\bibfnamefont {R.~G.}\ \bibnamefont {Hulet}},\ }\bibfield  {title} {\bibinfo {title} {Realization of a measurement of a ``weak value''},\ }\href {https://doi.org/10.1103/PhysRevLett.66.1107} {\bibfield  {journal} {\bibinfo  {journal} {Phys. Rev. Lett.}\ }\textbf {\bibinfo {volume} {66}},\ \bibinfo {pages} {1107} (\bibinfo {year} {1991})}\BibitemShut {NoStop}%
\bibitem [{\citenamefont {Jozsa}(2007)}]{Josza_complexWV}%
  \BibitemOpen
  \bibfield  {author} {\bibinfo {author} {\bibfnamefont {R.}~\bibnamefont {Jozsa}},\ }\bibfield  {title} {\bibinfo {title} {Complex weak values in quantum measurement},\ }\href {https://doi.org/10.1103/PhysRevA.76.044103} {\bibfield  {journal} {\bibinfo  {journal} {Phys. Rev. A}\ }\textbf {\bibinfo {volume} {76}},\ \bibinfo {pages} {044103} (\bibinfo {year} {2007})}\BibitemShut {NoStop}%
\bibitem [{\citenamefont {Dressel}\ and\ \citenamefont {Jordan}(2012)}]{Dressel_imag2012}%
  \BibitemOpen
  \bibfield  {author} {\bibinfo {author} {\bibfnamefont {J.}~\bibnamefont {Dressel}}\ and\ \bibinfo {author} {\bibfnamefont {A.~N.}\ \bibnamefont {Jordan}},\ }\bibfield  {title} {\bibinfo {title} {Significance of the imaginary part of the weak value},\ }\href {https://doi.org/10.1103/PhysRevA.85.012107} {\bibfield  {journal} {\bibinfo  {journal} {Phys. Rev. A}\ }\textbf {\bibinfo {volume} {85}},\ \bibinfo {pages} {012107} (\bibinfo {year} {2012})}\BibitemShut {NoStop}%
\bibitem [{\citenamefont {Toll}(1956)}]{Toll_KK_relation}%
  \BibitemOpen
  \bibfield  {author} {\bibinfo {author} {\bibfnamefont {J.~S.}\ \bibnamefont {Toll}},\ }\bibfield  {title} {\bibinfo {title} {Causality and the dispersion relation: Logical foundations},\ }\href {https://doi.org/10.1103/PhysRev.104.1760} {\bibfield  {journal} {\bibinfo  {journal} {Phys. Rev.}\ }\textbf {\bibinfo {volume} {104}},\ \bibinfo {pages} {1760} (\bibinfo {year} {1956})}\BibitemShut {NoStop}%
\bibitem [{\citenamefont {Hu}(1989)}]{Hu_KK_relation}%
  \BibitemOpen
  \bibfield  {author} {\bibinfo {author} {\bibfnamefont {B.~Y.}\ \bibnamefont {Hu}},\ }\bibfield  {title} {\bibinfo {title} {{Kramers–Kronig in two lines}},\ }\href {https://doi.org/10.1119/1.15901} {\bibfield  {journal} {\bibinfo  {journal} {Am. J. Phys.}\ }\textbf {\bibinfo {volume} {57}},\ \bibinfo {pages} {821} (\bibinfo {year} {1989})}\BibitemShut {NoStop}%
\bibitem [{\citenamefont {Born}\ \emph {et~al.}(1999)\citenamefont {Born}, \citenamefont {Wolf}, \citenamefont {Bhatia}, \citenamefont {Clemmow}, \citenamefont {Gabor}, \citenamefont {Stokes}, \citenamefont {Taylor}, \citenamefont {Wayman},\ and\ \citenamefont {Wilcock}}]{born_wolf}%
  \BibitemOpen
  \bibfield  {author} {\bibinfo {author} {\bibfnamefont {M.}~\bibnamefont {Born}}, \bibinfo {author} {\bibfnamefont {E.}~\bibnamefont {Wolf}}, \bibinfo {author} {\bibfnamefont {A.~B.}\ \bibnamefont {Bhatia}}, \bibinfo {author} {\bibfnamefont {P.~C.}\ \bibnamefont {Clemmow}}, \bibinfo {author} {\bibfnamefont {D.}~\bibnamefont {Gabor}}, \bibinfo {author} {\bibfnamefont {A.~R.}\ \bibnamefont {Stokes}}, \bibinfo {author} {\bibfnamefont {A.~M.}\ \bibnamefont {Taylor}}, \bibinfo {author} {\bibfnamefont {P.~A.}\ \bibnamefont {Wayman}},\ and\ \bibinfo {author} {\bibfnamefont {W.~L.}\ \bibnamefont {Wilcock}},\ }\href {https://doi.org/10.1017/CBO9781139644181} {\emph {\bibinfo {title} {Principles of Optics: Electromagnetic Theory of Propagation, Interference and Diffraction of Light}}},\ \bibinfo {edition} {7th}\ ed.\ (\bibinfo  {publisher} {Cambridge University Press},\ \bibinfo {year} {1999})\BibitemShut {NoStop}%
\bibitem [{\citenamefont {Crisp}(1971)}]{crisp_concept_1971}%
  \BibitemOpen
  \bibfield  {author} {\bibinfo {author} {\bibfnamefont {M.~D.}\ \bibnamefont {Crisp}},\ }\bibfield  {title} {\bibinfo {title} {Concept of {Group} {Velocity} in {Resonant} {Pulse} {Propagation}},\ }\href {https://doi.org/10.1103/PhysRevA.4.2104} {\bibfield  {journal} {\bibinfo  {journal} {Phys. Rev. A}\ }\textbf {\bibinfo {volume} {4}},\ \bibinfo {pages} {2104} (\bibinfo {year} {1971})}\BibitemShut {NoStop}%
\bibitem [{\citenamefont {Garrett}\ and\ \citenamefont {McCumber}(1970)}]{garrett_propagation_1970}%
  \BibitemOpen
  \bibfield  {author} {\bibinfo {author} {\bibfnamefont {C.~G.~B.}\ \bibnamefont {Garrett}}\ and\ \bibinfo {author} {\bibfnamefont {D.~E.}\ \bibnamefont {McCumber}},\ }\bibfield  {title} {\bibinfo {title} {Propagation of a {Gaussian} {Light} {Pulse} through an {Anomalous} {Dispersion} {Medium}},\ }\href {https://doi.org/10.1103/PhysRevA.1.305} {\bibfield  {journal} {\bibinfo  {journal} {Phys. Rev. A}\ }\textbf {\bibinfo {volume} {1}},\ \bibinfo {pages} {305} (\bibinfo {year} {1970})}\BibitemShut {NoStop}%
\bibitem [{\citenamefont {Peatross}\ \emph {et~al.}(2000)\citenamefont {Peatross}, \citenamefont {Glasgow},\ and\ \citenamefont {Ware}}]{peatross_average_2000}%
  \BibitemOpen
  \bibfield  {author} {\bibinfo {author} {\bibfnamefont {J.}~\bibnamefont {Peatross}}, \bibinfo {author} {\bibfnamefont {S.~A.}\ \bibnamefont {Glasgow}},\ and\ \bibinfo {author} {\bibfnamefont {M.}~\bibnamefont {Ware}},\ }\bibfield  {title} {\bibinfo {title} {Average {Energy} {Flow} of {Optical} {Pulses} in {Dispersive} {Media}},\ }\href {https://doi.org/10.1103/PhysRevLett.84.2370} {\bibfield  {journal} {\bibinfo  {journal} {Phys. Rev. Lett.}\ }\textbf {\bibinfo {volume} {84}},\ \bibinfo {pages} {2370} (\bibinfo {year} {2000})}\BibitemShut {NoStop}%
\bibitem [{\citenamefont {Ware}\ \emph {et~al.}(2001)\citenamefont {Ware}, \citenamefont {Glasgow},\ and\ \citenamefont {Peatross}}]{ware_role_2001}%
  \BibitemOpen
  \bibfield  {author} {\bibinfo {author} {\bibfnamefont {M.}~\bibnamefont {Ware}}, \bibinfo {author} {\bibfnamefont {S.~A.}\ \bibnamefont {Glasgow}},\ and\ \bibinfo {author} {\bibfnamefont {J.}~\bibnamefont {Peatross}},\ }\bibfield  {title} {\bibinfo {title} {Role of group velocity in tracking field energy in linear dielectrics},\ }\href {https://doi.org/10.1364/OE.9.000506} {\bibfield  {journal} {\bibinfo  {journal} {Opt. Express}\ }\textbf {\bibinfo {volume} {9}},\ \bibinfo {pages} {506} (\bibinfo {year} {2001})}\BibitemShut {NoStop}%
\bibitem [{\citenamefont {Sinclair}(2021)}]{Sinclair2021_thesis}%
  \BibitemOpen
  \bibfield  {author} {\bibinfo {author} {\bibfnamefont {J.}~\bibnamefont {Sinclair}},\ }\emph {\bibinfo {title} {{Weakly Measuring the Time a Transmitted Photon Causes Atoms to Spend in the Excited State}}},\ \href@noop {} {Ph.D. thesis},\ \bibinfo  {school} {University of Toronto} (\bibinfo {year} {2021})\BibitemShut {NoStop}%
\bibitem [{\citenamefont {Brillouin}(1960)}]{brillouin_wave_1960}%
  \BibitemOpen
  \bibfield  {author} {\bibinfo {author} {\bibfnamefont {L.}~\bibnamefont {Brillouin}},\ }\href@noop {} {\emph {\bibinfo {title} {Wave {Propagation} and {Group} {Velocity}}}}\ (\bibinfo  {publisher} {Academic Press},\ \bibinfo {year} {1960})\BibitemShut {NoStop}%
\bibitem [{\citenamefont {Stenner}\ \emph {et~al.}(2003)\citenamefont {Stenner}, \citenamefont {Gauthier},\ and\ \citenamefont {Neifeld}}]{stenner_speed_2003}%
  \BibitemOpen
  \bibfield  {author} {\bibinfo {author} {\bibfnamefont {M.~D.}\ \bibnamefont {Stenner}}, \bibinfo {author} {\bibfnamefont {D.~J.}\ \bibnamefont {Gauthier}},\ and\ \bibinfo {author} {\bibfnamefont {M.~A.}\ \bibnamefont {Neifeld}},\ }\bibfield  {title} {\bibinfo {title} {The speed of information in a ‘fast-light’ optical medium},\ }\href {https://doi.org/10.1038/nature02016} {\bibfield  {journal} {\bibinfo  {journal} {Nature}\ }\textbf {\bibinfo {volume} {425}},\ \bibinfo {pages} {695} (\bibinfo {year} {2003})}\BibitemShut {NoStop}%
\bibitem [{\citenamefont {Pusey}(2014)}]{pusey_anomalous_2014}%
  \BibitemOpen
  \bibfield  {author} {\bibinfo {author} {\bibfnamefont {M.~F.}\ \bibnamefont {Pusey}},\ }\bibfield  {title} {\bibinfo {title} {Anomalous {Weak} {Values} {Are} {Proofs} of {Contextuality}},\ }\href {https://doi.org/10.1103/PhysRevLett.113.200401} {\bibfield  {journal} {\bibinfo  {journal} {Phys. Rev. Lett.}\ }\textbf {\bibinfo {volume} {113}},\ \bibinfo {pages} {200401} (\bibinfo {year} {2014})}\BibitemShut {NoStop}%
\bibitem [{\citenamefont {Hallaji}\ \emph {et~al.}(2017)\citenamefont {Hallaji}, \citenamefont {Feizpour}, \citenamefont {Dmochowski}, \citenamefont {Sinclair},\ and\ \citenamefont {Steinberg}}]{Hallaji2017}%
  \BibitemOpen
  \bibfield  {author} {\bibinfo {author} {\bibfnamefont {M.}~\bibnamefont {Hallaji}}, \bibinfo {author} {\bibfnamefont {A.}~\bibnamefont {Feizpour}}, \bibinfo {author} {\bibfnamefont {G.}~\bibnamefont {Dmochowski}}, \bibinfo {author} {\bibfnamefont {J.}~\bibnamefont {Sinclair}},\ and\ \bibinfo {author} {\bibfnamefont {A.}~\bibnamefont {Steinberg}},\ }\bibfield  {title} {\bibinfo {title} {Weak-value amplification of the nonlinear effect of a single photon},\ }\href {https://doi.org/10.1038/nphys4040} {\bibfield  {journal} {\bibinfo  {journal} {Nat. Phys.}\ }\textbf {\bibinfo {volume} {13}},\ \bibinfo {pages} {540} (\bibinfo {year} {2017})}\BibitemShut {NoStop}%
\bibitem [{\citenamefont {Hosten}\ and\ \citenamefont {Kwiat}(2008)}]{Kwiat2008}%
  \BibitemOpen
  \bibfield  {author} {\bibinfo {author} {\bibfnamefont {O.}~\bibnamefont {Hosten}}\ and\ \bibinfo {author} {\bibfnamefont {P.}~\bibnamefont {Kwiat}},\ }\bibfield  {title} {\bibinfo {title} {Observation of the spin hall effect of light via weak measurements},\ }\href {https://doi.org/10.1126/science.1152697} {\bibfield  {journal} {\bibinfo  {journal} {Science}\ }\textbf {\bibinfo {volume} {319}},\ \bibinfo {pages} {787} (\bibinfo {year} {2008})}\BibitemShut {NoStop}%
\bibitem [{\citenamefont {Resch}\ \emph {et~al.}(2004)\citenamefont {Resch}, \citenamefont {Lundeen},\ and\ \citenamefont {Steinberg}}]{resch_experimental_2004}%
  \BibitemOpen
  \bibfield  {author} {\bibinfo {author} {\bibfnamefont {K.~J.}\ \bibnamefont {Resch}}, \bibinfo {author} {\bibfnamefont {J.~S.}\ \bibnamefont {Lundeen}},\ and\ \bibinfo {author} {\bibfnamefont {A.~M.}\ \bibnamefont {Steinberg}},\ }\bibfield  {title} {\bibinfo {title} {Experimental realization of the quantum box problem},\ }\href {https://doi.org/10.1016/j.physleta.2004.02.042} {\bibfield  {journal} {\bibinfo  {journal} {Phys. Lett. A}\ }\textbf {\bibinfo {volume} {324}},\ \bibinfo {pages} {125} (\bibinfo {year} {2004})}\BibitemShut {NoStop}%
\bibitem [{\citenamefont {{Von Neumann}}\ \emph {et~al.}(2018)\citenamefont {{Von Neumann}}, \citenamefont {Beyer},\ and\ \citenamefont {Wheeler}}]{VonNeumann_book}%
  \BibitemOpen
  \bibfield  {author} {\bibinfo {author} {\bibfnamefont {J.}~\bibnamefont {{Von Neumann}}}, \bibinfo {author} {\bibfnamefont {R.~T.}\ \bibnamefont {Beyer}},\ and\ \bibinfo {author} {\bibfnamefont {N.~A.}\ \bibnamefont {Wheeler}},\ }\href@noop {} {\emph {\bibinfo {title} {{Mathematical foundations of quantum mechanics }}}}\ (\bibinfo  {publisher} {Princeton University Press},\ \bibinfo {address} {Princeton, N.J.},\ \bibinfo {year} {2018})\BibitemShut {NoStop}%
\bibitem [{Note2()}]{Note2}%
  \BibitemOpen
  \bibinfo {note} {Of course, for any fixed pointer distribution width $\Delta \Psi (x)$ there will be an $n$ beyond which this approximation breaks down; however, the magnitude of the probability amplitudes $-\abs {t_1}^2r_2(r_1r_2)^{n-1}$ decrease exponentially with $n$, meaning one can always choose a width $\Delta \Psi (x)$ that is sufficiently large that the probability amplitudes for $n\gtrapprox \Delta \Psi (x)/x_{\protect \rm rt}$ are negligible}\BibitemShut {NoStop}%
\bibitem [{\citenamefont {Aharonov}\ and\ \citenamefont {Vaidman}(1991)}]{aharonov_complete_1991}%
  \BibitemOpen
  \bibfield  {author} {\bibinfo {author} {\bibfnamefont {Y.}~\bibnamefont {Aharonov}}\ and\ \bibinfo {author} {\bibfnamefont {L.}~\bibnamefont {Vaidman}},\ }\bibfield  {title} {\bibinfo {title} {Complete description of a quantum system at a given time},\ }\href {https://doi.org/10.1088/0305-4470/24/10/018} {\bibfield  {journal} {\bibinfo  {journal} {J. Phys. A: Math.}\ }\textbf {\bibinfo {volume} {24}},\ \bibinfo {pages} {2315} (\bibinfo {year} {1991})}\BibitemShut {NoStop}%
\bibitem [{\citenamefont {Vaidman}(1996)}]{vaidman_weak-measurement_1996}%
  \BibitemOpen
  \bibfield  {author} {\bibinfo {author} {\bibfnamefont {L.}~\bibnamefont {Vaidman}},\ }\bibfield  {title} {\bibinfo {title} {Weak-measurement elements of reality},\ }\href {https://doi.org/10.1007/BF02148832} {\bibfield  {journal} {\bibinfo  {journal} {Found. Phys.}\ }\textbf {\bibinfo {volume} {26}},\ \bibinfo {pages} {895} (\bibinfo {year} {1996})}\BibitemShut {NoStop}%
\bibitem [{\citenamefont {Banerjee}\ \emph {et~al.}(2022)\citenamefont {Banerjee}, \citenamefont {Solomons}, \citenamefont {Black}, \citenamefont {Marcucci}, \citenamefont {Eger}, \citenamefont {Davidson}, \citenamefont {Firstenberg},\ and\ \citenamefont {Boyd}}]{banerjee_anomalous_2022}%
  \BibitemOpen
  \bibfield  {author} {\bibinfo {author} {\bibfnamefont {C.}~\bibnamefont {Banerjee}}, \bibinfo {author} {\bibfnamefont {Y.}~\bibnamefont {Solomons}}, \bibinfo {author} {\bibfnamefont {A.~N.}\ \bibnamefont {Black}}, \bibinfo {author} {\bibfnamefont {G.}~\bibnamefont {Marcucci}}, \bibinfo {author} {\bibfnamefont {D.}~\bibnamefont {Eger}}, \bibinfo {author} {\bibfnamefont {N.}~\bibnamefont {Davidson}}, \bibinfo {author} {\bibfnamefont {O.}~\bibnamefont {Firstenberg}},\ and\ \bibinfo {author} {\bibfnamefont {R.~W.}\ \bibnamefont {Boyd}},\ }\bibfield  {title} {\bibinfo {title} {Anomalous optical drag},\ }\href {https://doi.org/10.1103/PhysRevResearch.4.033124} {\bibfield  {journal} {\bibinfo  {journal} {Phys. Rev. Res.}\ }\textbf {\bibinfo {volume} {4}},\ \bibinfo {pages} {033124} (\bibinfo {year} {2022})}\BibitemShut {NoStop}%
\bibitem [{\citenamefont {Safari}\ \emph {et~al.}(2016)\citenamefont {Safari}, \citenamefont {De~Leon}, \citenamefont {Mirhosseini}, \citenamefont {Magaña-Loaiza},\ and\ \citenamefont {Boyd}}]{safari_light-drag_2016}%
  \BibitemOpen
  \bibfield  {author} {\bibinfo {author} {\bibfnamefont {A.}~\bibnamefont {Safari}}, \bibinfo {author} {\bibfnamefont {I.}~\bibnamefont {De~Leon}}, \bibinfo {author} {\bibfnamefont {M.}~\bibnamefont {Mirhosseini}}, \bibinfo {author} {\bibfnamefont {O.~S.}\ \bibnamefont {Magaña-Loaiza}},\ and\ \bibinfo {author} {\bibfnamefont {R.~W.}\ \bibnamefont {Boyd}},\ }\bibfield  {title} {\bibinfo {title} {Light-{Drag} {Enhancement} by a {Highly} {Dispersive} {Rubidium} {Vapor}},\ }\href {https://doi.org/10.1103/PhysRevLett.116.013601} {\bibfield  {journal} {\bibinfo  {journal} {Phys. Rev. Lett.}\ }\textbf {\bibinfo {volume} {116}},\ \bibinfo {pages} {013601} (\bibinfo {year} {2016})}\BibitemShut {NoStop}%
\bibitem [{\citenamefont {Hau}\ \emph {et~al.}(1999)\citenamefont {Hau}, \citenamefont {Harris}, \citenamefont {Dutton},\ and\ \citenamefont {Behroozi}}]{hau_light_1999}%
  \BibitemOpen
  \bibfield  {author} {\bibinfo {author} {\bibfnamefont {L.~V.}\ \bibnamefont {Hau}}, \bibinfo {author} {\bibfnamefont {S.~E.}\ \bibnamefont {Harris}}, \bibinfo {author} {\bibfnamefont {Z.}~\bibnamefont {Dutton}},\ and\ \bibinfo {author} {\bibfnamefont {C.~H.}\ \bibnamefont {Behroozi}},\ }\bibfield  {title} {\bibinfo {title} {Light speed reduction to 17 metres per second in an ultracold atomic gas},\ }\href {https://doi.org/10.1038/17561} {\bibfield  {journal} {\bibinfo  {journal} {Nature}\ }\textbf {\bibinfo {volume} {397}},\ \bibinfo {pages} {594} (\bibinfo {year} {1999})}\BibitemShut {NoStop}%
\bibitem [{\citenamefont {Angulo}\ \emph {et~al.}()\citenamefont {Angulo}, \citenamefont {Thompson}, \citenamefont {Nixon}, \citenamefont {Jiao}, \citenamefont {Wiseman},\ and\ \citenamefont {Steinberg}}]{InPreparation}%
  \BibitemOpen
  \bibfield  {author} {\bibinfo {author} {\bibfnamefont {D.}~\bibnamefont {Angulo}}, \bibinfo {author} {\bibfnamefont {K.}~\bibnamefont {Thompson}}, \bibinfo {author} {\bibfnamefont {V.-M.}\ \bibnamefont {Nixon}}, \bibinfo {author} {\bibfnamefont {X.}~\bibnamefont {Jiao}}, \bibinfo {author} {\bibfnamefont {H.~M.}\ \bibnamefont {Wiseman}},\ and\ \bibinfo {author} {\bibfnamefont {A.~M.}\ \bibnamefont {Steinberg}},\ }\href@noop {} {\bibinfo  {journal} {In preparation}\ }\BibitemShut {NoStop}%
\end{thebibliography}%

\end{document}